\documentclass{article}
\usepackage{fullpage}
\usepackage[usenames]{color}
\usepackage[colorlinks = true]{hyperref}
\usepackage[english]{babel}
\usepackage{bm,bbm,amsmath,amssymb,amsthm,graphicx}
\usepackage{url}
\usepackage[caption=false,font=normalsize,labelfont=sf,textfont=sf]{subfig}
\usepackage[title]{appendix}
\usepackage{tikz}
\usetikzlibrary{calc,patterns,decorations.pathmorphing,decorations.markings}

\numberwithin{equation}{section}

\newcommand{\change}[1]{\ensuremath{\operatorname{#1}}}
\newcommand{\MAT}{\left[ \begin{array}}  
\newcommand{\mat}{\end{array} \right]}

\newtheorem{Corollary}{Corollary}[section]
\newtheorem{Lemma}{Lemma}[section]
\newtheorem{Theorem}{Theorem}[section]

\newtheorem{Q}{Question}[section]
\newtheorem{Proposition}{Proposition}
\newtheorem{Remark}{Remark}[section]

\def \a {\bm{a}}
\def \A {\mathbf{A}}
\def \AA {\mathcal{A}}
\def \Ah {\hat{A}}
\def \b {\bm{b}}

\def \BB {\mathcal{B}}
\def \C {\mathbf{C}}
\def \CC {\mathcal{C}}
\def \CCC {\mathbb{C}}
\def \d {\bm{d}}
\def \D {\mathbf{D}}
\def \e {\bm{e}}
\def \E {\mathbf{E}}
\def \EEE{\mathbb{E}}
\def \f {\bm{f}}
\def \fh {\hat{f}}
\def \FF {\mathcal{F}}
\def \g {\bm{g}}
\def \G {\mathbf{G}}

\def \H {\mathbf{H}}
\def \I {\mathbf{I}}

\def \K {\mathbf{K}}
\def \KK {\mathcal{K}}
\def \M {\mathbf{M}}

\def \NN {\mathcal{N}}

\def \PP {\mathcal{P}}
\def \PPP {\mathbb{P}}

\def \Q {\mathbf{Q}}
\def \QQ {\mathcal{Q}}
\def \r {\bm{r}}

\def \RRR {\mathbb{R}}
\def \S {\mathbf{S}}
\def \SSS {\mathbb{S}}

\def \TT {\mathcal{T}}
\def \u {\bm{u}}

\def \uh {\hat{\bm{u}}}

\def \V {\mathbf{V}}
\def \Vh {\hat{\mathbf{V}}}

\def \W {\mathbf{W}}
\def \WW {\mathcal{W}}

\def \Wt {\tilde{W}}
\def \x {\bm{x}}
\def \X {\mathbf{X}}
\def \Xh {\hat{\mathbf{X}}}
\def \y {\bm{y}}
\def \Y {\mathbf{Y}}
\def \z {\bm{z}}
\def \Z {\mathbf{Z}}
\def \zt {\tilde{\bm{z}}}
\def \Zh {\hat{\mathbf{Z}}}
\def \Zt {\tilde{\mathbf{Z}}}

\def \balpha {\boldsymbol{\alpha}}
\def \bbeta {\boldsymbol{\beta}}
\def \bpsi {\boldsymbol{\psi}}
\def \bpsih {\hat{\boldsymbol{\psi}}}

\def \bPsih {\hat{\boldsymbol{\Psi}}}
\def \bPhi {\boldsymbol{\Phi}}

\def \bmu {\boldsymbol{\mu}}
\def \bnu {\boldsymbol{\nu}}
\def \bxi {\boldsymbol{\xi}}

\def \zero {\mathbf{0}}
\def \one{\mathbbm{1}}

\begin{document}

 
\title{Atomic Norm Minimization for Modal Analysis\\ from Random and Compressed Samples}

\author{Shuang Li, Dehui Yang, Gongguo Tang, and Michael B. Wakin\thanks{Department of Electrical Engineering, Colorado School of Mines. Email: \{shuangli,dyang,gtang,mwakin\}@mines.edu}}

\maketitle

\begin{abstract}
Modal analysis is the process of estimating a system's modal parameters such as its natural frequencies and mode shapes. One application of modal analysis is in structural health monitoring (SHM), where a network of sensors may be used to collect vibration data from a physical structure such as a building or bridge. There is a growing interest in developing automated techniques for SHM based on data collected in a wireless sensor network. In order to conserve power and extend battery life, however, it is desirable to minimize the amount of data that must be collected and transmitted in such a sensor network. In this paper, we highlight the fact that modal analysis can be formulated as an atomic norm minimization (ANM) problem, which can be solved efficiently and in some cases recover perfectly a structure's mode shapes and frequencies. We survey a broad class of sampling and compression strategies that one might consider in a physical sensor network, and we provide bounds on the sample complexity of these compressive schemes in order to recover a structure's mode shapes and frequencies via ANM. A main contribution of our paper is to establish a bound on the sample complexity of modal analysis with random temporal compression, and in this scenario we prove that the required number of samples per sensor can actually decrease as the number of sensors increases. We also extend an atomic norm denoising problem to the multiple measurement vector (MMV) setting in the case of uniform sampling.
\end{abstract}

\section{Introduction}
\label{intr}

Modal analysis is the process of estimating a system's modal parameters such as its natural frequencies, mode shapes, and damping factors. One application of modal analysis is in structural health monitoring (SHM), where a network of sensors may be used to collect vibration data from a physical structure such as a building or bridge. The vibration characteristics of a structure are captured in its modal parameters, which can be estimated from the recorded displacement data. Changes in these parameters over time may be indicative of damage to the structure. Modal analysis has been widely used in civil structures \cite{cunha2006experimental}, space structures \cite{kammer1991sensor}, acoustical instruments \cite{marshall1985modal}, and so on.


Due to the considerable time and expense required to perform manual inspections of physical structures, and the difficulty of repeating these inspections frequently, there is a growing interest in developing automated techniques for SHM based on data collected in a wireless sensor network. For example, one could envision a collection of battery-operated wireless sensors deployed across a structure that record vibrational displacements over time and then transmit this information to a central node for analysis. In order to conserve power and extend battery life, however, it is desirable to minimize the amount of data that must be collected and transmitted in such a sensor network  \cite{o2014compressed}.

In this paper, we highlight the fact that modal analysis can be formulated as an atomic norm minimization (ANM) problem, which can be solved efficiently and in some cases recover perfectly a structure's mode shapes and frequencies. We survey several possible protocols for data collection, compression, and transmission, and we review the sampling requirements in each case. ANM generalizes the widely used $\ell_1$-minimization framework for finding sparse solutions to underdetermined linear inverse problems~\cite{Chand12}. It has recently been shown to be an efficient and powerful way for exactly recovering unobserved time-domain samples and identifying unknown frequencies in signals having sparse frequency spectra~\cite{Tang13,Li15,yang2014continuous}, in particular when the unknown frequencies are continuous-valued and do not belong to a discrete grid. Sampling guarantees have been established both in the single measurement vector (SMV) scenario~\cite{Tang13,heckel2016generalized} and in the multiple measurement vector (MMV) scenario under a joint sparse model~\cite{Yang14,Li15}; these results characterize the number of uniform or random time-domain samples to achieve exact frequency localization as a function of the minimum separation between the unknown frequencies. In this sense, by achieving exact recovery, ANM can completely avoid the effects of basis mismatch~\cite{Duarte11,Chi11,Herman10} which can plague conventional grid-based compressive sensing techniques.

To provide context in this paper, we survey a broad class of sampling and compression strategies that one might consider in a physical sensor network, and we provide bounds on the sample complexity of these compressive schemes in order to recover a structure's mode shapes and frequencies via ANM. In total, we consider five measurement schemes:
\begin{itemize}
\item {\bf uniform sampling}, where vibration signals at each sensor are sampled synchronously at or above the Nyquist rate and transmitted without compression to a central node,
\item {\bf synchronous random sampling}, where vibration signals at each sensor are sampled randomly in time (as a subset of the Nyquist grid), but at the same time instants at each sensor,
\item {\bf asynchronous random sampling}, where vibration signals at each sensor are sampled randomly in time (as a subset of the Nyquist grid), at different time instants at each sensor,
\item {\bf random temporal compression}, where Nyquist-rate samples are compressed at each sensor via random matrix multiplication, before transmission to a central node, and
\item {\bf random spatial compression}, where Nyquist-rate samples are compressed en route to the central node. 
\end{itemize} 
Note that all sensors share the same Nyquist grid in each of these measurement schemes.

Modal analysis is a particular instance of the joint sparse frequency estimation problem, which has also been commonly studied in the context of direction-of-arrival (DOA) estimation \cite{yang2016sparse}. Some conventional methods for joint sparse frequency estimation, such as MUSIC \cite{schmidt1982signal}, can identify frequencies using a sample covariance matrix as long as a sufficient number of snapshots is given. However, as noted in~\cite{Yang14}, these methods usually assume that the source signals are spatially uncorrelated and their performance would decrease with source correlations. ANM does not have this limitation. Moreover, joint sparse frequency estimation techniques such as MUSIC do not naturally accommodate the sort of randomized sampling and compression protocols that we consider in this paper.

In this work, we explain how the results from~\cite{Li15,Yang14} can be interpreted in the context of exactly recovering a structure's mode shapes and frequencies from uniform samples, synchronous random samples, and asynchronous random samples. These random sampling results have an unfortunate scaling, however, in that the number of samples per sensor actually increases as the number of sensors increases; intuition and simulations suggest that the opposite should be true. A main contribution of our paper, then, is to establish a bound on the sample complexity of modal analysis with random temporal compression, and in this scenario we prove that the  required number of samples per sensor can actually {\em decrease} as the number of sensors increases. A similar phenomenon---that the estimation accuracy increases as the number of snapshots increases---occurs in covariance fitting methods for DOA estimation~\cite{yang2016sparse}. We also explain how our previous work~\cite{yang2016super} on super-resolution of complex exponentials from modulations with unknown waveforms can be used to establish a bound on the sample complexity of modal analysis with random spatial compression. For the noisy case, we extend the SMV atomic norm denoising problem in \cite{bhaskar2013atomic, tang2015near} to an MMV atomic norm denoising problem in the case of uniform sampling. We derive theoretical guarantees for this MMV atomic norm denoising problem. Although we focus on modal analysis to put our analysis and simulations in a specific context, our results can apply to joint sparse frequency estimation more generally.



The remainder of this paper is organized as follows. In Section~\ref{prel}, we provide some background on modal analysis and on the atomic norm. In Section~\ref{main}, we formulate our problem, and we characterize the ability of ANM to exactly recover a structure's mode shapes and frequencies under each of the above five measurement schemes. We also provide a bound on the performance of ANM in the case of noisy, uniform samples. In Section~\ref{simu}, we present simulation results to illustrate some of the essential trends in the theoretical results and to demonstrate the favorable performance of ANM. We prove our main theorems in Section~\ref{sec:mainproofs}, and we make some concluding remarks in Section~\ref{conc}. The Appendix provides supplementary theoretical results.

\section{Background}
\label{prel}

\subsection{Modal analysis}
\label{sec:modal}

For an $N$ degree-of-freedom linear time-invariant system~\cite{Ewins00}, the second-order equations of motion which represent the dynamic behavior of the system can be formulated as
\begin{align}
\M\ddot{\x}(t)+\C\dot{\x}(t)+\K\x(t)=\f(t) \label{EOM}
\end{align}
where $\M$, $\C$ and $\K$ denote the $N \times N$ mass, damping, and stiffness matrices, respectively. In this equation, $\x(t)\in \RRR^N$ represents a length-$N$ vector of displacement values at time $t$, and $\f(t)\in \RRR^N$ represents the excitation force at the $N$ nodes. We will assume that each of the $N$ displacement values is associated with a wireless sensor node that can sample, record, and transmit that displacement value; as an example, one could equip each of the six boxcars shown in Fig.~\ref{fig:boxcar} with a displacement sensor.

\begin{figure}
\hspace*{0.6in}
\begin{tikzpicture}
\tikzstyle{spring}=[thick,decorate,decoration={zigzag,pre length=0.3cm,post length=0.3cm,segment length=6}]
\tikzstyle{damper}=[thick,decoration={markings,
  mark connection node=dmp,
  mark=at position 0.5 with
  {
    \node (dmp) [thick,inner sep=0pt,transform shape,rotate=-90,minimum width=15pt,minimum height=3pt,draw=none] {};
    \draw [thick] ($(dmp.north east)+(2pt,0)$) -- (dmp.south east) -- (dmp.south west) -- ($(dmp.north west)+(2pt,0)$);
    \draw [thick] ($(dmp.north)+(0,-5pt)$) -- ($(dmp.north)+(0,5pt)$);
  }
}, decorate]
\tikzstyle{ground}=[fill,pattern=north east lines,draw=none,minimum width=0.75cm,minimum height=0.3cm]

\node (M) [draw,outer sep=0pt,thick,minimum width=1cm, minimum height=1cm] {$m_1$};
\node (M2) [draw,outer sep=0pt,thick,minimum width=1cm, minimum height=1cm] at (2,0) {$m_2$};
\node (M3) [draw,outer sep=0pt,thick,minimum width=1cm, minimum height=1cm] at (4,0) {$m_3$};
\node (M4) [draw,outer sep=0pt,thick,minimum width=1cm, minimum height=1cm] at (6,0) {$m_4$};
\node (M5) [draw,outer sep=0pt,thick,minimum width=1cm, minimum height=1cm] at (8,0) {$m_5$};
\node (M6) [draw,outer sep=0pt,thick,minimum width=1cm, minimum height=1cm] at (10,0) {$m_6$};

\node (ground) [ground,anchor=north,xshift=5.0cm,yshift=-0.25cm,minimum width=12.9cm] at (M.south) {};

\draw (ground.north east) -- (ground.north west);

\draw (ground.north east) -- (ground.north west);
\draw [thick] (M.south west) ++ (0.2cm,-0.125cm) circle (0.125cm)  (M.south east) ++ (-0.2cm,-0.125cm) circle (0.125cm);
\draw [thick] (M2.south west) ++ (0.2cm,-0.125cm) circle (0.125cm)  (M2.south east) ++ (-0.2cm,-0.125cm) circle (0.125cm);
\draw [thick] (M3.south west) ++ (0.2cm,-0.125cm) circle (0.125cm)  (M3.south east) ++ (-0.2cm,-0.125cm) circle (0.125cm);
\draw [thick] (M4.south west) ++ (0.2cm,-0.125cm) circle (0.125cm)  (M4.south east) ++ (-0.2cm,-0.125cm) circle (0.125cm);
\draw [thick] (M5.south west) ++ (0.2cm,-0.125cm) circle (0.125cm)  (M5.south east) ++ (-0.2cm,-0.125cm) circle (0.125cm);
\draw [thick] (M6.south west) ++ (0.2cm,-0.125cm) circle (0.125cm)  (M6.south east) ++ (-0.2cm,-0.125cm) circle (0.125cm);

\node (wall) [ground, rotate=-90, minimum width=1.8cm,xshift=0.15cm,yshift=-1.6cm] {};
\draw (wall.north east) -- (wall.north west);

\node (wall1) [ground, rotate=-270, minimum width=1.8cm,xshift=-0.15cm,yshift=-11.6cm] {};
\draw (wall1.north west) -- (wall1.north east);

\draw [spring] (M.180) -- ($(wall.north east)!(wall.225)!(wall.north east)$);
\draw [spring] (M2.180) -- ($(M.north east)!(M.180)!(M.south east)$);
\draw [spring] (M3.180) -- ($(M2.north east)!(M2.180)!(M2.south east)$);
\draw [spring] (M4.180) -- ($(M3.north east)!(M3.180)!(M3.south east)$);
\draw [spring] (M5.180) -- ($(M4.north east)!(M4.180)!(M4.south east)$);
\draw [spring] (M6.180) -- ($(M5.north east)!(M5.180)!(M5.south east)$);
\draw [spring] ($(wall1.north west)!(wall1.45)!(wall1.north west)$) -- ($(M6.north east)!(M6.180)!(M6.south east)$);

\node at (-1,.3) {$k_1$};
\node at (1,.3) {$k_2$};
\node at (3,.3) {$k_3$};
\node at (5,.3) {$k_4$};
\node at (7,.3) {$k_5$};
\node at (9,.3) {$k_6$};
\node at (11,.3) {$k_7$};

\end{tikzpicture}

\caption{\small\sl Undamped 6-degree-of-freedom boxcar system with masses $m_1=1,~m_2=2,~m_3=3,~m_4=4,~m_5=5,~m_6=6$ kg and stiffness values $k_1=k_7=500$, $k_2=k_6=150$, $k_3=k_5=100$, $k_4=50$ N/m. See Section~\ref{uniform} for simulations related to this system.}\label{fig:boxcar}
\end{figure}
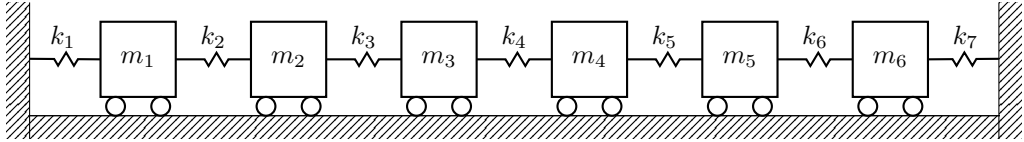

As a simplification and as in~\cite{park2014modal}, we consider systems in free vibration, where the forcing input $\f(t)=\bm{0}$ and the system vibrates freely in response to some set of initial conditions.\footnote{One can also consider modal analysis with forced vibration~\cite{mead1969forced,yamakoshi1990ultrasonic,salawu1995bridge,cho2015forced}, when external forces such as an earthquake, wind, or vehicle loadings are applied.} In this case, the general solution to~\eqref{EOM} takes the form~\cite{Rade}
\begin{align*}
\x(t) = \sum_{k=1}^K \bpsi_k q_k(t),
\end{align*}
where $\bpsi_1, \bpsi_2, \dots, \bpsi_K \in \RRR^N$ are a collection of mode shape vectors whose span characterizes the set of possible displacement profiles of the structure, and $q_1(t), q_2(t), \dots, q_K(t)$ are modal responses in the form of monotone exponentially-decaying sinusoids:
\begin{align*}
q_k(t)=u_k e^{-\zeta_k 2\pi f_k t} \cos(2\pi f_{d_k} t + \theta_k).
\end{align*}
Here, $u_k$ and $\theta_k$ are determined by the initial conditions, and the number of nonzero amplitudes $K \le N$ corresponds to the number of active modes. The parameters $\zeta_k$, $f_k$, and $f_{d_k}=\sqrt{1-\zeta_k^2}f_k$ denote the $k$th damping ratio, the natural frequency, and the damped frequency, respectively, and these parameters along with the mode shapes are intrinsic properties of the system determined by the mass, damping, and stiffness matrices. Modal analysis refers to the identification of these parameters---particularly the mode shapes, damping ratios, and frequencies---from observations of the displacement vector $\x(t)$.

%
%

In recent years, blind source separation (BSS) based methods have become very popular in modal analysis. The authors in \cite{Yang13} and \cite{yu2014estimation} propose a new modal identification algorithm based on sparse component analysis (SCA) to deal with even the underdetermined case where sensors may be highly limited compared to the number of active modes. In \cite{Sadhu13}, a novel decentralized modal identification method termed PARAllel FACtor (PARAFAC) based sparse BSS (PSBSS) method is proposed. Independent component analysis (ICA) is a powerful method to solve the BSS problem. Yang et al.\ present an ICA based method to identify the modal parameters of lightly and highly damped systems, even in cases with heavy noise and nonstationarity \cite{yang2012time}.
Despite the favorable empirical performance of these methods, few of them have been supported by theoretical analysis. In this paper, inspired by promising recent work in line spectrum estimation, we provide theoretical support for ANM as a powerful technique for modal analysis.

Meanwhile, over the past decade, the development of compressive sensing (CS) has highlighted the possibility of capturing essential signal information with sampling rates much lower than the Nyquist rate~\cite{donoho2006compressed,cande2008introduction,candes2006robust}. Exploring the possibility of using CS in modal analysis, Park et al.\ provide a theoretical analysis of a singular value decomposition (SVD) based technique for estimating a structure's mode shapes in free vibration without damping~\cite{park2014modal}. The work in \cite{park2014modal} builds upon a previous observation~\cite{feeny1998physical} that with a sufficiently large number of samples, the SVD technique can identify the mode shapes of systems with uniformly distributed mass (which leads to mutually orthogonal mode shapes) and light or no damping; as such, the results in~\cite{park2014modal} are limited to the assumption that the mode shapes are mutually orthogonal, which is not satisfied for general systems. In a more recent work,  Yang et al.\ propose a BSS based method that can identify non-orthogonal mode shapes from video measurements~\cite{yang2017blind}. The same group also develops another BSS based method to identify modal parameters from uniformly sub-Nyquist (temporally-aliased) video measurements~\cite{yang2017blind1}. Another recent work~\cite{yang2015output} also presents a new method based on a combination of CS and complexity pursuit (CP) to solve the modal identification problem.

In a very recent paper \cite{heckel2016generalized}, Heckel and Soltanolkotabi consider solving a generalized line spectrum estimation problem with convex optimization. In particular, they recover a data vector from its compressed measurements using an ANM formulation that has been mentioned in \cite{Chand12}. Several different classes of random sensing matrices are considered for obtaining the measurements. The result for Gaussian random matrices in~\cite{heckel2016generalized} can be viewed as a special case of our modal analysis result for random temporal compression (Theorem~\ref{THMCS}), where we have multiple measurement vectors and use a different sensing matrix for each measurement vector. Our analysis is inspired by their proof.

Finally, Lu et al.~\cite{lu2015distributed} propose a concatenated ANM approach for joint recovery of frequency sparse signals under a certain joint sparsity model. This problem differs from ours in the signal model, and the paper~\cite{lu2015distributed} does not establish theoretical bounds on the sample complexity.

\subsection{The atomic norm}

Frequency estimation from a mixture of complex sinusoids is a classical problem in signal processing. As mentioned in the introduction, ANM has recently been considered as a technique for solving this problem, both in the SMV ($N=1$) and MMV ($N > 1$) scenarios. Under the right conditions, ANM can achieve exact frequency localization, avoiding the effects of basis mismatch which can plague grid-based techniques. In Section~\ref{main}, we survey various formulations of the ANM problem that can be used with the various compressive measurement protocols for modal analysis. All of these formulations rely on the same core definition of an atomic norm, which is established in~\cite{Li15,Yang14} and repeated here.

Suppose each column of an $M \times N$ data matrix $\X=[\bm{x}_1,~\bm{x}_2,~\cdots ,~\bm{x}_N]$ is a spectrally sparse signal with $K$ distinct frequency components and denoted as
\begin{align}
\bm{x}_i=\sum_{k=1}^K c_{k,i} \a(f_k).
\nonumber
\end{align}
Here, the vector
\begin{equation}
\a(f)= [e^{j2\pi f 0} ,  e^{j2\pi f 1},  \cdots,  e^{j2\pi f (M-1)}]^\top
\label{eq:littlea}
\end{equation}
corresponds to a collection of $M$ samples of a complex exponential signal with frequency $f\in [0,1)$. Note that we use ``$\top$" and ``$*$" to denote transpose and conjugate transpose, respectively.

As shown in \cite{Li15} and  \cite{Yang14}, one can define an atomic set to represent such a data matrix $\X$ with each {\em atom} defined as
\begin{align}
\A(f,\bm{b})=\a(f)\bm{b}^*,
\label{atom}
\end{align}
where $f\in [0,1)$, $\bm{b}\in \mathbb{C}^N$ with $\|\bm{b}\|_2=1$. The corresponding {\em atomic set} can be defined as
\begin{align}
\AA=\{\A(f,\b):~f\in [0,1),~\|\bm{b}\|_2=1 \}.
\nonumber
\end{align}
The {\em atomic norm} of $\X$ is then defined as
\begin{align}
\|\X\|_{\AA}&=\inf\left\{t>0: \X\in t\change{~conv}(\AA)\right\}\nonumber \\
&=\inf\left\{ \sum_k c_k: \X = \sum_k c_k \A(f_k,\b_k),~c_k\geq 0\right\},\nonumber
\end{align}
where $\change{conv}(\AA)$ is the convex hull of $\AA$. This atomic norm is equivalent to the solution of the following semidefinite program (SDP):
\begin{align}
\|\X\|_{\AA}=\inf_{\bm{u}\in\mathbb{C}^M,~\bm{V}\in\mathbb{C}^{N\times N}}\left\{ \frac{1}{2M}\change{Tr}(\mathcal{T}(\bm{u}))+\frac{1}{2}\change{Tr}(\V): ~ \left[  \begin{array}{cc}
  \mathcal{T}(\bm{u})&\X \nonumber\\
  {\X}^*&\V
\end{array}   \right]\succeq 0\right\},\nonumber
\end{align}
where $\mathcal{T}(\bm{u})$ is the Hermitian Toeplitz matrix with the vector $\bm{u}$ as its first column. $\change{Tr}(\cdot)$ denotes the trace of a square matrix. The proof of this SDP form can be found in \cite{Tang13,Li15}.

The {\em dual norm} of $\| \cdot \|_{\AA}$ is defined as
\begin{equation}
\begin{aligned}
\|\Q\|_{\AA}^*&=\sup_{\|\X\|_{\AA}\leq 1} \langle \Q,\X  \rangle_{\RRR}\\
&=\sup_{f\in [0,1), \|\b\|_2=1} \langle \Q,\a(f)\b^* \rangle_{\RRR}\\
&=\sup_{f\in [0,1), \|\b\|_2=1} \langle \b,\Q^*\a(f) \rangle_{\RRR}\\
&=\sup_{f\in [0,1)} \|\Q^*\a(f)\|_2\\
&=\sup_{f\in [0,1)} \|\QQ(f)\|_2, \label{dualnorm}
\end{aligned}
\end{equation}
where $\QQ(f) = \Q^*\a(f)$ is known as the dual polynomial. Above, $\langle \Q,\X  \rangle_{\RRR}=\change{Re}(\langle \Q,\X  \rangle)$ corresponds to the real inner product between two matrices.

\section{Main Result}
\label{main}

\subsection{Preliminaries}
\label{prob}

As in~\cite{park2014modal}, we assume the structure has no damping ($\C = \zero$) and that the real-valued displacement\footnote{Although we refer to displacement data in this work, the ANM method can also be applied to acceleration data. In particular, the ground truth acceleration signal vector will have the same form as~\eqref{signal} but with a different set of amplitudes $A_k$.} signal has been converted to its complex analytic form, which can be accomplished by using the Hilbert transform in practice \cite{boashash2015time, smith2011spectral}. In this case, the ground truth displacement signal vector can be written as
\begin{align}
\x^{\star}(t)=\sum_{k=1}^K A_k \bpsi_k^{\star} e^{j2\pi F_k t},\label{signal}
\end{align}
which is a superposition of $K$ complex sinusoids. Here, $A_k$, $F_k$, and $\bpsi_k^{\star}=[\psi_{1,k}^{\star},\psi_{2,k}^{\star},\cdots,\psi_{N,k}^{\star}]^{\top}\in\CCC^N$ are the complex amplitudes, frequencies, and mode shapes, respectively. We assume without loss of generality that the mode shapes are normalized\footnote{Equation~\eqref{signal} will hold with any choice of normalization for the mode shapes, as any rescaling of $\bpsi_k^{\star}$ can be simply absorbed into $A_k$. Any such renormalization can be applied after the Euclidean-normalized mode shapes are recovered and thus does not affect our results. A certain ``mass normalization'' is desired in some applications. Employing mass normalization requires either knowledge of the mass matrix or extra experiments~\cite{lopez2005some}.} such that $\| \bpsi_k^{\star} \|_2 = 1$. As in Section~\ref{sec:modal}, $N$ denotes the number of sensors. Thus, $\sum_{k=1}^K A_k \psi_{n,k}^{\star} e^{j2\pi F_k t}$ is the displacement signal at the $n$th sensor.

Because we assume that $\x^{\star}(t)$ is an analytic signal, all $F_k \ge 0$.\footnote{It would not be difficult to extend our analysis to the case where the frequencies $F_k$ in~\eqref{signal} can be negative or positive: one would set $F_c=\max_{1\leq k \leq K} |F_k|$ and $T_s < \frac{1}{2 F_c}$, which yields discrete frequencies $f_k \in [-1/2,1/2)$, which is equivalent to the interval $[0,1)$.} Define $F_c=\max_{1\leq k \leq K} F_k$. If one were to take regularly spaced Nyquist samples of $\x^{\star}(t)$ at times
\begin{align}
T=\{t_1,t_2,\ldots, t_M\}=\{ 0,T_s,\ldots,(M-1)T_s \}, \label{eq:Ttimes}
\end{align}
where $T_s < \frac{1}{F_c}$, and then stack the signal from the $n$th sensor as the $n$th column of a data matrix $\X^{\star}$,
one would have
\begin{equation}
\begin{aligned}
\X^{\star}&=\sum_{k=1}^K A_k
\left[
\begin{array}{cccc}
   \psi_{1,k}^{\star} e^{j2\pi F_k t_{1}} & \psi_{2,k}^{\star} e^{j2\pi F_k t_{1}}  &  \cdots & \psi_{N,k}^{\star} e^{j2\pi F_k t_{1}} \\
   \psi_{1,k}^{\star} e^{j2\pi F_k t_{2}} & \psi_{2,k}^{\star} e^{j2\pi F_k t_{2}}  &  \cdots & \psi_{N,k}^{\star} e^{j2\pi F_k t_{2}} \\
  \vdots &\vdots &\ddots & \vdots\\
  \psi_{1,k}^{\star} e^{j2\pi F_k t_{M}} & \psi_{2,k}^{\star} e^{j2\pi F_k t_{M}}  &  \cdots & \psi_{N,k}^{\star} e^{j2\pi F_k t_{M}}
\end{array}
\right]\\
&=\sum_{k=1}^K |A_k|
\left[
\begin{array}{cccc}
   \psi_{1,k} e^{j2\pi f_k 0} & \psi_{2,k} e^{j2\pi f_k 0}  &  \cdots & \psi_{N,k} e^{j2\pi f_k 0} \\
   \psi_{1,k} e^{j2\pi f_k 1} & \psi_{2,k} e^{j2\pi f_k 1}  &  \cdots & \psi_{N,k} e^{j2\pi f_k 1} \\
  \vdots &\vdots &\ddots & \vdots\\
  \psi_{1,k} e^{j2\pi f_k (M-1)} & \psi_{2,k} e^{j2\pi f_k (M-1)}  &  \cdots & \psi_{N,k} e^{j2\pi f_k (M-1)}
\end{array}
\right] \\
&=\sum_{k=1}^K |A_k| \a(f_k) \bpsi_k^{\top}=\sum_{k=1}^K |A_k| \A(f_k,\b_k),
\label{data_uniform}
\end{aligned}
\end{equation}
where $\psi_{n,k}=\psi_{n,k}^{\star} e^{j\phi_{A_k}}$ with $\phi_{A_k}$ being the phase of the complex amplitude $A_k$, $f_k=F_kT_s \in [0,1)$ are the discrete frequencies, $\a(f_k)$ are sampled complex exponentials as defined in~\eqref{eq:littlea}, and $\A(f_k,\b_k)$ are the atoms defined in~\eqref{atom}. The ability to express $\X^{\star}$ as a linear combination of $K$ atoms from the atomic set $\AA$ inspires the use of ANM to recover $\X^{\star}$ from partial information, as we discuss below.

\subsection{Modal analysis for noiseless signals}
\label{sec:mans}

In this section, we survey a broad class of sampling and compression strategies that one might consider in a physical sensor network, and we provide bounds on the sample complexity of these compressive schemes in order to recover a structure's mode shapes and frequencies via ANM.

The performance of ANM in these various scenarios will depend on the minimum separation of the discrete frequencies $f_k$'s, which is defined to be
\begin{align*}
\Delta_f=\min_{k\neq j}| f_k-f_j|,
\end{align*}
where $|f_k-f_j|$ is understood as the wrap-around distance on the unit circle.

We also note that, because each complex mode shape vector $\bpsi_k^{\star}$ is multiplied by a complex amplitude $A_k$ in our model~\eqref{signal}, recovery of these mode shapes and amplitudes is possible only up to a phase ambiguity. We denote our estimated mode shapes as $\hat{\bpsi}_k$ and measure recovery performance using the absolute inner product $|\langle\bpsi_k^\star,\hat{\bpsi}_k\rangle|$. Because $\bpsi_k^{\star}$ and $\hat{\bpsi}_k$ are both normalized, achieving $|\langle\bpsi_k^\star,\hat{\bpsi}_k\rangle| = 1$ corresponds to exact recovery of the mode shape up to the unknown phase term.

\subsubsection{Uniform sampling}
\label{US}

As a baseline, we begin by considering a conventional data collection scheme, in which samples from each sensor are collected uniformly in time with sampling rate $F_s=\frac{1}{T_s}>F_c$. In this case, the data matrix $\X^{\star}$, defined in~\eqref{data_uniform}, is fully observed. In this section, we assume the samples are collected without noise; Section~\ref{MMVAND} revisits the uniform sampling scheme in the case of noisy samples.

To identify the mode shapes and frequencies in this scenario, it can be useful to consider the following ANM formulation:
\begin{align}
\Xh=\arg\min_{\X}\|\X\|_{\AA}~\change{s.t.}~\X=\X^{\star}.
\label{ANMfull}
\end{align}
Although this problem has a trivial solution (namely $\Xh = \X^{\star}$, which is already available), solving this problem in a certain way can reveal information about the mode shapes and frequencies. In particular, \eqref{ANMfull} is equivalent to the following SDP
\begin{align*}
\Xh, \uh, \Vh=\arg\min_{\X,\u,\V} ~~&\frac{1}{2M}\change{Tr}(\mathcal{T}(\bm{u}))+\frac{1}{2}\change{Tr}(\V)\\
\change{s.t.} ~~~&\left[  \begin{array}{cc}
  \mathcal{T}(\bm{u})&\X \\
  {\X}^*&\V
\end{array}   \right]\succeq 0,~~\X=\X^{\star}.
\end{align*}
Certain approaches to solving this SDP\footnote{In practice, one can use the CVX software package \cite{grant2008cvx}.} will return the dual solution $\Q$ directly. From this, following Proposition 1 in \cite{Li15}, one can formulate the dual polynomial $\QQ(f)=\Q^*\a(f)$ to identify the frequencies and mode shapes. In particular, one can identify the frequencies by localizing where the dual polynomial achieves $\| \QQ(f) \|_2=1$.
Note that the data matrix $\X^\star$ can also be rewritten as
\begin{align*}
\X^\star &= [\a(f_1)~ \cdots ~\a(f_K)] [|A_1|\bpsi_1~ \cdots~ |A_K|\bpsi_K]^\top.
\end{align*}
Then, using the estimated data matrix $\Xh$ and estimated frequencies $\fh_k$, one can solve a least-squares problem to estimate the complex amplitudes and mode shapes:
\begin{align*}
\bPsih_A \triangleq [|\Ah_1|\bpsih_1~ \cdots~ |\Ah_K|\bpsih_K] = (\A_{\fh}^\dag \Xh)^\top,
\end{align*}
where $\A_{\fh}^\dag$ denotes the pseudoinverse of the matrix $\A_{\fh}=[\a(\fh_1)~ \cdots ~\a(\fh_K)]$. Because the true mode shapes are assumed to be normalized, we normalize the estimated mode shapes by setting $|\Ah_k|=\|\bPsih_A(:,k)\|_2$ and $\bpsih_k=\frac{\bPsih_A(:,k)}{\|\bPsih_A(:,k)\|_2}$.

The following theorem is adapted from Theorem 4 in \cite{Yang14}.

\begin{Theorem}
\label{THMuniform}
\cite{Yang14} Assume the data matrix $\X^{\star}$ shown in (\ref{data_uniform}) is obtained by uniformly sampling the displacement signal vector $\x^{\star}(t)$ in time with a sampling interval $T_s<\frac{1}{F_c}=\frac{1}{\max_{k} F_k}$. If the number of samples $M$ from each sensor satisfies
\begin{align}
M\geq \max\left\{\frac{4}{\Delta_f}+1, ~257\right \},\label{Muniform}
\end{align}
then ANM perfectly recovers all frequencies $f_k$ and all mode shapes up to a phase ambiguity, i.e.,  $|\langle\bpsi_k^\star,\hat{\bpsi}_k\rangle|=1$.
\end{Theorem}



This theorem indicates that more measurements are needed at each sensor for perfect recovery as the minimum separation $\Delta_f$ decreases, i.e., as the true frequencies become closer to each other.\footnote{The resolution limit $\frac{4}{M}$ is commonly encountered in atomic norm minimization. A recent work~\cite{yang2016enhancing} proposes a reweighted atomic norm minimization algorithm that can improve upon this resolution limit.} As noted in Section~\ref{sec:modal}, Park et al.~\cite{park2014modal} have provided a theoretical analysis of an SVD based technique for modal analysis. That work also considers uniform sampling and provides a bound on $M$ that is similar to what is required in~\eqref{Muniform}. However, the analysis in~\cite{park2014modal} is limited to the case where the true mode shapes are mutually orthogonal, a condition that we do not require here. Moreover, the SVD guarantees apply only to approximate recovery of the mode shapes; as confirmed in Theorem~\ref{THMuniform}, ANM can offer exact recovery of both the mode shapes and frequencies.

\subsubsection{Synchronous random sampling}

As a first alternative to conventional uniform sampling, we now consider the case where vibration signals at each sensor are sampled randomly in time, but at the same time instants at each sensor. We refer to this data collection scheme as {\em synchronous random sampling}.

To be specific, we suppose that the random sample times are chosen as a subset of the Nyquist grid; thus, the random samples are merely a subset of a collection of uniform samples. Equivalently, we suppose that we observe $\X^{\star}$ (which is defined in~\eqref{data_uniform}) on a set of indices $\Omega_S \times [N]$ with $[N] \triangleq \{1,2, \dots, N\}$ and $\Omega_S \subset T$, where $T$ is defined in~\eqref{eq:Ttimes}. Thus, each column of $\X^{\star}$ is observed at the same times, indexed by $\Omega_S$.

Denoting the observed matrix as $\X^\star_{\Omega_S \times [N]}$, the ANM problem can be formulated as
\begin{align}
\Xh=\arg\min_{\X}\|\X\|_{\AA}~\change{s.t.}~\X_{\Omega_S \times [N]}=\X^\star_{\Omega_S \times [N]},
\label{ANM}
\end{align}
which is equivalent to the following SDP
\begin{align*}
\Xh, \uh, \Vh=\arg\min_{\X,\u,\V} ~~&\frac{1}{2M}\change{Tr}(\mathcal{T}(\bm{u}))+\frac{1}{2}\change{Tr}(\V)\\
\change{s.t.} ~~~&\left[  \begin{array}{cc}
  \mathcal{T}(\bm{u})&\X \\
  {\X}^*&\V
\end{array}   \right]\succeq 0,~~\X_{\Omega_S \times [N]}=\X^\star_{\Omega_S \times [N]}.
\end{align*}
As in Section \ref{US}, one can solve the above SDP, obtain the estimated frequencies from the dual polynomial, and recover the mode shapes and amplitudes by solving a least-squares problem.

The following theorem from~\cite{Yang14} shows that we can recover $\X^{\star}$ and estimate the frequencies accurately with high probability.

\begin{Theorem}
\label{synch}
\cite{Yang14} Suppose the data matrix $\X^{\star}$ is observed on the index set $\Omega_S \times [N]$, with $\Omega_S$ selected uniformly at random as a subset of $T$. Assume that $\{\bpsi_k^{\star}\}_{k=1}^K$ are independent random vectors with $\mathbb{E} \bpsi_k^{\star}=\bm{0}$, chosen independently of the fixed amplitudes $A_k$. That is, each $\bpsi_k^{\star}$ is sampled independently from a zero-mean distribution on the complex hyper-sphere; this distribution may vary with $k$. If $M\geq\frac{4}{\Delta_f}+1$, then there exists a numerical constant $C$ such that
\begin{align}
|\Omega_S|\geq C \max \left\{ \log^2\frac{\sqrt{N}M}{\delta},K\log\frac{K}{\delta}\log \frac{\sqrt{N}M}{\delta} \right\}
\label{eq:syncbound}
\end{align}
is sufficient to guarantee that we can exactly recover $\X^{\star}$ via \eqref{ANM} and perfectly recover the frequencies and mode shapes up to a phase ambiguity with probability at least $1-\delta$.
\end{Theorem}

The above theorem indicates that the number of measurements $|\Omega_S|$ needed from each sensor for perfect recovery scales almost linearly with the number of active modes $K$. The number of random samples per sensor $|\Omega_S|$ also increases logarithmically with $M$, the number of underlying uniform samples in $T$. The lower limit on $M$ is similar to what appears in Theorem~\ref{THMuniform} for the uniform sampling case. Thus, the total duration over which the signals must be observed does not change. However, what is significant is that in many cases the lower bound on $|\Omega_S|$ will be smaller than $M$, which confirms that in general during this time span it is not necessary to fully sample the uniform data matrix $\X^{\star}$; sensing and communication costs can be reduced by randomly subsampling this matrix.

Because its derivation relies on concentration arguments such as Hoeffding's inequality, Theorem~\ref{synch} does assume the mode shapes to be generated randomly, which is not physically plausible. Moreover, this result actually requires the number $|\Omega_S|$ of measurements per sensor to {\em increase} (albeit logarithmically) as the number of sensors $N$ increases. However, intuition suggests that the opposite should be true: the $N$ signals share a common structure (analogous to a joint sparse model~\cite{baron2009distributed} in distributed compressive sensing), and observing more signals gives more information about this common structure.

It is reasonable to expect that one will actually need {\em fewer} measurements per sensor as the number of sensors increases. This is supported by simulations (not shown in this paper), and it remains an open question to support this with theory for the case of synchronous random sampling.\footnote{After the initial submission of this manuscript, \cite{yang2017on} has appeared as a complement to \cite{Yang14} and improves upon the sample complexity in Theorem~\ref{synch} under a different random assumption on the mode shapes.}

\subsubsection{Asynchronous random sampling}

We now consider the case where vibration signals at each sensor are sampled randomly in time, but at different time instants at each sensor. We refer to this data collection scheme as {\em asynchronous random sampling}. To be specific, suppose that we observe $\X^{\star}$ on the indices $\Omega_A \subset T \times [N]$. This observation model allows each column of $\X^{\star}$ to be observed at different times, all of which are still restricted to be drawn from the uniform sampling grid $T$, defined in~\eqref{eq:Ttimes}.

Recovery of the full data matrix from asynchronous random samples is analogous to the conventional matrix completion problem. This problem can be solved using ANM with the same formulation as in~\eqref{ANM} but with $\Omega_S \times [N]$ replaced by $\Omega_A$. The following theorem from~\cite{Li15} shows that one can again recover $\X^{\star}$ and estimate the frequencies and mode shapes accurately with high probability.\footnote{An explicit proof of Theorem~\ref{asynch}, which concerns MMV ANM, is not included in~\cite{Li15}. We do note that essentially the same conclusion holds if one applies SMV ANM and uses a union bound over the $N$ sensors.}

\begin{Theorem}
\label{asynch}
\cite{Li15} Suppose the data matrix $\X^{\star}$ is observed on the index set $\Omega_A \subset T \times [N]$, which is selected uniformly at random. Assume the signs $\frac{\bpsi_{n,k}^\star}{|\bpsi_{n,k}^\star|}$ are drawn independently from the uniform distribution on the complex unit circle (and independently of the fixed amplitudes $A_k$) and that $M\geq\frac{4}{\Delta_f}+1$. Then there exists a numerical constant $C$ such that
\begin{align}
|\Omega_A|\geq C N \max \left\{ \log^2\frac{MN}{\delta},K\log\frac{KN}{\delta}\log \frac{MN}{\delta} \right\}
\label{eq:asyncbound}
\end{align}
is sufficient to guarantee that we can exactly recover $\X^{\star}$ via \eqref{ANM} and exactly recover the frequencies  and mode shapes up to a phase ambiguity with probability at least $1-\delta$.
\end{Theorem}

Up to small differences in the logarithmic factors, the total number $|\Omega_A|$ of samples required in Theorem~\ref{asynch} is comparable to the number of sensors $N$ times the number of samples per sensor $|\Omega_S|$ required in Theorem~\ref{synch}. Many of the same comments on the theorem apply here, including the use of a (different) randomness assumption on the mode shapes, and the fact that the required average number of measurements per sensor $|\Omega_A|/N$ again increases logarithmically as the number of sensors $N$ increases. Again, the lower limit on $M$ is similar to what happens in Theorem \ref{THMuniform}, and so the total duration over which the signals must be observed does not change. However, during this time span the requisite number of asynchronous random samples may be far lower than the number of uniform samples demanded by the Nyquist rate.

Comparing the logarithmic terms in the right hand sides of~\eqref{eq:syncbound} and~\eqref{eq:asyncbound}, we do see that the measurement requirement is slightly stronger for asynchronous random sampling than for synchronous random sampling. However, simulations (e.g., Fig.~\ref{AsyVsSyn}(b)) indicate that the opposite may be true. It remains an open question to support this with theory.

\subsubsection{Random temporal compression}

Inspired by alternatives to random sampling that have been considered in CS, we now consider the case where the displacement signal at each sensor is compressed via random matrix multiplication before being transmitted to a central node. We refer to this as {\em random temporal compression}. This compression strategy was also considered in~\cite{park2014modal}, which provided a theoretical analysis of an SVD based technique for modal analysis.

Let $\x^{\star}_n$ be the $n$th column of the data matrix $\X^\star$, and let $\bPhi_n\in \RRR^{M'\times M}$ be a compressive matrix generated randomly with independent and identically distributed (i.i.d.)\ Gaussian entries. At each sensor, assume that we compute the measurements
\begin{align*}
\y_n=\bPhi_n\x^{\star}_n,~~n=1,\ldots,N.
\end{align*}
At the central node, the problem of recovering the original data matrix $\X^{\star}$ from the compressed measurements can be formulated as
\begin{equation}
\begin{aligned}
\Xh=\arg &\min_{\X} \|\X\|_{\AA}\\
&~~\text{s.t.}~\y_n=\bPhi_n\x_n,~~n=1,\ldots,N,\\
&~~~~~~~~\X=[\x_1,\x_2,\cdots,\x_N],
\label{ANMOri}
\end{aligned}
\end{equation}
which is equivalent to
\begin{align*}
\Xh, \uh, \Vh=\arg\min_{\X,\u,\V} ~~&\frac{1}{2M}\change{Tr}(\mathcal{T}(\bm{u}))+\frac{1}{2}\change{Tr}(\V)\\
\change{s.t.} ~~~&\left[  \begin{array}{cc}
  \mathcal{T}(\bm{u})&\X \\
  {\X}^*&\V
\end{array}   \right]\succeq 0,~~\y_n=\bPhi_n\x_n,~~n=1,\ldots,N,\\
&\X=[\x_1,\x_2,\cdots,\x_N]
\end{align*}

As in Section \ref{US}, one can solve the above SDP, obtain the estimated frequencies from the dual polynomial, and recover the mode shapes and amplitudes by solving a least-squares problem. We have the following result.

\begin{Theorem}
\label{THMCS}
Suppose $\bPhi_1, \bPhi_2, \dots, \bPhi_N \in \RRR^{M'\times M}$ are independently drawn random matrices with i.i.d.\ Gaussian entries having mean zero and variance $1$. Assume that the true frequencies satisfy the minimum separation condition
\begin{align}
\Delta_f\geq\frac{1}{\lfloor(M-1)/4\rfloor}. \label{mise}
\end{align}
Then there exists a numerical constant $C$ such that if
\begin{align}
M' \geq CK\log(M),  \label{Mp}
\end{align}
$\X^{\star}$ is the unique optimal solution of the ANM problem~\eqref{ANMOri} and we can exactly recover the frequencies  and mode shapes up to a phase ambiguity with probability at least $1-e^{-\frac{N(M'-2)}{8}}$.
\end{Theorem}

We prove Theorem~\ref{THMCS} in Section~\ref{proofTHMCS}, extending analysis from Theorem 2 in \cite{heckel2016generalized}, which concerned the SMV ($N=1$) version of this problem.

From Theorem~\ref{THMCS}, we see that modal analysis is possible in the random temporal compression scenario using a number $M'$ of measurements per sensor that scales linearly with the number of active modes $K$. One logarithmic term remains in this bound~\eqref{Mp}, while some of the other logarithmic terms appearing in~\eqref{eq:syncbound} and~\eqref{eq:asyncbound} have disappeared. Expressing the failure probability as $\delta$ for easier comparison with Theorems~\ref{synch} and~\ref{asynch}, we see that when
\begin{align}
M'\geq \max \left\{ \frac{8}{N} \log\left(\frac{1}{\delta}\right) +2, CK\log(M)  \right\}, \label{eq:ourmprime}
\end{align}
exact recovery of the mode shapes and frequencies is possible with probability at least $1-\delta$. In the first term inside this expression, the number $M'$ of measurements per sensor actually {\em decreases} as the number $N$ of sensors increases, a contrast with Theorem~\ref{synch} and~\ref{asynch} where the opposite trend holds. We note that the benefit of increasing the number of sensors $N$ comes from the fact that in the setup of Theorem~\ref{THMCS}, each sensor uses a different measurement matrix $\Phi_n$. When the second term in~\eqref{eq:ourmprime} dominates, the number $M'$ of measurements per sensor still does not increase with $N$. Similar to the random sampling schemes, the total duration over which the signals must be observed does not change, but within this time span significant compression may be possible.

Also in contrast with Theorems~\ref{synch} and~\ref{asynch}, Theorem~\ref{THMCS} requires no randomness assumption on the mode shapes. This allows for its use in practical scenarios, where in general mode shapes will not be generated randomly. It would be interesting to remove the randomness condition from Theorems~\ref{synch} and~\ref{asynch}; we leave this as a question for future work. 


\subsubsection{Random spatial compression}

As a final strategy to compress the data matrix, we consider the scenario where each sensor modulates its sample value by a random number at each time sample, then the sensors transmit these values coherently to the central node, where the modulated values add to result in a single measurement vector. As discussed in~\cite{bajwa2006compressive}, such randomized spatial aggregation of the measurements can be achieved as part of using phase-coherent analog
transmissions to the base station. This strategy is sometimes referred to as {\em compressive wireless sensing}; we refer to it as {\em random spatial compression}.

In random spatial compression, the collected measurements can be expressed as
\begin{align*}
y_m=\langle \X^{\star \top}(:,m), \tilde{\b}_m \rangle=\langle \X^{\star \top},\tilde{\b}_m\bm{e}_m^\top\rangle,~~m=1,\ldots,M
\end{align*}
where $\tilde{\b}_m \in \mathbb{C}^{N\times1}$ and $\bm{e}_m\in \mathbb{R}^{M\times 1}$ is the $m$th canonical basis vector.
%
%
%
%
This allows us to formulate the modal analysis problem as an ANM problem:
\begin{equation}
\begin{aligned}
\Xh = &\arg\min_{\X}\|\X\|_{\AA}~\\
&\change{s.t.}~y_m=\langle\X^\top,\tilde{\b}_m\e_m^\top\rangle,~~1\leq m \leq M,\label{ANMspatial}
\end{aligned}
\end{equation}
which is equivalent to
\begin{align*}
\Xh, \uh, \Vh=\arg\min_{\X,\u,\V} ~~&\frac{1}{2M}\change{Tr}(\mathcal{T}(\bm{u}))+\frac{1}{2}\change{Tr}(\V)\\
\change{s.t.} ~~~&\left[  \begin{array}{cc}
  \mathcal{T}(\bm{u})&\X \\
  {\X}^*&\V
\end{array}   \right]\succeq 0,~~y_m=\langle\X^\top,\tilde{\b}_m\e_m^\top\rangle,~~1\leq m \leq M.
\end{align*}

As in Section \ref{US}, one can solve the above SDP, obtain the estimated frequencies from the dual polynomial, and recover the mode shapes and amplitudes by solving a least-squares problem. The following theorem follows from our previous work~\cite{yang2016super} on super-resolution of complex exponentials from modulations with unknown waveforms.

\begin{Theorem}
\label{thm:randspace}
\cite{yang2016super} Suppose we observe the data matrix $\X^\star$ with the above random spatial compression scheme. Assume that the random vectors $\tilde{\b}_m$ are i.i.d.\ samples from an isotropic and $\mu-$incoherent distribution (see~\cite{yang2016super} for details). Also, suppose that $\bpsi_{k}^\star$ are drawn i.i.d. from the uniform distribution on the complex unit sphere and that the minimum separation condition (\ref{mise}) is satisfied. Then there exists a numerical constant $C$ such that\footnote{Note that the bound in (\ref{eq:randspcresult}) contains $M$ on both sides. One could remove $M$ from the right hand side using the Lambert W-function as in \cite[(64)--(66)]{eftekhari2016stabilizing}. Here, we prefer the form in (\ref{eq:randspcresult}) for simplicity and to highlight the relationship between $M$ and $K,N$.}
\begin{align}
M \geq C\mu KN \log\left(\frac{MKN}{\delta} \right)\log^2 \left( \frac{MN}{\delta} \right)
\label{eq:randspcresult}
\end{align}
is sufficient to guarantee that we can exactly recover $\X^{\star}$ via \eqref{ANMspatial} and exactly recover the frequencies and mode shapes up to a phase ambiguity with probability at least $1-\delta$.
\end{Theorem}

In random spatial compression, the central node receives exactly one compressed measurement at each time instant. Thus, the number of time samples $M$ equals the total number of compressed measurements. Thus, Theorem~\ref{thm:randspace} states that it is sufficient for the total number of compressed measurements to scale essentially linearly with $KN$ (as long as (\ref{mise}) is also satisfied). Since $K$ is the number of active modes and $N$ is the number of sensors (and thus the length of each unknown mode shape), the number of unknown degrees of freedom in this problem scales with $KN$. In this sense, the result in Theorem~\ref{thm:randspace} compares favorably with those in Theorems~\ref{synch}, \ref{asynch}, and~\ref{THMCS}. Theorem~\ref{thm:randspace} does require a randomness assumption on the mode shapes.

We note that, in some applications, once (\ref{mise}) is satisfied (which imposes a lower bound on $M$ that is comparable to what appears in Theorems~\ref{THMuniform}--\ref{THMCS}), it could be the case that~\eqref{eq:randspcresult} is also satisfied. In this case, the same uniform data matrix $\X^\star$ that suffices for perfect recovery according to Theorem~\ref{THMuniform} can be completely compressed in the spatial dimension, reducing the total number of samples from $MN$ to $M$. Thus, significant savings may be possible in structures where the number of sensors $N$ is large. Indeed, up to the point where~\eqref{eq:randspcresult} becomes a stronger condition than (\ref{mise}), one could continue adding sensors without increasing the requisite number of compressed samples.

\subsection{Modal analysis for noisy signals}
\label{MMVAND}

In this section, we revisit the uniform sampling scenario where a data matrix is fully observed, but we now consider the case where the samples are corrupted by additive white Gaussian noise. While the analysis in this section may be of its own independent interest, it is also used in our proof of Theorem~\ref{THMCS}.

We consider observations of the form $\Y=\X^{\star}+\W$, where the entries of $\W$ satisfy $\CC\NN(0,\sigma^2)$, and we consider the following atomic norm denoising problem:
\begin{align}
\min_{\X} \frac{1}{2} \|\Y-\X\|_F^2+\lambda\|\X\|_{\AA}.
\label{AND}
\end{align}
The theorem below provides an upper bound on the recovery error in Frobenius norm and is proved in Section~\ref{ProofANDB}.
\begin{Theorem}
\label{ANDB}
Suppose the true data matrix $\X^{\star}$ is given as in~\eqref{data_uniform} with the true frequencies satisfying the minimum separation condition~\eqref{mise}. Given the noisy data $\Y=\X^{\star}+\W$, where the entries of $\W$ are i.i.d.\ complex Gaussian random variables which satisfy $\CC\NN(0,\sigma^2)$, the estimate $\Xh$ obtained by solving the atomic norm denoising problem~\eqref{AND} (with regularizing parameter $\lambda =\eta \sigma\sqrt{4MN\log(M)}$ and with $\eta\in(1,\infty)$ chosen sufficiently large) will satisfy
\begin{align}
\|\Xh-\X^{\star}\|_F^2\leq C\sigma^2KN\log(M) \label{Ebound}
\end{align}
with probability at least $1-\frac{1}{M^2}$ for a numerical constant $C$.
\end{Theorem}

\begin{Corollary}
\label{EAND}
Under the assumptions of Theorem~\ref{ANDB}, the estimate $\Xh$ obtained by solving the atomic norm denoising problem~\eqref{AND} (with regularizing parameter $\lambda =\eta \sigma\sqrt{4MN\log(M)}$ and with $\eta\in(1,\infty)$ chosen sufficiently large) will satisfy
\begin{align}
\EEE \|\Xh-\X^{\star}\|_F^2 \leq C\sigma^2KN\log(M). \label{EG}
\end{align}
\end{Corollary}

Note that Theorem \ref{ANDB} provides a bound which extends the SMV case in \cite{tang2015near} to the MMV case. Corollary \ref{EAND} is a direct consequence of Theorem \ref{ANDB} and is proved in Section \ref{ProofEAND}. The above MMV atomic norm denoising problem (\ref{AND}) is also considered in \cite{Li15}. However, there the authors provide only an asymptotic bound on $\EEE \|\Xh-\X^{\star}\|_F^2$.

\section{Simulation Results}
\label{simu}

In this section, we present some experiments on synthetic data to exhibit the performance of ANM based modal analysis in the various sampling and compression scenarios.\footnote{We use ADMM \cite{Li15} to solve the atomic norm denoising problem in Section \ref{noda}. All other simulations are implemented with CVX~\cite{grant2008cvx}.} We use the modal assurance criterion (MAC) to evaluate the quality of recovered mode shapes, which is defined as
\begin{align}
\change{MAC}(\bpsi_k^\star,\hat{\bpsi}_k) =|\langle\bpsi_k^\star,\hat{\bpsi}_k\rangle|, \nonumber
\end{align}
where $\hat{\bpsi}_k$ is the $k$th estimated mode shape and $\bpsi_k^\star$ is the $k$th true mode shape.\footnote{To correctly pair the recovered mode shapes with the true ones, we assume the true frequencies $f_k$ are in ascending order, and we adopt the same convention for the estimated frequencies.}
A value of $\change{MAC}=1$ would indicate perfect recovery of the true mode shape. We consider mode shape recovery to be a success if $\change{MAC}(\bpsi_k^\star,\hat{\bpsi}_k) \geq 0.99$ for all $k$. We consider data matrix recovery to be a success if $\frac{\|\Xh-\X^{\star}\|_F}{\|\X^{\star}\|_F}\leq 10^{-5}$.

\subsection{Uniform sampling}
\label{uniform}

As mentioned in Section~\ref{US}, the full data, uniform sampling case has previously been considered in~\cite{Yang14}. In this experiment, we compare the ANM based algorithm with the SVD based algorithm from~\cite{park2014modal} on a simple 6-degree-of-freedom boxcar system. The boxcar system is shown in Fig.~\ref{fig:boxcar}. We consider this undamped system under the context of free vibration. The system parameters are set as follows: the masses are $m_1=1,~m_2=2,~m_3=3,~m_4=4,~m_5=5,~m_6=6$ kg, and the stiffness values are $k_1=k_7=500$, $k_2=k_6=150$, $k_3=k_5=100$, $k_4=50$ N/m. Thus, the mass, damping and stiffness matrices in~\eqref{EOM} are given as $\M=\change{diag}([m_1,\ldots,m_6])$, $\C=\zero$, and
\begin{align*}
\K=\left[
\begin{array}{cccccc}
k_1+k_2& -k_2&0&0&0&0\\
-k_2&k_2+k_3&-k_3&0&0&0\\
0&-k_3&k_3+k_4&-k_4&0&0\\
0&0&-k_4&k_4+k_5&-k_5&0\\
0&0&0&-k_5&k_5+k_6&-k_6\\
0&0&0&0&-k_6&k_6+k_7
\end{array}
\right],
\end{align*}
respectively. The true mode shapes and natural frequencies of this system can be obtained from the (normalized) generalized eigenvectors and square root of the generalized eigenvalues of the stiffness matrix $\K$ and mass matrix $\M$. In particular, the true frequencies are $F_1 = 0.5384$, $F_2 = 0.8964$, $F_3 = 1.2404$, $F_4 = 1.7434$, $F_5 = 1.7881$, and $F_6 = 4.1218$ Hz.  We collect $M=100$ uniform samples from this system with sampling interval $T_s = 0.9/F_c$, where $F_c=\max_{1\leq k \leq 6} |F_k|$.
The amplitudes are set as $A_1=1$, $A_2=0.85$, $A_3=0.7$, $A_4=0.5$, $A_5=0.25$ and  $A_6=0.1$. We apply both ANM and SVD to the obtained data matrix to identify the modal parameters of this boxcar system. (Note that the SVD based algorithm only estimates the mode shapes and not the frequencies.) Figure~\ref{Uniform_BoxCar_All}(a) shows how the dual polynomial from ANM can be used to localize the frequencies. In particular, we identify the frequencies by identifying where the dual polynomial achieves $\| \QQ(f) \|_2=1$. The true frequencies and estimated frequencies are presented in Fig.~\ref{Uniform_BoxCar_All}(b). The estimated mode shapes from the two algorithms are illustrated in Fig.~\ref{Uniform_BoxCar_All}(c). As mentioned in Section~\ref{sec:modal}, the SVD based algorithm can only return mode shape estimates that are mutually orthogonal (in fact they are orthogonal singular vectors of the data matrix). The true mode shapes in this experiment are not orthogonal, which hampers the performance of the SVD. The ANM algorithm is not restricted to returning orthogonal mode shape estimates, and in this experiment it recovers the mode shapes perfectly. In particular, the MAC for AMN is $(1,1,1,1,1,1)$, while the MAC for SVD is $(0.9176, 0.7231, 0.8060, 0.9749,0.9659,0.9883)$.

In this experiment, the minimum separation $\Delta_f=0.0098$. Theorem~\ref{THMuniform} guarantees that perfect recovery is possible via ANM when $M \geq \max\{\frac{4}{\Delta_f}+1,257 \}=\max \{ 410, 275 \}$. We see perfect recovery in this simulation with $M = 100$ uniform samples.

In the following sections, for convenience we will use random mode shapes and/or discrete frequencies to test the ANM based algorithms.


%

\begin{figure}[t]
\begin{minipage}{0.32\linewidth}
\centering
\includegraphics[width=2.3in]{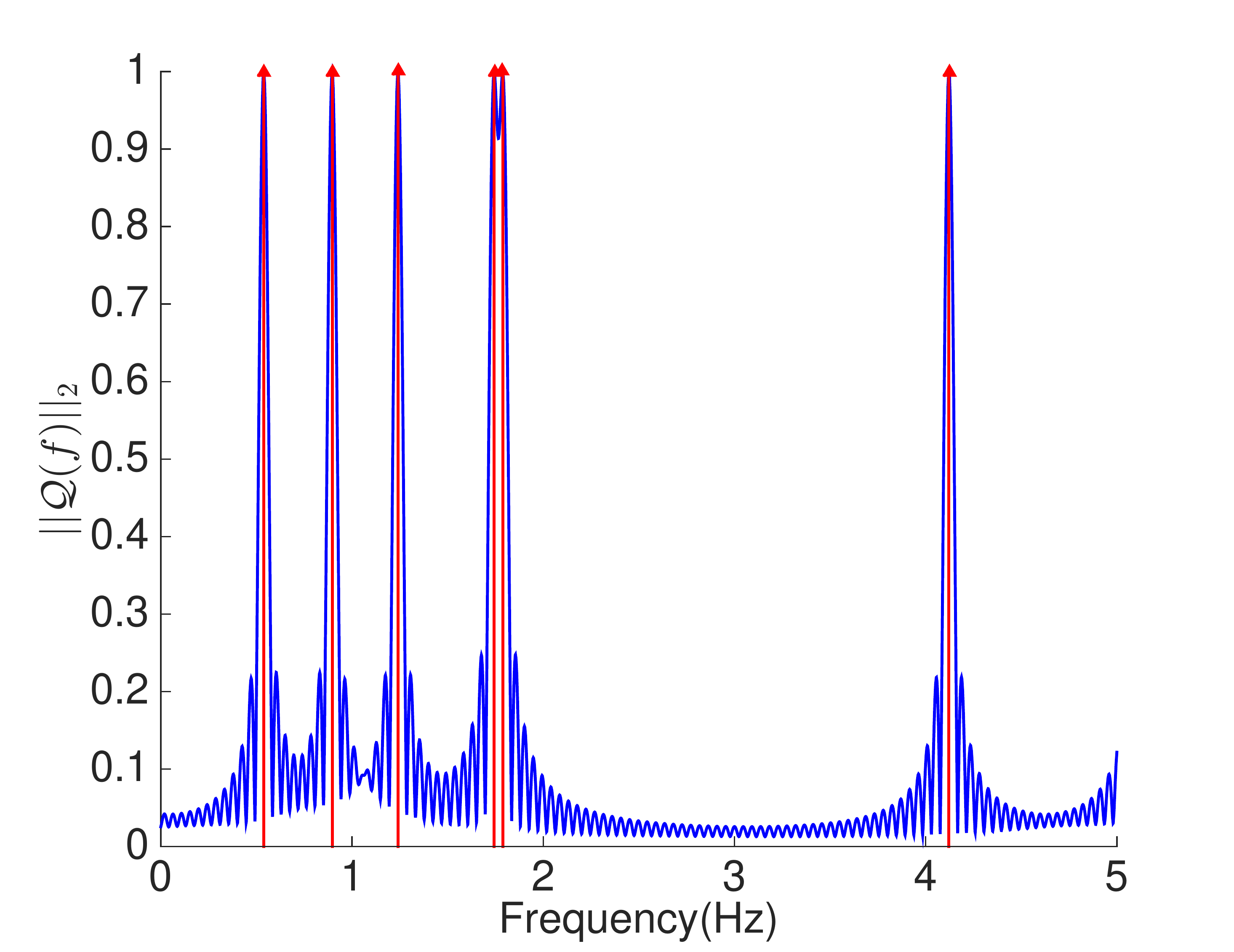}
\centerline{\footnotesize{(a)}}
\end{minipage}
\hfill
\begin{minipage}{0.32\linewidth}
\centering
\includegraphics[width=2.3in]{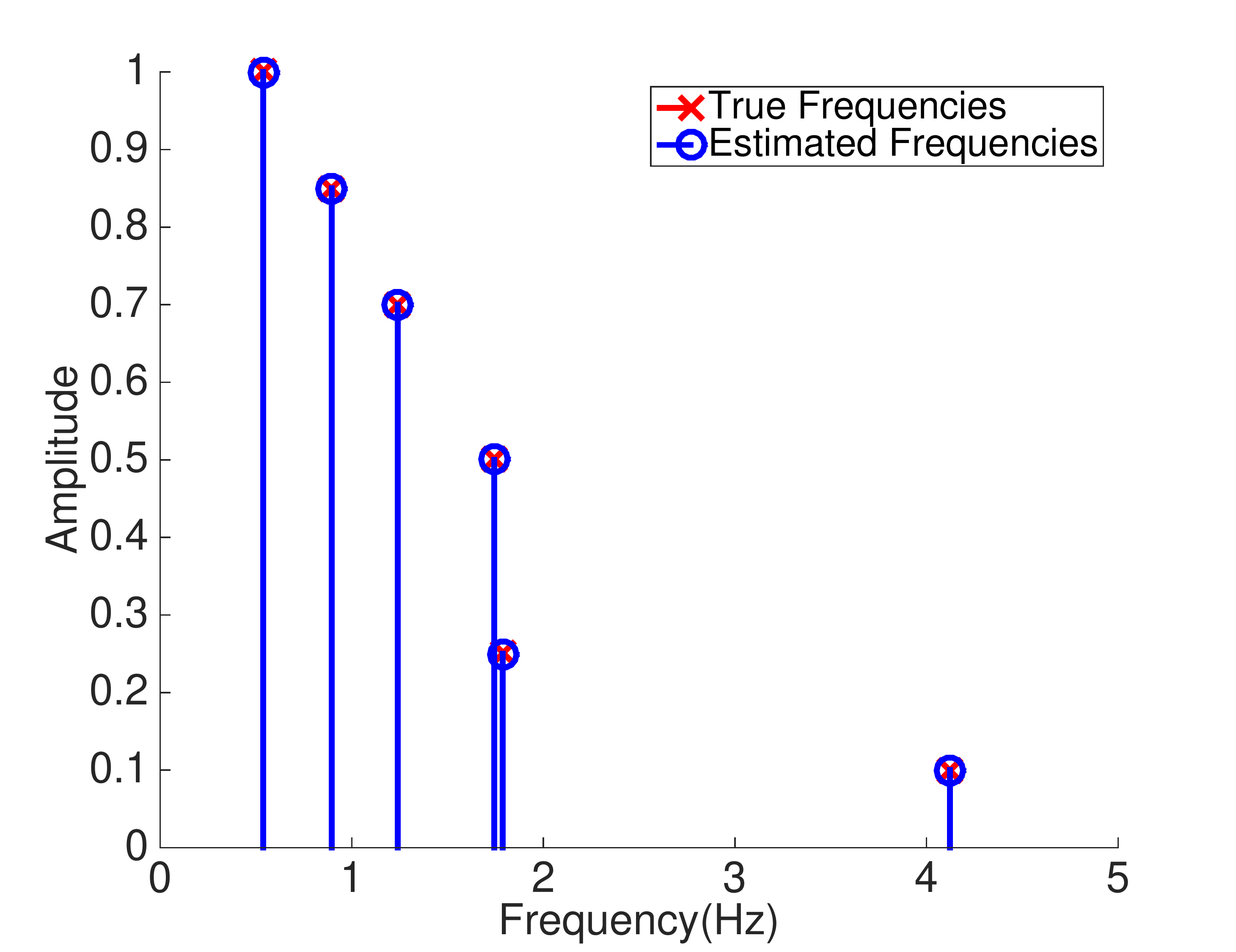}
\centerline{\footnotesize{(b)}}
\end{minipage}
\hfill
\begin{minipage}{0.32\linewidth}
\centering
\includegraphics[width=2.3in]{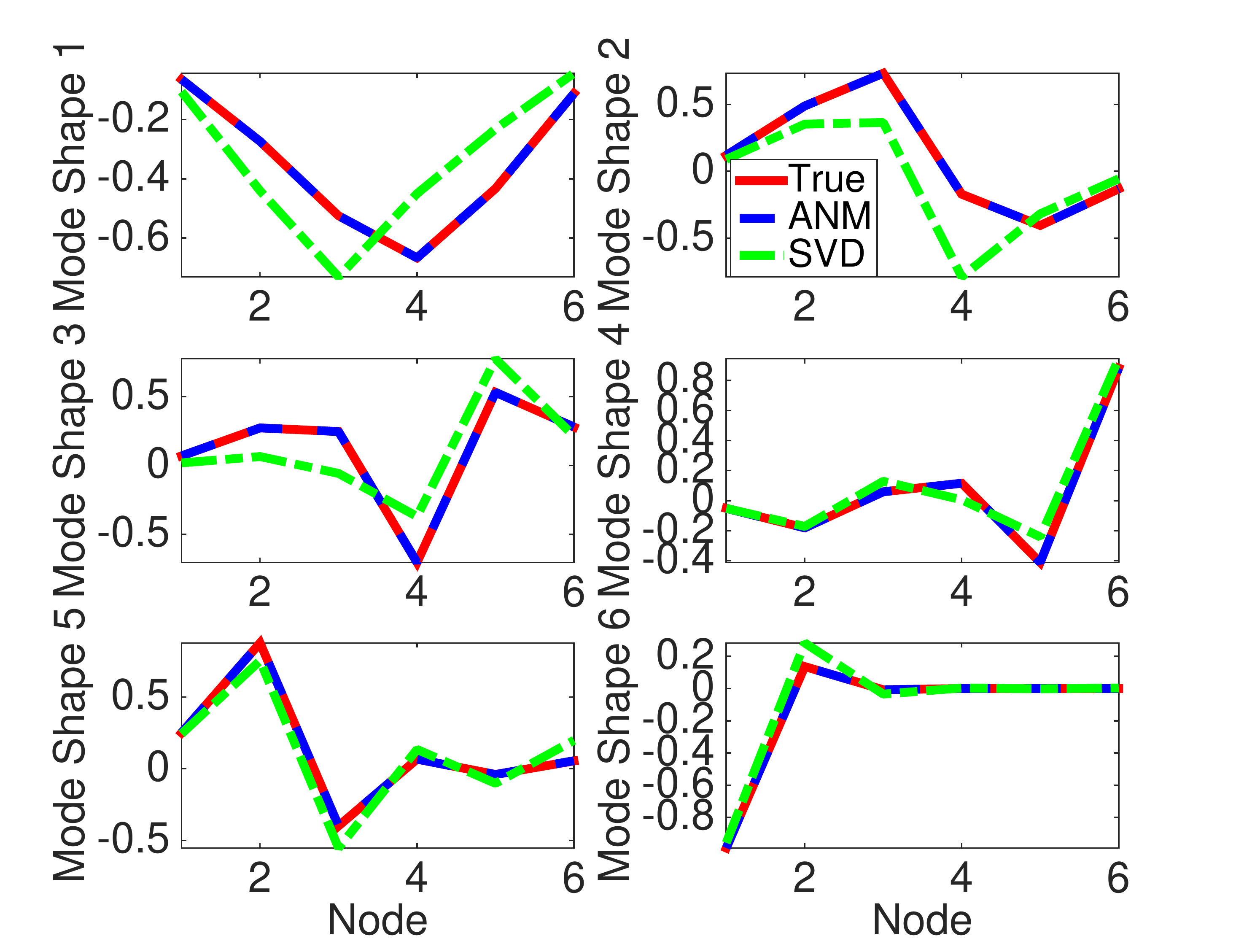}
\centerline{\footnotesize{(c)}}
\end{minipage}
\caption{Uniform sampling in the boxcar system. (a) ANM based frequency localization from dual polynomial $\QQ(f)$. The horizontal axis is shown in units of Hz, which corresponds to the digital frequency $f$ times the sampling rate $\frac{1}{T_s}$.
(b) Estimated frequencies by ANM. (c) Estimated mode shapes by ANM and SVD.}
\label{Uniform_BoxCar_All}
\end{figure}

\subsection{Asynchronous vs.\ synchronous random sampling}
\label{simu_AsyVsSyn}

In this experiment, we compare the performance of asynchronous and synchronous random sampling in a case where the mode shapes are randomly generated but also correlated.\footnote{We generate the first and third mode shapes randomly with i.i.d.\ Gaussian entries and then normalize. The second mode shape is generated by slightly perturbing the first mode shape and then normalizing.} An example of such correlated mode shapes is shown in Fig.~\ref{AsyVsSyn}(a). (Only the first two mode shapes are correlated.) The true discrete frequencies are set to $f_1= 0.1$, $f_2 = 0.15$, and $f_3 = 0.5$. We collect $M=80$ uniform samples from each sensor. However, from these, on average, we keep only $M' < M$ random samples from each sensor, where the value for $M'$ ranges from $2$ to $20$. In the case of synchronous random sampling, we keep exactly $M'$ samples from each sensor at the same times. In the case of asynchronous random sampling, we generate $\Omega_A$ uniformly at random, with $|\Omega_A| = M'N$.

Other parameters are set the same as in Section \ref{uniform}. We perform 300 trials (each with a new set of mode shapes) for each value of $M'$. Figure~\ref{AsyVsSyn}(b) shows that when compared with synchronous sampling, asynchronous random sampling needs fewer observed measurements to achieve the same probability of successful recovery. This observation is reasonable (given the additional diversity in the asynchronous observations) but stands in contrast with the relative difference between the theoretical bounds in Theorems~\ref{synch} and~\ref{asynch}. More work may be needed to theoretically characterize the performance difference between asynchronous and synchronous random sampling.

%

\begin{figure}[t]
\begin{minipage}{0.49\linewidth}
\centering
\includegraphics[width=2.4in]{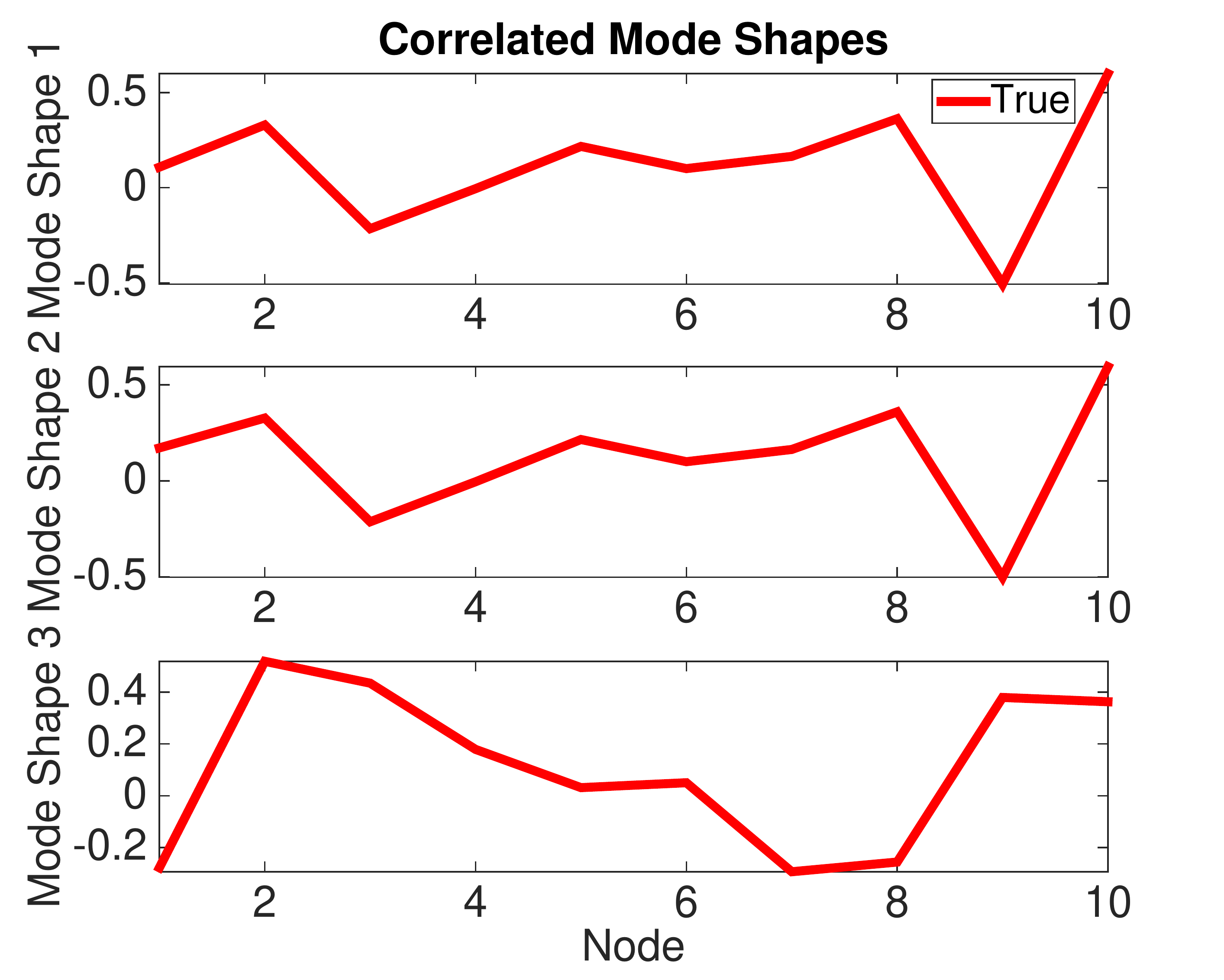}
\centerline{\footnotesize{(a)}}
\end{minipage}
\hfill
\begin{minipage}{0.49\linewidth}
\centering
\includegraphics[width=2.4in]{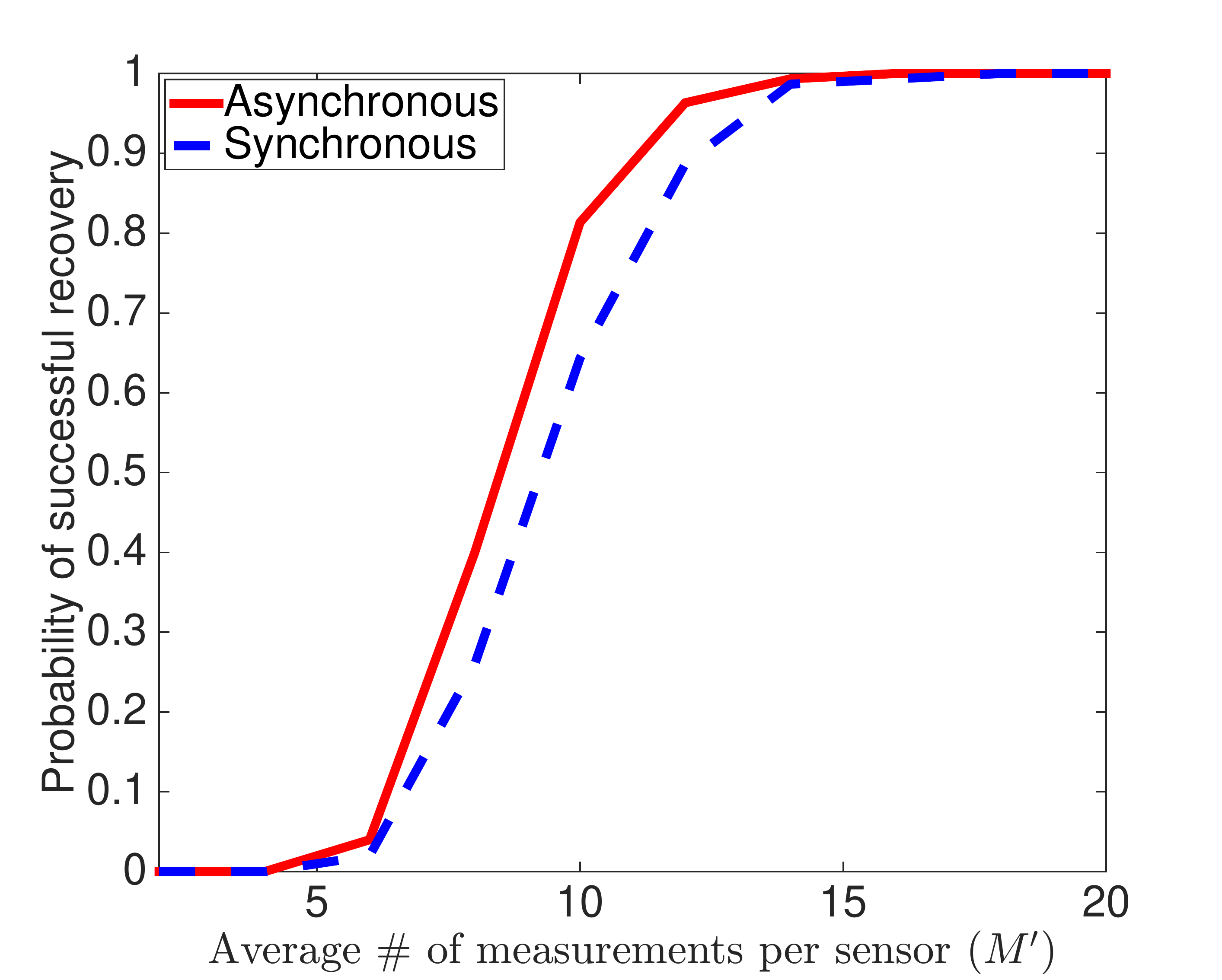}
\centerline{\footnotesize{(b)}}
\end{minipage}
\caption{Asynchronous vs.\ synchronous random sampling. (a) Correlated random mode shapes used in experiments. (b) Probability of successful recovery of mode shapes with each sampling scheme.}
\label{AsyVsSyn}
\end{figure}

\subsection{Random temporal compression}

In the next set of experiments, we generate a series of random matrices $\bPhi_n\in \mathbb{C}^{M' \times M},~n=1,\ldots,N$ to compress a set of $M=80$ uniform samples at each sensor. In the first experiment, we recover the data matrix and mode shapes both jointly (via the MMV approach from \eqref{ANMOri}) and separately (by solving $N$ separate SMV problems) to show the advantage of joint recovery. We choose $M'$ between $3$ and $30$ and perform 100 trials for each value of $M'$. We use random mode shapes, all generated with i.i.d.\ Gaussian entries and then normalized. The true discrete frequencies are set to $f_1= 0.1$, $f_2 = 0.15$, and $f_3 = 0.5$, giving a separation of $\Delta_f = 0.05$, which is slightly smaller than the separation condition prescribed in \eqref{mise}. Figures \ref{RandTemp}(a), (b) show the probability of successful recovery for the data matrix and mode shapes, respectively, when we use separate ANM (dashed lines) and joint ANM (solid lines). The number on each line denotes the number of sensors $N$ used in the experiments. It can be seen that joint recovery outperforms separate recovery significantly. Moreover, these results also indicate that the number of sensors has an important effect on the performance of joint recovery. In particular, for a given number of measurements $M'$, the probability of successful joint recovery will increase as the number of sensors increases, which is consistent with our theoretical analysis. In contrast, the probability of successful separate recovery will decrease.

In the second experiment, we set the number of sensors to be $N=10$ and investigate the minimal number of measurements per sensor $M'$ needed for perfect joint recovery with various numbers of active modes $K$. The true mode shapes are generated randomly and we set $M=100$. For each value of $K$, we randomly pick $K$ discrete frequencies from a frequency set $\FF=0.03: \change{Sep}:0.99$, where $\change{Sep}=2/M$. The amplitudes $A_k,~ k=1,\ldots,K$ are chosen randomly from the uniform distribution between $0$ and $1$. It can be seen in Fig.~\ref{RandTemp}(c) that the minimal number of measurements needed by each sensor for perfect recovery does scale roughly linearly with the number of active modes $K$, as indicated in Theorem \ref{THMCS}.
%

\begin{figure*}[t]
\begin{minipage}{0.32\linewidth}
\centering
\includegraphics[width=2.3in]{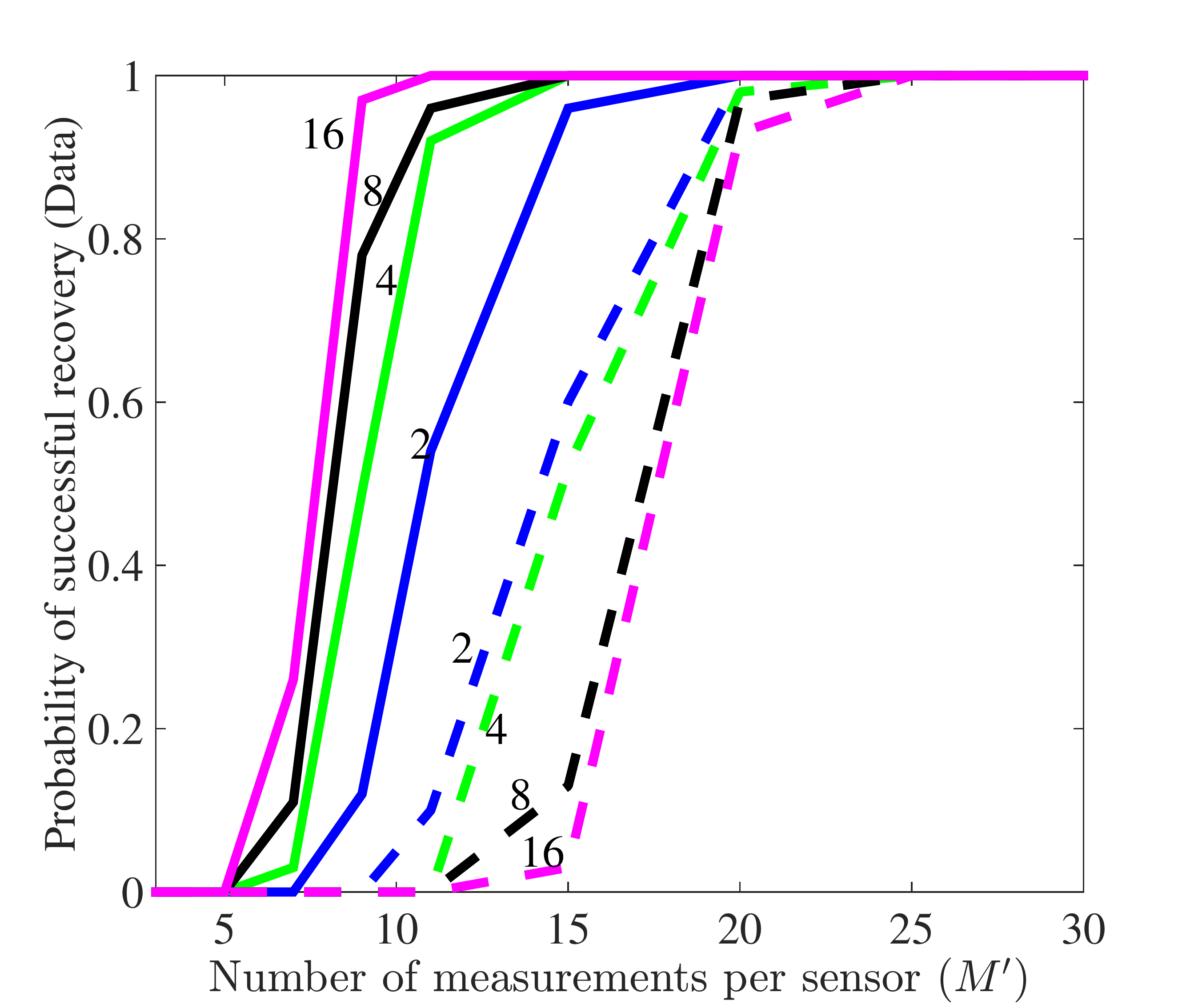}
\centerline{\footnotesize{(a)}}
\end{minipage}
\hfill
\begin{minipage}{0.32\linewidth}
\centering
\includegraphics[width=2.3in]{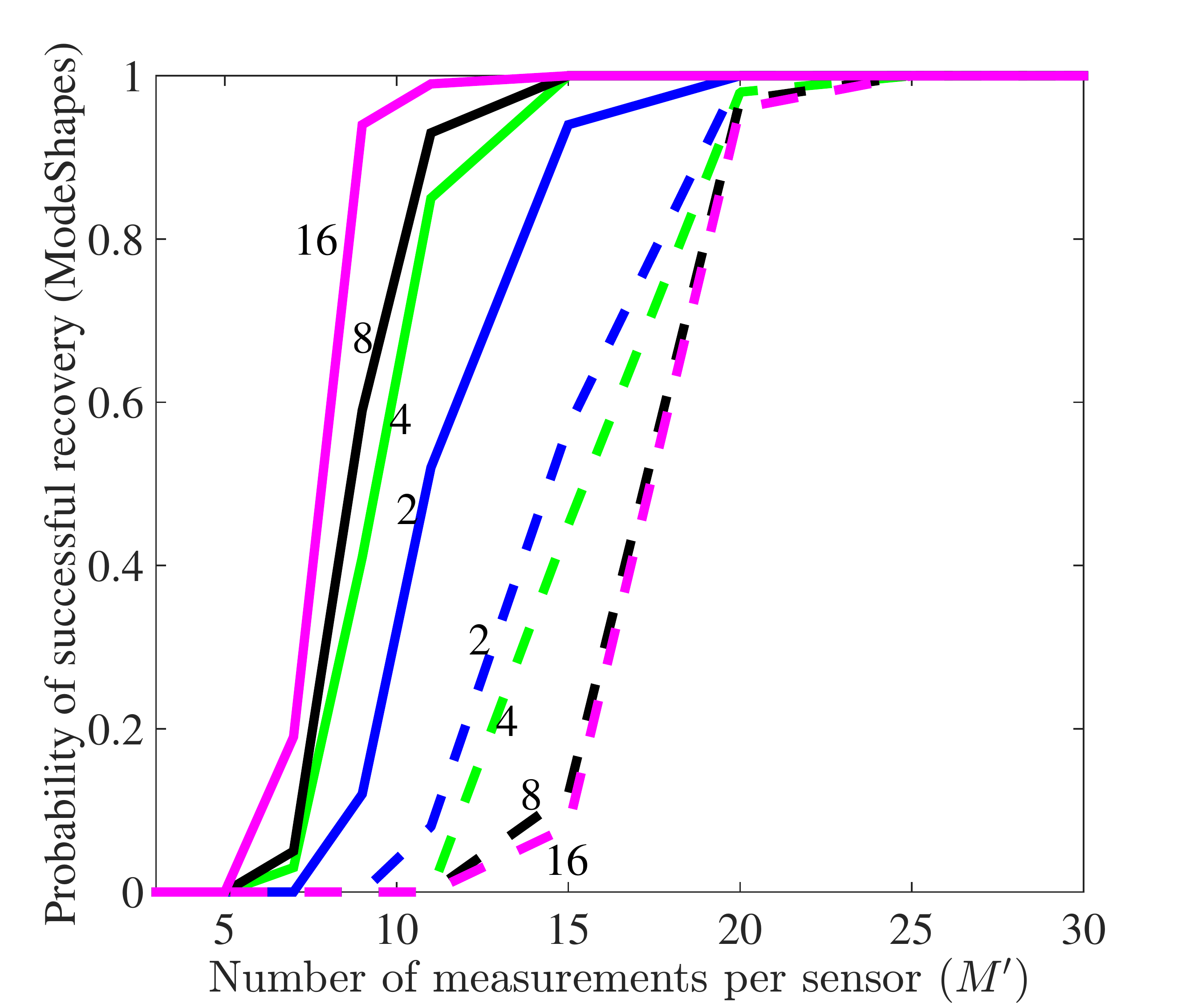}
\centerline{\footnotesize{(b)}}
\end{minipage}
\hfill
\begin{minipage}{0.32\linewidth}
\centering
\includegraphics[width=2.3in]{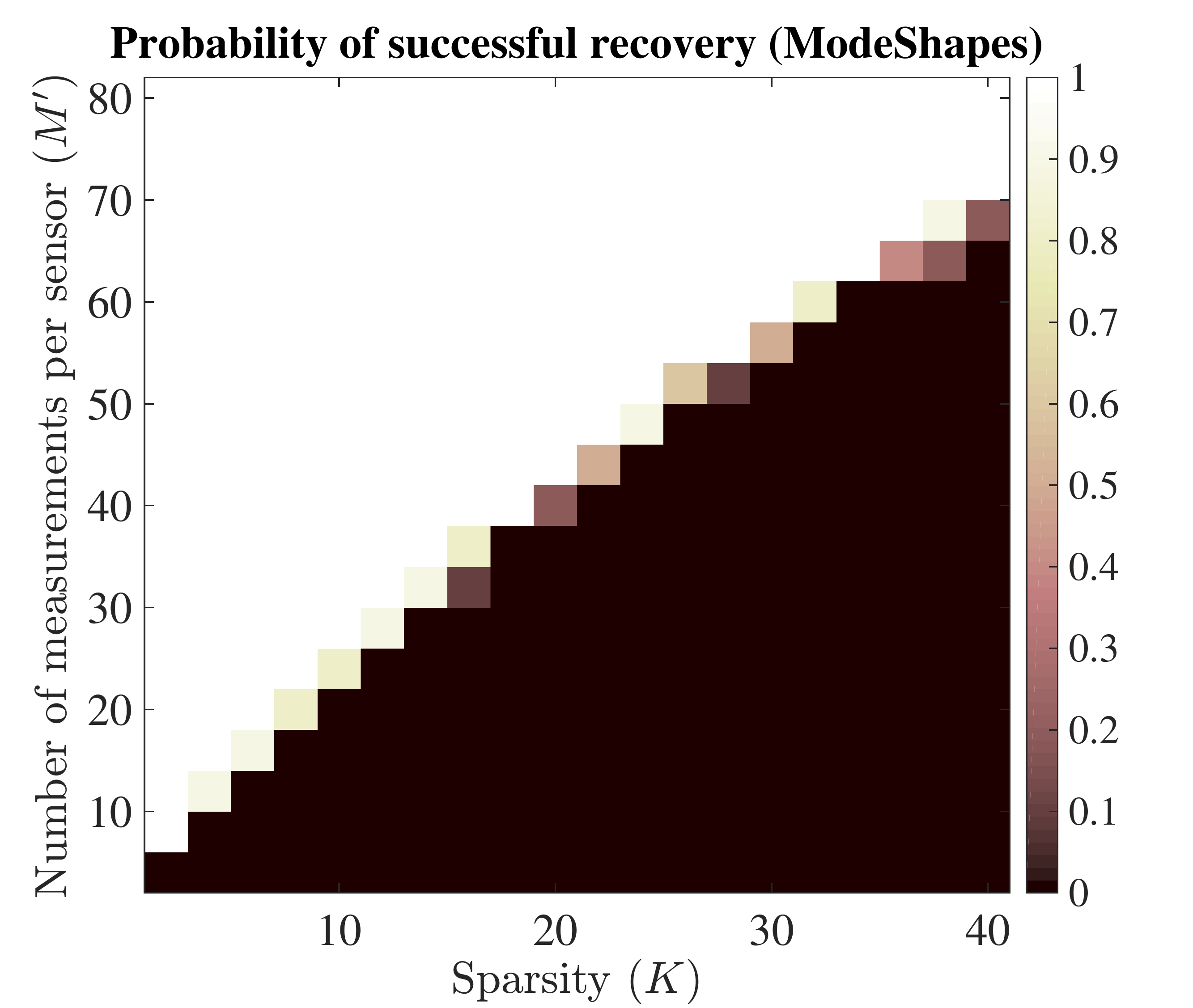}
\centerline{\footnotesize{(c)}}
\end{minipage}
\caption{Random temporal compression. (a)~Recovery of data matrix via joint ANM (solid lines) and separate ANM (dashed lines). (b)~Recovery of mode shapes via joint ANM (solid lines) and separate ANM (dashed lines). (c)~Probability of successful recovery for mode shapes with ANM when $N=10$ is fixed.}
\label{RandTemp}
\end{figure*}

\subsection{Random spatial compression}

Using the same parameters as in Section~\ref{simu_AsyVsSyn}, we simulate the random spatial compression strategy with ANM based modal analysis and compare to the SVD algorithm studied in \cite{park2014modal}, which we apply to a full $M \times N$ data matrix. We generate mode shapes randomly, testing both orthogonal mode shapes and correlated mode shapes. It can be seen from Fig.~\ref{RandSpa}(a), (b) that the SVD based method performs poorly when the mode shapes are correlated. However, the proposed ANM based algorithm performs well in both cases. Note that although SVD based method performs very well when the mode shapes are orthogonal, this is using $M \cdot N = 80 \cdot 10 = 800$ samples while the ANM based algorithm uses only $M=80$ spatially compressed measurements to recover both the mode shapes and the frequencies.


\begin{figure}[t]
\begin{minipage}{0.49\linewidth}
\centering
\includegraphics[width=2.4in]{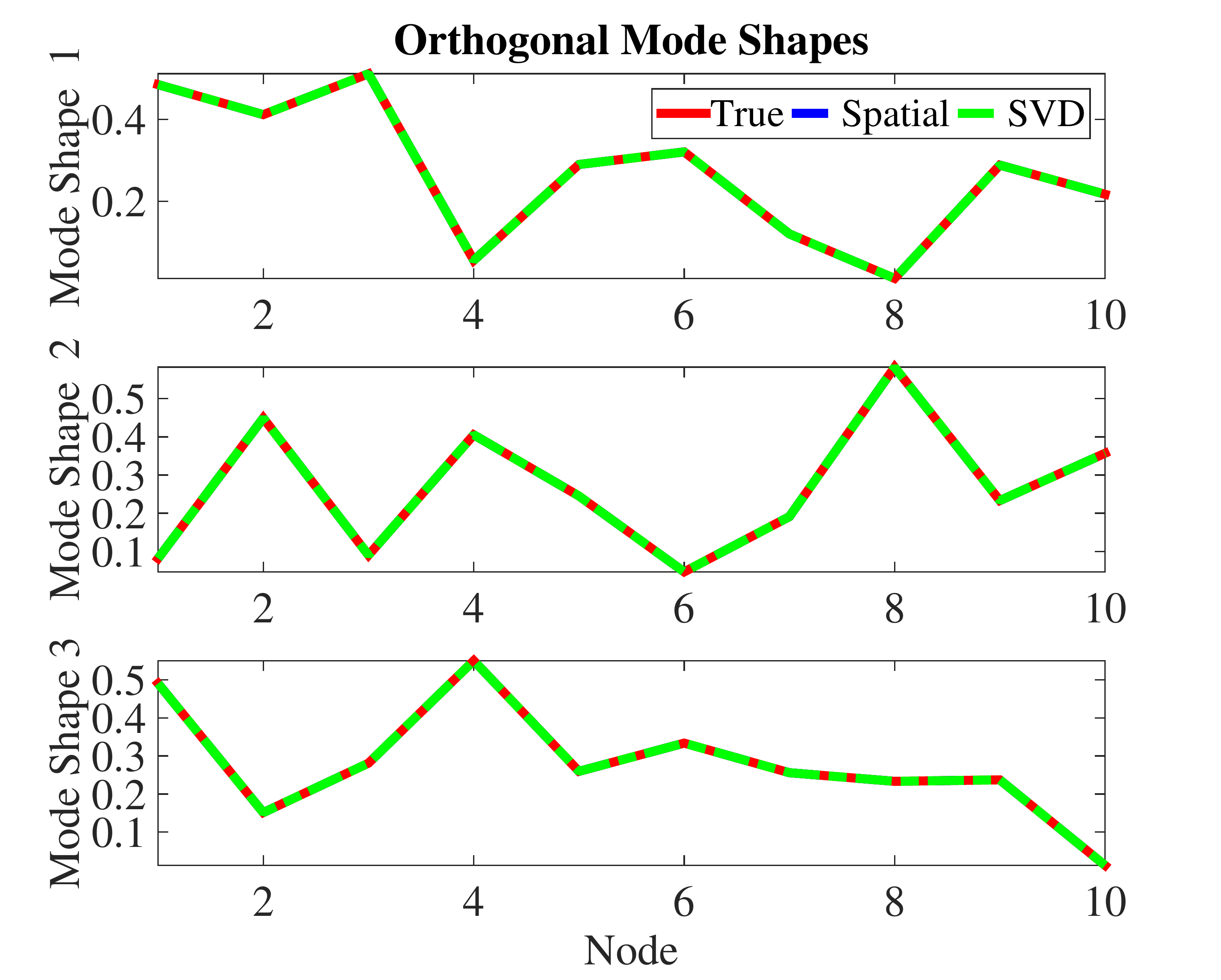}
\centerline{\footnotesize{(a)}}
\end{minipage}
\hfill
\begin{minipage}{0.49\linewidth}
\centering
\includegraphics[width=2.4in]{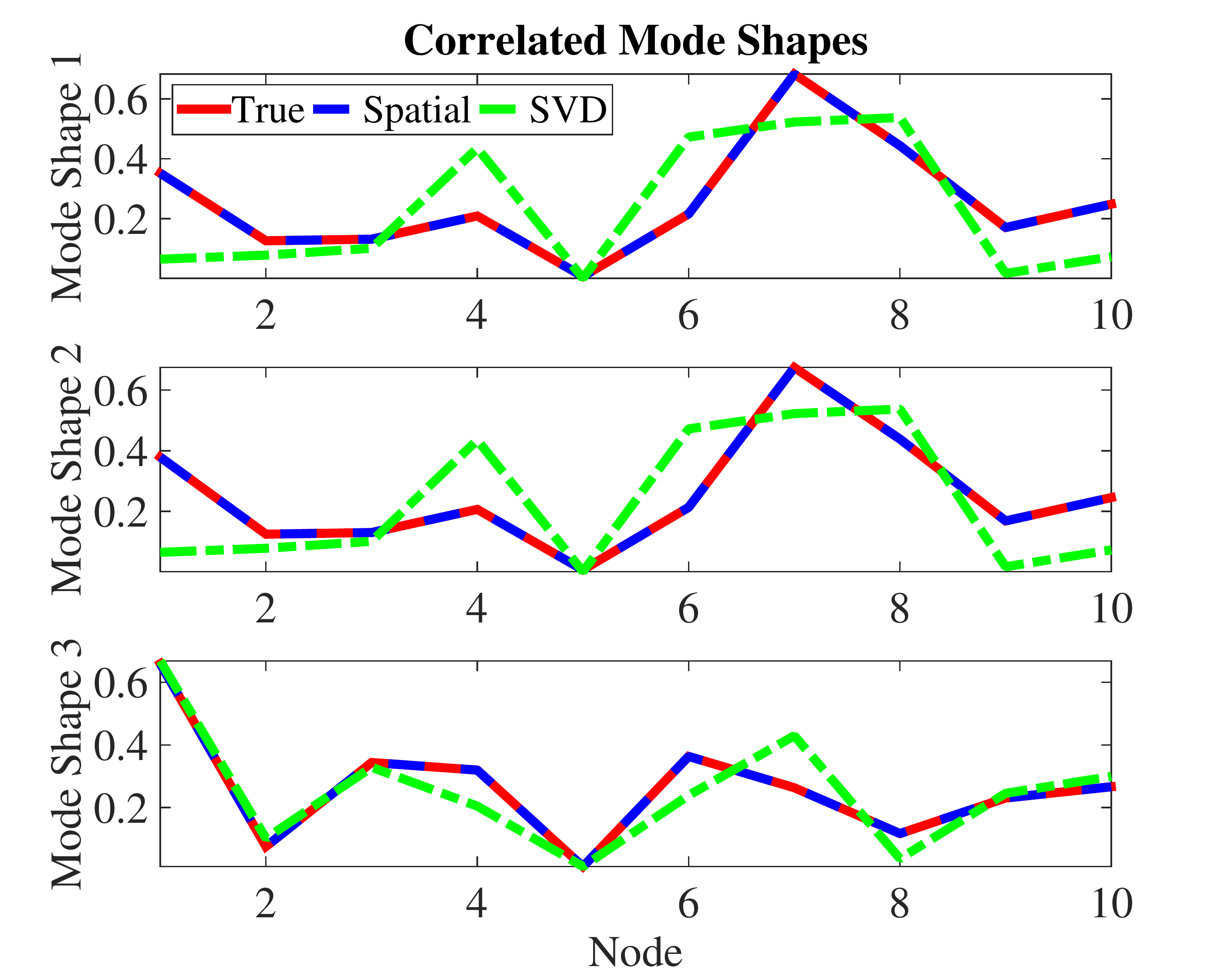}
\centerline{\footnotesize{(b)}}
\end{minipage}
\caption{Random spatial compression. (a) Estimated mode shapes by ANM (on compressed data) and SVD (on uncompressed data) with orthogonal mode shapes. (b) Estimated mode shapes with correlated mode shapes.}
\label{RandSpa}
\end{figure}

\subsection{Noisy data}
\label{noda}

Finally, we simulate the atomic norm denoising problem presented in Section~\ref{MMVAND}. We assume the entries of $\W$ are i.i.d.\ random variables from the distribution ${\cal{CN}}(0,\sigma^2)$. We set the signal-to-noise ratio (SNR) between $15$ and $50$ dB, which corresponds to values of $\sigma$ between $0.1754$ and $0.0031$. We perform 50 trials for each value of $\sigma$, with random mode shapes in each trial. The regularization parameter is set to $\lambda=\sigma\sqrt{4MN\log(M)}$ with $M=80$ and $N=10$. Other parameters are set the same as in Section~\ref{simu_AsyVsSyn}. We define $\|\Xh-\X^\star\|_F^2$ and $\frac{\|\Xh-\X^\star\|_F}{\|\X^\star\|_F}$ as the mean square error (MSE) and relative error, respectively.
It can be seen in Fig.~\ref{UniformNoise}(a) that the MSE is linearly correlated with $\sigma^2$, which is consistent with the theory in Section~\ref{MMVAND}. We also present Fig.~\ref{UniformNoise}(b) to illustrate how the relative error behaves with different noise levels.


 \begin{figure}[t]
\begin{minipage}{0.49\linewidth}
\centering
\includegraphics[width=2.4in]{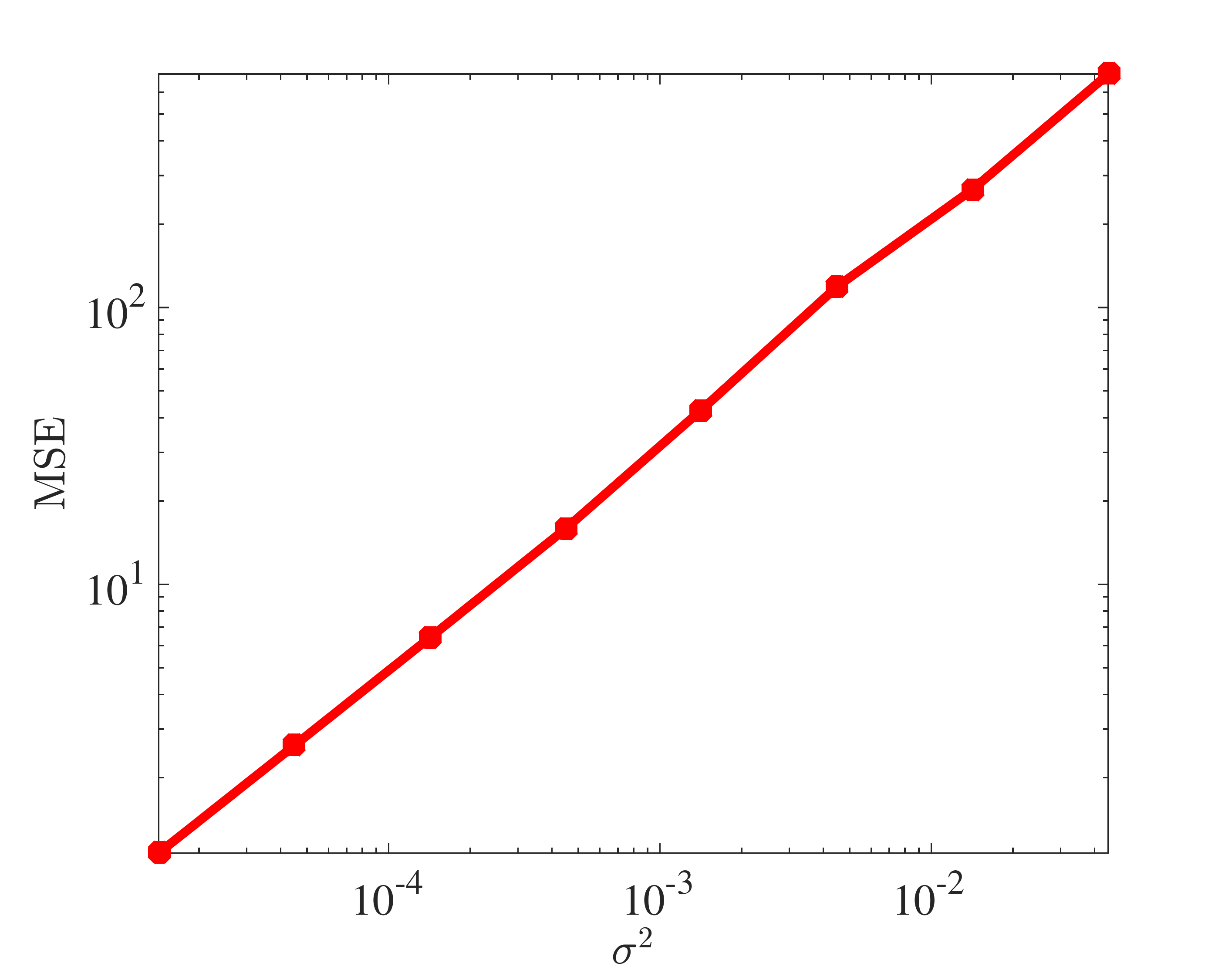}
\centerline{\footnotesize{(a)}}
\end{minipage}
\hfill
\begin{minipage}{0.49\linewidth}
\centering
\includegraphics[width=2.4in]{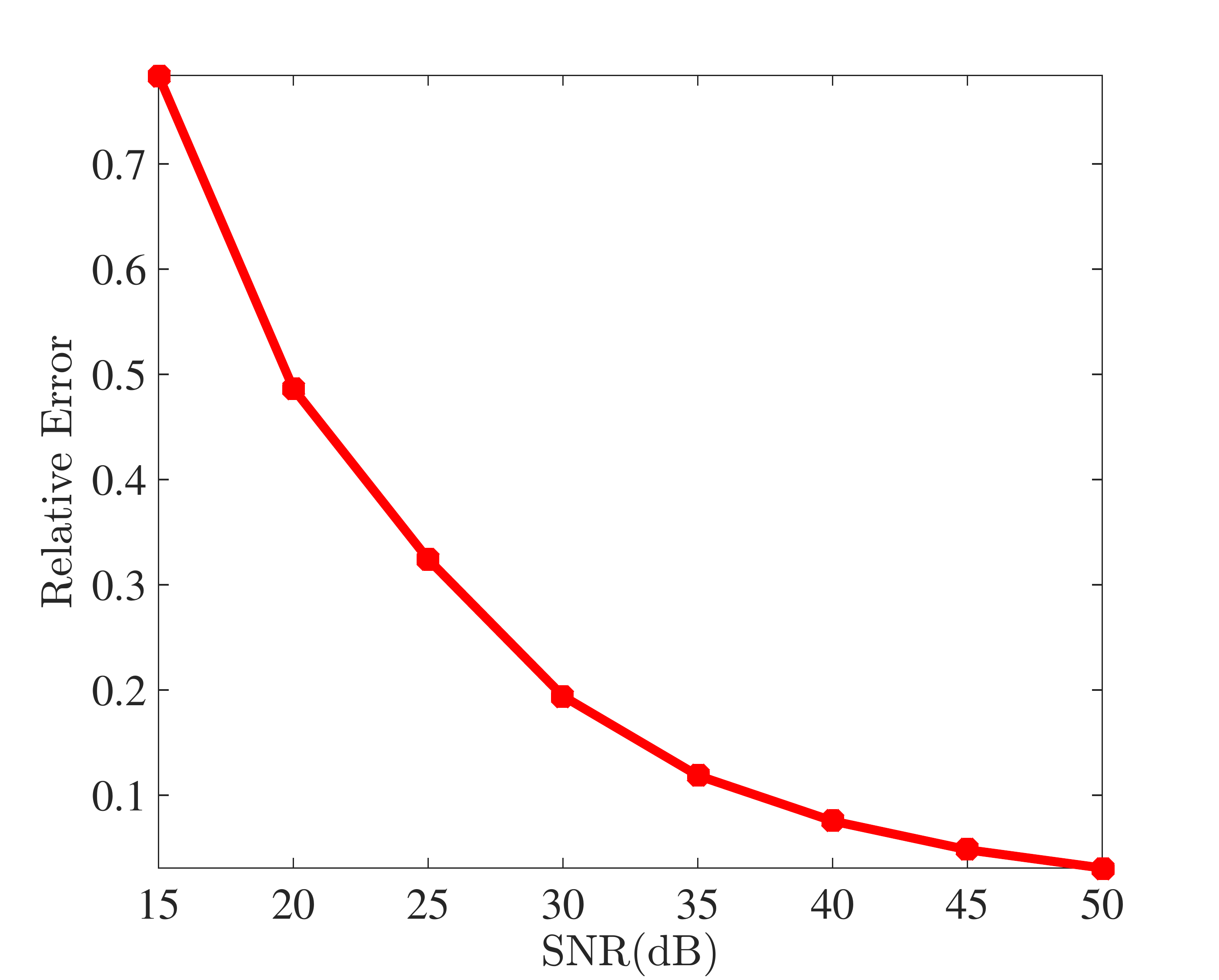}
\centerline{\footnotesize{(b)}}
\end{minipage}
\caption{Atomic norm denoising. (a) Recovered MSE with respect to different noise levels. (b) Relative error with respect to different SNR levels.}
\label{UniformNoise}
\end{figure}

\section{Proofs}
\label{sec:mainproofs}

In this section, we will prove Theorem~\ref{THMCS}, Theorem~\ref{ANDB}, and Corollary~\ref{EAND}. Some proof techniques are inspired by the work in \cite{Chand12, heckel2016generalized, bhaskar2013atomic, tang2015near}.

\subsection{Convex analysis}
\label{ca}

We first review some basic concepts from convex analysis \cite{boyd2004convex}. A set $\Lambda$ is a cone if
\begin{align*}
\theta \x \in \Lambda~ \change{for~any} \x\in \Lambda,
\end{align*}
where $\theta$ is a nonnegative constant. $\Lambda$ is a convex cone if
\begin{align*}
\theta_1 \x_1+\theta_2 \x_2 \in \Lambda
\end{align*}
holds for any $\x_1,~\x_2\in \Lambda$ and $\theta_1,~\theta_2\geq 0$.
The polar cone of a cone $\Lambda$ is
\begin{align*}
\Lambda^\circ \triangleq \{ \z: \langle \x,\z \rangle\leq0 ~\forall~\x\in \Lambda \}.
\end{align*}

The tangent cone and normal cone at $\X$ with respect to the scaled unit ball $\|\X\|_{\AA}\change{conv}(\AA)$ are defined as
\begin{align*}
\TT_{\AA}(\X)\triangleq\change{cone}\{\D:\|\X+\D\|_{\AA}\leq \|\X\|_{\AA}\},
\end{align*}
and
\begin{align*}
\NN_{\AA}(\X)\triangleq\{\S: \langle \S,\D \rangle \leq 0, \forall~ \D ~\change{s.t.} ~\|\X+\D\|_{\AA}\leq \|\X\|_{\AA}\},
\end{align*}
respectively. Note that the tangent cone $\TT_{\AA}(\X)$ is
the set of descent directions of the atomic norm $\| \cdot \|_{\AA}$ at $\X$.
The normal cone $\NN_{\AA}(\X)$ is the polar cone of the tangent cone $\TT_{\AA}(\X)$ and vice-versa.

Let $\Omega$ be a subset of $\SSS^{MN-1}$, where $\SSS^{MN-1}\triangleq\{\Z\in \CCC^{M\times N}:\|\Z\|_F=1 \}$ denotes the unit sphere. Then, the Gaussian width of $\Omega$ is
\begin{align}
\omega(\Omega) \triangleq \EEE_{\H}\left[ \max_{\Z\in\Omega}  \langle \H,\Z  \rangle \right], \label{GWid}
\end{align}
where $\H$ is a Gaussian matrix with i.i.d.\ entries from the distribution $\NN(0,1)$.

\subsection{Proof of Theorem \ref{THMCS}}
\label{proofTHMCS}

We start by showing that the number of measurements $M'$ needed for perfect recovery can be lower bounded with a Gaussian width in Section~\ref{MpG}. Then, we upper bound the Gaussian width with the expectation of the recovery error obtained from the atomic norm denoising problem~\eqref{AND} in Section~\ref{BGW}.

\subsubsection{Bounding $M'$ with a Gaussian width}
\label{MpG}

It can be shown that the ANM problem in (\ref{ANMOri}) is equivalent to the following optimization problem
\begin{align*}
&\min_{\D} \|\X^{\star}+\D\|_{\AA}\\
&~~\text{s.t.}~\bPhi_n\d_n=\zero,~~n=1,\ldots,N,\\
&~~~~~~~~~~~~\D = [\d_1,\d_2,\cdots,\d_N].
\end{align*}
In particular, the above equivalent optimization problem can be obtained by eliminating the equality constraints in (\ref{ANMOri}). Let $\X=\X^\star + \D$ with $\bPhi_n \d_n=\zero,~n=1,\ldots,N$. It can be seen that the two optimization problems are equivalent. 

Let $\bPhi=[\bPhi_1,\bPhi_2,\cdots,\bPhi_N]$. With a slight abuse of notation, we define
\begin{align*}
\change{null}(\bPhi)\triangleq\{\D\in\CCC^{M\times N}: ~ \bPhi_n\d_n=\zero,~~n=1,\ldots,N\}
\end{align*}
as a set of matrices with columns belonging to the null space of the corresponding sensing matrix. We also define a block diagonal matrix $\change{diag}(\Z)$ corresponding to a matrix $\Z=[\z_1,\z_2,\cdots,\z_N]\in \CCC^{M\times N}$ as
\begin{align}
\change{diag}(\Z)\triangleq\MAT{cccc} \z_1 & & &\\ & \z_2 & & \\ & & \ddots & \\ & & & \z_N \mat\in \CCC^{MN\times N}. \label{diag}
\end{align}
Inspired by \cite{Chand12}, we have the following proposition which gives us an optimal condition for exact recovery.
\begin{Proposition}
\label{TN}
$\Xh=\X^{\star}$ is the unique optimal solution to the ANM (\ref{ANMOri}) if and only if
\begin{align*}
\TT_{\AA}(\X^{\star}) \cap \change{null}(\bPhi)=\{\zero\}.
\end{align*}
\end{Proposition}

\begin{proof}
On one hand, if $\X^{\star}$ is the unique optimal solution, then $\|\X^{\star}+\D\|_{\AA}> \|\X^{\star}\|_{\AA}$ holds for all $\D\in \change{null}(\bPhi)/\{\zero\}$. It follows that $\D \notin \TT_{\AA}(\X^{\star})$. Thus, we can get $\TT_{\AA}(\X^{\star}) \cap \change{null}(\bPhi)=\{\zero\}.$

On the other hand, if $\TT_{\AA}(\X^{\star}) \cap \change{null}(\bPhi)=\{\zero\}$, then, for all $\D\in \change{null}(\bPhi)/\{\zero\}$, i.e., $\bPhi_n\d_n=\zero,~~n=1,\ldots,N,$ we have $\|\X^{\star}+\D\|_{\AA}> \|\X^{\star}\|_{\AA}$. Thus, $\X^{\star}$ is the unique optimal solution.
\end{proof}

Define $\Omega=\TT_{\AA}(\X^{\star}) \cap \SSS^{MN-1}$, where $\SSS^{MN-1}=\{\Z\in \CCC^{M\times N}:\|\Z\|_F=1 \}$. It can be seen that $\Omega$ is a subset of the unit sphere. According to Proposition \ref{TN}, to show that $\X^\star$ is the unique optimal solution to (\ref{ANMOri}), we hope to demonstrate that $\Z \notin \change{null}(\bPhi)$ holds for any $\Z=[\z_1,\z_2,\cdots,\z_N]\in \Omega$, i.e.,
\begin{align*}
\sum_{n=1}^N \|\bPhi_n\z_n\|_2>0,
\end{align*}
which will hold if
\begin{align*}
\sum_{n=1}^N \|\bPhi_n\z_n\|_2^2=\|[\bPhi_1\z_1~\bPhi_2\z_2~\cdots~\bPhi_N\z_N]\|_F^2=\|\bPhi \Zh\|_F^2>0.
\end{align*}
Therefore, we need to show
\begin{align*}
\|\bPhi \Zh\|_F>0,
\end{align*}
where $\bPhi=[\bPhi_1, \bPhi_2,\cdots, \bPhi_N]\in \RRR^{M'\times MN}$, and $\Zh= \change{diag}(\Z)\in \CCC^{MN\times N}$ as is defined in (\ref{diag}). Note that we have $\|\Zh\|_F=\|\Z\|_F=1$ since $\Z$ is in a subset of the unit sphere.

Next, we will show that $\|\bPhi\Zh\|_F>0$ holds for all $\Z\in \Omega=\TT_{\AA}(\X^{\star}) \cap \SSS^{MN-1}$ with high probability. It then follows that $\X^\star$ is the unique optimal solution of ANM (\ref{ANMOri}).

\begin{Theorem}
\label{EPhiZ}
Let $\bPhi=[\bPhi_1, \bPhi_2,\cdots, \bPhi_N]\in \RRR^{M'\times MN}$ be a random matrix with i.i.d  Gaussian entries which satisfy $\NN(0,1)$. Let $\Omega$ be a subset of $\SSS^{MN-1}$. For all $\Z=[\z_1,\z_2,\cdots,\z_N]\in \Omega$, define $\Zh= \change{diag}(\Z)$ as (\ref{diag}). Then, we have
\begin{align*}
\EEE \left[ \min_{\Z \in \Omega} \|\bPhi \Zh\|_F \right] \geq \lambda_{M'N}-\omega(\Omega),
\end{align*}
where $\lambda_{M'N}=\EEE_{\G\sim \NN(\zero,\I)}\|\G\|_F\geq \frac{\sqrt{N}M'}{\sqrt{M'+1}}$ and $\omega(\Omega)$ is the Gaussian width of $\Omega$ defined in (\ref{GWid}).
\end{Theorem}

\begin{Remark}
The above theorem is based on Gordon's work \cite{gordon1988milman}, and provides us a lower bound for the minimum gain of the operator $\bPhi$ restricted to a set $\Omega$. The proof of this theorem is given in Appendix \ref{ProofEPhiZ}.
\end{Remark}


\begin{Corollary}
\label{MpGauss}
Let $\bPhi=[\bPhi_1,\bPhi_2,\cdots,\bPhi_N]\in \RRR^{M'\times MN}$ be a random matrix with i.i.d. Gaussian entries which satisfy $\NN(0,1)$. Define $\Omega \triangleq \TT_{\AA}(\X^{\star}) \cap \SSS^{MN-1}$ as a subset of the unit sphere. Then, $\X^{\star}$ is the unique optimal solution of ANM (\ref{ANMOri}) with probability at least $1-e^{-\frac{N(M'-2)}{8}}$ if
\begin{align}
M' \geq \frac{4}{N}\omega^2(\Omega). \label{MpO}
\end{align}
\end{Corollary}

\begin{Remark}
Corollary \ref{MpGauss} is an immediate consequence of Theorem $\ref{EPhiZ}$. It can be seen that the number of measurements $M'$ needed for exact recovery can be lower bounded by a Gaussian width. The proof details (presented in Appendix \ref{ProofMpGauss}) are based on a Gaussian concentration inequality \cite{pisier1986probabilistic}.
\end{Remark}

\subsubsection{Bounding the Gaussian width}
\label{BGW}

This section is dedicated to finding an upper bound on $\omega(\TT_{\AA}(\X^{\star})\cap \SSS^{MN-1})$. First, we define the mean-square distance of a set $\CC$ as
\begin{align*}
\change{dist}(\CC)=\EEE_{\G\sim\NN(\zero,\I)}\left[ \min_{\D\in \CC} \|\G-\D\|_F^2 \right].
\end{align*}

Inspired by the work in \cite{Chand12} and \cite{heckel2016generalized}, we have
\begin{align*}
&\omega^2(\TT_{\AA}(\X^{\star})\cap \SSS^{MN-1})\\
\leq &\!\left[ \EEE_{\G}\!\!\left[ \min_{\D \in \TT_{\AA}^{\circ}(\X^{\star})}\! \|\G-\D\|_F \right]\! \right]^2\!\!
\!\leq \EEE_{\G}\!\!\left[ \min_{\D \in \TT_{\AA}^{\circ}(\X^{\star})}\! \|\G-\D\|_F \!\right]^2\\
=&\EEE_{\G}\!\!\left[ \min_{\D \in \TT_{\AA}^{\circ}(\X^{\star})}\! \|\G-\D\|_F^2 \right]
\!\!=\!\change{dist}(\TT_{\AA}^{\circ}(\X^{\star})\!)\!
\\
=&\change{dist}(\change{cone}(\partial \|\X^{\star}\|_{\AA}))
\leq \change{dist}(\lambda \partial \|\X^{\star}\|_{\AA})\\
=&\max_{\sigma>0} \frac{\EEE_{\G}[ \| \Xh(\X^{\star}+\sigma \G,\sigma\lambda)-\X^{\star} \|_F^2 ]}{\sigma^2},
\end{align*}
where the first two inequalities follow from Proposition 3.6 in \cite{Chand12} and Jensen's inequality, respectively. 
The third equality and last inequality can be found in equation (67) in \cite{heckel2016generalized} while the last equality comes from Theorem 1.1 in \cite{oymak2016sharp}.
Note that $\TT_{\AA}^{\circ}$ is the polar cone of $\TT_{\AA}$, as  defined at the beginning of Section~\ref{ca}.
Here, $\Xh(\Y,\lambda)$ is the solution to the MMV atomic norm denoising problem (\ref{AND}) with $\lambda$ being the regularization parameter.
With the inequality given in (\ref{EG}), we have
\begin{align*}
\frac{\EEE_{\G}[ \| \Xh(\X^{\star}+\sigma \G,\sigma\lambda)-\X^{\star} \|_F^2 ]}{\sigma^2} \leq CKN\log(M),
\end{align*}
which implies
\begin{align}
\omega^2(\TT_{\AA}(\X^{\star})\cap \SSS^{MN-1})&\leq CKN\log(M). \label{OMG}
\end{align}
Thus, plugging (\ref{OMG}) into (\ref{MpO}), we can get (\ref{Mp}) and finish the proof of Theorem \ref{THMCS}. \qed

\subsection{Proof of Theorem \ref{ANDB}}
\label{ProofANDB}

In this section, we prove Theorem \ref{ANDB} by extending the results in \cite{bhaskar2013atomic} and  \cite{tang2015near}  to the MMV case. For the SMV case, it is shown in \cite{bhaskar2013atomic} that a good choice of the regularization parameter can achieve accelerated convergence rates. Inspired by their choice of the regularization parameter, we use
\begin{align*}
\lambda \approx \eta \EEE\|\W\|_{\AA}^*
\end{align*}
in the MMV atomic norm denoising problem (\ref{AND}). Here, $\eta\in (1,\infty)$ is some constant which ultimately must be set large enough to enable the proof of Lemma~\ref{FI2}, and $\W$ is a complex Gaussian matrix. To set $\lambda$, we need to find an upper bound for $\EEE\|\W\|_{\AA}^*$.

\subsubsection{Bounding $\EEE\|\W\|_{\AA}^*$}

\begin{Lemma}
Let $\W\in\CCC^{M\times N}$ be a random matrix with i.i.d.\ complex Gaussian entries from the distribution $\CC\NN(0,\sigma^2)$. Then, there exists a numerical constant $C$ such that
\begin{align*}
\EEE\|\W\|_{\AA}^*\leq C\sigma \sqrt{MN\log(M)}.
\end{align*}
\end{Lemma}

\begin{proof}
According to the definition of the dual atomic norm in (\ref{dualnorm}), we have
\begin{align*}
(\|\W\|_{\AA}^*)^2&=\sup_{f\in [0,1)} \|\W^*\a(f)\|_2^2=\sup_{f\in [0,1)}\sum_{n=1}^N\left| \sum_{m=0}^{M-1} \W_{mn}^*e^{j2\pi f m} \right|^2\\
&=\sup_{f\in [0,1)}\sum_{n=1}^N \sum_{m,p=0}^{M-1} \W_{mn}^*\W_{pn}e^{j2\pi f (m-p)}\\
&=\sup_{f\in [0,1)} \WW_M\left(e^{j 2\pi f} \right),
\end{align*}
where the polynomial $\WW_M$ is defined as
\begin{align*}
\WW_M\left(e^{j 2\pi f} \right)\triangleq \sum_{n=1}^N \sum_{m,p=0}^{M-1} \W_{mn}^*\W_{pn}e^{j2\pi f (m-p)}.
\end{align*}

For all $f_1,f_2\in[0,1)$ we have
\begin{align*}
\WW_M\left(e^{j 2\pi f_1} \right)- \WW_M\left(e^{j 2\pi f_2} \right)&\leq \left| e^{j2\pi f_1}- e^{j2\pi f_2} \right| \sup_{f\in [0,1)} \WW_M'\left(e^{j 2\pi f} \right)\\
&\leq 2\pi M|f_1-f_2| \sup_{f\in [0,1)} \WW_M\left(e^{j 2\pi f} \right)
\end{align*}
by using the mean value theorem and Bernstein's inequality for polynomials \cite{schaeffer1941inequalities}.
By letting $f_2$ take any of the $L$ values $0,\frac{1}{L},\ldots,\frac{L-1}{L}$, we have
\begin{align*}
\sup_{f\in[0,1)} \WW_M\left(e^{j 2\pi f} \right) \leq \max_{l=0,\ldots,L-1} \WW_M\left(e^{j 2\pi l/L} \right)+\frac{2\pi M}{L} \sup_{f\in [0,1)} \WW_M\left(e^{j 2\pi f} \right).
\end{align*}
It follows that
\begin{align*}
(\|\W\|_{\AA}^*)^2=\sup_{f\in[0,1)} \WW_M\left(e^{j 2\pi f} \right) &\leq \left(1-\frac{2 \pi M}{L}\right)^{-1}\max_{l=0,\ldots,L-1} \WW_M\left(e^{j 2\pi l/L} \right)\\
&\leq \left(1+\frac{4 \pi M}{L}\right)\max_{l=0,\ldots,L-1} \WW_M\left(e^{j 2\pi l/L} \right),
\end{align*}
where the last inequality holds if $L\geq 4\pi M$.
Then, we get
\begin{align}
\|\W\|_{\AA}^*\leq \left(1+\frac{4 \pi M}{L}\right)^{\frac{1}{2}}\left[\max_{l=0,\ldots,L-1} \WW_M\left(e^{j 2\pi l/L} \right)\right]^{\frac{1}{2}}
\label{WDN}
\end{align}
and
\begin{align*}
\EEE\|\W\|_{\AA}^*\leq \left(1+\frac{4 \pi M}{L}\right)^{\frac{1}{2}}\left[ \EEE\max_{l=0,\ldots,L-1} \WW_M\left(e^{j 2\pi l/L} \right)\right]^{\frac{1}{2}}.
\end{align*}

Define
\begin{align*}
x_l\triangleq\frac{1}{\sigma\sqrt{M}}\sum_{m=0}^{M-1}\W_{mn}^*e^{j 2\pi lm/L}
\end{align*}
for $l=0,1,\ldots L-1$. It can be seen that $x_l$ is complex Gaussian variable satisfying $\CC\NN(0,1)$ since  the entries of $\W$ satisfy $\CC\NN(0,\sigma^2)$. Then, we have
\begin{align*}
\EEE\left[\max_{l=0,\ldots,L-1} \WW_M\left(e^{j 2\pi l/L} \right)\right]
&=\EEE\left[\max_{l=0,\ldots,L-1} \sum_{n=1}^N\left| \sum_{m=0}^{M-1} \W_{mn}^*e^{j2\pi l m/L} \right|^2\right]\\
&\leq\sum_{n=1}^N \EEE\left[\max_{l=0,\ldots,L-1} \left| \sum_{m=0}^{M-1} \W_{mn}^*e^{j2\pi l m/L} \right|^2\right]\\
&=\frac{1}{2}\sigma^2M\sum_{n=1}^N \EEE\left[\max_{l=0,\ldots,L-1} 2\left| x_l \right|^2\right]\\
&\leq \sigma^2MN(\log L+1),
\end{align*}
where the last inequality uses the result that $\EEE \left[\max_{l=0,\ldots,L-1} 2\left| x_l \right|^2\right]\leq 2\log(L)+2$, see Lemma 5 in \cite{bhaskar2013atomic}.

Choosing $L=4\pi M\log (M)$ gives
\begin{align*}
\EEE\|\W\|_{\AA}^*&\leq \sigma\left(1+\frac{1}{\log(M)}\right)^{\frac{1}{2}}\sqrt{MN[\log(M)+\log(4\pi\log(M))+1]}\\
&\leq C\sigma \sqrt{MN\log(M)}.
\end{align*}
Note that $C$ is a constant which belongs to $(1,2)$ when $M$ is large.
\end{proof}

\subsubsection{Bounding $\|\Xh-\X^{\star}\|_F^2$}
\label{XhmXs}

Now, we can set the regularizing parameter in the MMV atomic norm denoising problem (\ref{AND}) as $\lambda=\eta \sigma \sqrt{4MN\log(M)}$  for some $\eta \in (1,\infty)$. Then
\begin{align}
\|\W\|_{\AA}^* \leq \frac{\lambda}{\eta}
\label{EW}
\end{align}
holds with high probability. In particular, we have the following lemma.
\begin{Lemma}
\label{PW}
\begin{align*}
\PPP\left[ \|\W\|_{\AA}^*\geq \frac{\lambda}{\eta} \right]\leq \frac{1}{M^2}.
\end{align*}
\end{Lemma}

\begin{proof}
As is shown in (\ref{WDN}), letting $L=M$, we have
\begin{align*}
\|\W\|_{\AA}^* &\leq \left(1+4\pi \right)^{\frac{1}{2}}\left[\max_{l=0,\ldots,M-1} \WW_M\left(e^{j 2\pi l/M} \right)\right]^{\frac{1}{2}}\\
&= \left(1+4\pi \right)^{\frac{1}{2}}\left[\max_{l=0,\ldots,M-1} \sum_{n=1}^N \left| \sum_{m=0}^{M-1} \W_{mn}^* e^{j2\pi m l /M} \right|^2 \right]^{\frac{1}{2}}\\
&\leq \left(1+4\pi \right)^{\frac{1}{2}}\left[\max_{l=0,\ldots,M-1} \left(  \sum_{n=1}^N \left|\sum_{m=0}^{M-1} \W_{mn}^* e^{j2\pi m l /M} \right|\right)^2 \right]^{\frac{1}{2}}\\
&=\left(1+4\pi \right)^{\frac{1}{2}} \max_{l=0,\ldots,M-1}  \sum_{n=1}^N \left| \sum_{m=0}^{M-1}\W_{mn}^* e^{j2\pi m l /M} \right|\\
&\triangleq \left(1+4\pi \right)^{\frac{1}{2}} \max_{l=0,\ldots,M-1}  Z_l.
\end{align*}
It can be seen that $Z_l$ is stochastically upper bounded by a zero-mean Gaussian random variable with variance $\sigma^2 MN$ \cite{heckel2016generalized}. Then, for any $\beta >\frac{1}{\sqrt{2\pi}}$, we have
\begin{align*}
\PPP[|Z_m|\geq \sigma \sqrt{MN}\beta]\leq 2 e^{-\beta^2}.
\end{align*}
Let $\beta=\sqrt{4\log(M)}$, we can get
\begin{align*}
\PPP\left[ \|\W\|_{\AA}^*\geq \frac{\lambda}{\eta} \right]= \PPP\left[ \|\W\|_{\AA}^*\geq \sigma \sqrt{4MN\log{M}} \right]\leq 2Me^{-4\log(M)}\leq \frac{1}{M^2},
\end{align*}
where the first inequality comes from the union bound.
\end{proof}


The following lemma derived from convex analysis provides optimality conditions for $\Xh$ to be the solution of (\ref{AND}).

\begin{Lemma}
(Optimality Conditions): $\Xh$ is the solution of (\ref{AND}) if and only if
\begin{enumerate}
\item $\|\Y-\Xh\|_{\AA}^*\leq \lambda$,
\item $\langle  \Y-\Xh,\Xh \rangle=\lambda\|\Xh\|_{\AA}$.
\end{enumerate}
\label{Lemma_optcon}
\end{Lemma}

Define an atomic measure as
\begin{align*}
\bmu(f)=\sum_{k=1}^K |A_k| \b_k\delta (f-f_k)
\end{align*}
with $f\in [0,1),~\|\b_k\|_2=1$. Then, we have
\begin{align*}
\X^{\star}=\sum_{k=1}^K  |A_k| \a(f_k)\b_k^*=\int_0^1\a(f)\bmu^*(f)df.
\end{align*}
Similarly, the recovered data $\Xh$ can be represented as
\begin{align*}
\Xh=\int_0^1\a(f)\hat{\bmu}^*(f)df
\end{align*}
for some measure $\hat{\bmu}(f)$.
Define the difference measure as $\bnu=\hat{\bmu}-\bmu$. Then, the error matrix is given as
\begin{align*}
\E\triangleq\Xh-\X^{\star}=\int_0^1\a(f)\bnu^*(f)df.
\end{align*}
It follows that
\begin{equation}
\begin{aligned}
\|\E\|_F^2&=|\langle \E,\E \rangle|=\left|\left\langle \E,\int_0^1\a(f)\bnu^*(f)df \right\rangle\right|\\
&=\left| \int_0^1 \left\langle \E,\a(f)\bnu^*(f)\right\rangle df\right|\\
&= \left| \int_0^1\a^*(f) \E \bnu(f) df\right|\\
&=  \left| \int_0^1\bxi^*(f)  \bnu(f) df\right|\\
&\leq  \left| \int_F \bxi^*(f)  \bnu(f) df\right|+\sum_{k=1}^K  \left| \int_{N_k} \bxi^*(f)  \bnu(f) df\right| \label{Eb0}
\end{aligned}
\end{equation}
where $\bxi(f)\triangleq  \E^*\a(f) $ is defined as a vector-valued error function. Here,
\begin{equation}
\begin{aligned}
N_k&\triangleq\{f:d(f,f_k)\leq 0.16/M\},~k = 1, \ldots, K\\
F&\triangleq[0,1) /\cup_{k=1}^K N_k  \label{NFregion}
\end{aligned}
\end{equation}
are defined as the $k$th near region corresponding to $f_k$ and the far region, respectively.

With a little abuse of notation, we define
\begin{align}
\|\bxi(f)\|_{2,\infty}\triangleq\sup_{f\in[0,1)} \|\bxi(f)\|_2. \label{xi}
\end{align}
It turns out that
\begin{align*}
\|\bxi(f)\|_{2,\infty}&= \sup_{f\in[0,1)}  \|\E^*\a(f)\|_2\\
&=\sup_{f\in[0,1)}  \|(\Xh-\Y)^*\a(f)+\W^*\a(f)\|_2\\
&\leq \sup_{f\in[0,1)}  \|(\Y-\Xh)^*\a(f)\|_2+ \sup_{f\in[0,1)}  \|\W^*\a(f)\|_2\\
&=\|\Y-\Xh\|_{\AA}^*+\|\W\|_{\AA}^*\\
& \leq 2\lambda
\end{align*}
if the bound condition in (\ref{EW}) holds. The last inequality also follows from the first optimality condition in Lemma \ref{Lemma_optcon}.

Next, we extend Lemmas 1, 2 and 3 in \cite{tang2015near} to our MMV case and then bound the energy of the error matrix.
\begin{Lemma}
\label{UBE}
Since each entry of the vector-valued error function $\bxi(f)= \E^*\a(f)$ is an order-$M$ trigonometric polynomial, we have
\begin{align}
\|\E\|_F^2\leq \|\bxi(f)\|_{2,\infty}\left[  \int_F \|\bnu(f)\|_2df+I_0+I_1+I_2  \right] \label{EI}
\end{align}
with
\begin{align*}
&I_0^k=\left\|  \int_{N_k} \bnu(f)df\right\|_2,\\
&I_1^k=M\left\|  \int_{N_k}(f-f_k) \bnu(f)df\right\|_2,\\
&I_2^k=\frac{M^2}{2}\int_{N_k} (f-f_k)^2\|\nu(f)\|_2df\\
&I_l=\sum_{k=1}^K I_l^k, ~\change{for}~ l=0,1,2.
\end{align*}
\end{Lemma}
The proof of Lemma \ref{UBE} is given in  Appendix \ref{proofUBE}.

\begin{Lemma}
\label{I0I1}
There exist some numerical constants $C_0$ and $C_1$ such that
\begin{align*}
I_0\leq C_0\left( \frac{K\lambda}{M} +I_2+\int_F \|\bnu(f)\|_2 df \right),\\
I_1\leq C_1\left( \frac{K\lambda}{M} +I_2+\int_F \|\bnu(f)\|_2 df \right).
\end{align*}
\end{Lemma}
The proof of Lemma \ref{I0I1} is given in Appendix \ref{proofI0I1}.

\begin{Lemma}
\label{FI2}
For some sufficiently large $\eta>1$, set the regularizing parameter as $\lambda =\eta \sigma\sqrt{MN\log(M)}$. Then there exists a numerical constant $C$ such that
\begin{align*}
\int_F \|\bnu(f)\|_2df +I_2 \leq \frac{CK\lambda}{M}
\end{align*}
holds if $\|\W\|_{\AA}^*\leq \frac{\lambda}{\eta}$.
\end{Lemma}
The proof of Lemma \ref{FI2} is given in Appendix \ref{proofFI2}.

Using the above three lemmas, we have that
\begin{align*}
\| \Xh-\X^\star \|_F^2 \leq \frac{CK\lambda^2}{M}=C\sigma^2 KN \log(M)
\end{align*}
holds if $\eta$ is sufficiently large and if $\|\W\|_{\AA}^*\leq \frac{\lambda}{\eta}$. It follows from Lemma~\ref{PW} that the above recovery error bound holds with probability at least $1-\frac{1}{M^2}$. Thus, we finish the proof of Theorem \ref{ANDB}. \qed

\subsection{Proof of Corollary \ref{EAND}}
\label{ProofEAND}

As mentioned in the previous sections,
once we set $\lambda =\eta \sigma\sqrt{4MN\log(M)}$, then
\begin{align*}
\| \W \|_{\AA}^* \leq \frac{\lambda}{\eta}
\end{align*}
holds with probability at least $1-\frac{1}{M^2}$, which implies that the error bound in (\ref{Ebound}) holds with probability at least $1-\frac{1}{M^2}$.

Note that
\begin{align}
\EEE[\|\E\|_F^2]=\EEE\left[\|\E\|_F^2  \one \left\{ \|\W\|_{\AA}^*< \frac{\lambda}{\eta} \right\}\right]+\EEE\left[\|\E\|_F^2  \one \left\{ \|\W\|_{\AA}^*\geq \frac{\lambda}{\eta} \right\}\right],
\label{eq:corineq12}
\end{align}
where $\one$ is an indicator function.
Theorem \ref{ANDB} implies that
\begin{align}
\EEE\left[\|\E\|_F^2  \one \left\{ \|\W\|_{\AA}^*< \frac{\lambda}{\eta} \right\}\right]\leq C\sigma^2KN\log(M),
\label{eq:corineq1}
\end{align}
which can be proved with the definition of expectation. Also, by the Cauchy-Schwarz inequality, we have
\begin{align}
\EEE\left[\|\E\|_F^2  \one \left\{ \|\W\|_{\AA}^*\geq \frac{\lambda}{\eta} \right\}\right] &\leq \sqrt{\EEE[\| \E \|_F^4]}\sqrt{\PPP\left[ \|\W\|_{\AA}^*\geq \frac{\lambda}{\eta} \right]}
\leq C\sigma^2 N \log (M),
\label{eq:corineq2}
\end{align}
where the last inequality follows from the following lemma and Lemma \ref{PW}.

\begin{Lemma}
\label{EF4}
\begin{align}
\sqrt{\EEE\left[\|\E\|_F^4\right]}\leq C\sigma^2 MN \log(M). \label{SEF4}
\end{align}
\end{Lemma}


Combining~\eqref{eq:corineq12}, \eqref{eq:corineq1}, and~\eqref{eq:corineq2} completes the proof of Corollary $\ref{EAND}$. Now, we are left with proving Lemma~\ref{EF4}.


\begin{proof}
To prove Lemma \ref{EF4}, we first prove the following lemma.

\begin{Lemma}
\label{EF2}
There exists some numerical constant $C$ such that the energy of the error matrix $\E$ can be upper bounded with
\begin{align*}
\|\E\|_F^2\leq 8\|\W\|_F^2+C\lambda^2 \frac{K}{M}=8\|\W\|_F^2+C\eta^2\sigma^2KN\log(M).
\end{align*}
\end{Lemma}

\begin{proof}
The optimality of $\Xh$ implies
\begin{align*}
\frac{1}{2}\| \Y-\Xh \|_F^2+\lambda \|\hat{\bmu}\|_{2,TV}\leq \frac{1}{2}\| \Y-\X^\star \|_F^2+\lambda \|\bmu\|_{2,TV},
\end{align*}
which is equivalent to
\begin{align}
\frac{1}{2}\| \E \|_F^2+\lambda (\| \PP_{\FF}(\hat{\bmu})\|_{2,TV}+\| \PP_{\FF^c}(\hat{\bmu})\|_{2,TV})\leq \langle \E,\W \rangle+\lambda \|\bmu\|_{2,TV}. \label{E1}
\end{align}

Since $\hat{\bmu}=\bnu+\bmu$, we have
\begin{align*}
\| \PP_{\FF}(\hat{\bmu})\|_{2,TV}&=\| \PP_{\FF}(\bnu+\bmu)\|_{2,TV}\\
&= \| \bmu\|_{2,TV}+\| \PP_{\FF}(\bnu)\|_{2,TV}\\
&= \| \bmu\|_{2,TV}+ \| \bnu\|_{2,TV}-\| \PP_{\FF^c}(\bnu)\|_{2,TV}\\
&= \| \bmu\|_{2,TV}+ \| \bnu\|_{2,TV}-\| \PP_{\FF^c}(\hat{\bmu})\|_{2,TV}.
\end{align*}
Then, the inequality (\ref{E1}) becomes
\begin{align*}
\frac{1}{2}\| \E \|_F^2+\lambda \| \bnu\|_{2,TV}\leq \langle \E,\W \rangle,
\end{align*}
which implies
\begin{equation}
\begin{aligned}
\frac{1}{2}\| \E \|_F^2&\leq \langle \E,\W \rangle-\lambda \| \bnu\|_{2,TV}\\
&=\langle \E,\W \rangle-\lambda \int_0^1 \bnu^*(f)\QQ(f)df\\
&=\langle \E,\W \rangle-\lambda \langle \E,\Q \rangle. \label{EW1}
\end{aligned}
\end{equation}
We have used  Parseval's theorem in  the last equality. Moreover,
\begin{equation}
\begin{aligned}
-\lambda \langle \E,\Q \rangle&\leq \lambda \|\E\|_F \|\Q\|_F= \left(\frac{\|\E\|_F}{2}\right)(2\lambda \|\Q\|_F)\\
&\leq \frac{1}{8} \|\E\|_F^2+2\lambda^2 \|\Q\|_F^2. \label{EW2}
\end{aligned}
\end{equation}
Combining (\ref{EW1}) and (\ref{EW2}) and using the Cauchy-Schwarz inequality, we get
\begin{align*}
\frac{3}{8}\| \E \|_F^2&\leq  \left(\frac{\|\E\|_F}{2}\right)(2 \|\W\|_F)+2\lambda^2 \|\Q\|_F^2\\
&\leq \frac{1}{8}\| \E \|_F^2+2 \|\W\|_F^2+2\lambda^2 \|\Q\|_F^2.
\end{align*}
As a consequence, we have
\begin{align*}
\| \E \|_F^2 &\leq 8\|\W\|_F^2+8\lambda^2 \|\Q\|_F^2\\
&=8\|\W\|_F^2+8\lambda^2 \langle \QQ(f),\QQ(f) \rangle\\
&\leq 8\|\W\|_F^2+8\lambda^2 \|\QQ(f)\|_1  \|\QQ(f)\|_{\infty} \\
&\leq 8\|\W\|_F^2+C\lambda^2\frac{K}{M},
\end{align*}
where the last inequality follows from $\|\QQ(f)\|_1\leq C\frac{K}{M}$ as is shown in (\ref{Qf1}).
This completes the proof of Lemma \ref{EF2}.
\end{proof}

Now, continuing the proof of Lemma \ref{EF4}, note that $\| \W/\sigma \|_F^2$ is a $\chi^2$ distribution with $MN$ degrees of freedom. It follows that
\begin{align*}
\EEE[\| \W \|_F^2]&=\sigma^2 MN,\\
\EEE[\| \W \|_F^4]&=\sigma^4 MN(MN+2)\leq C\sigma^4M^2N^2.
\end{align*}
Then, with Lemma \ref{EF2}, we have
\begin{align*}
\EEE[\| \E \|_F^4]&\leq 64\EEE[\| \W \|_F^4]+C\EEE[\| \W \|_F^2]\eta^2\sigma^2KN\log(M)+C\eta^4\sigma^4K^2N^2\log^2(M)\\
&\leq C\sigma^4M^2N^2+C\eta^2\sigma^4MKN^2\log(M)+C\eta^4\sigma^4K^2N^2\log^2(M)\\
&\leq C\sigma^4M^2N^2\log^2(M).
\end{align*}
By taking the square root, we obtain (\ref{SEF4}). This completes the proof of Lemma \ref{EF4}.
\end{proof}


%
%


\section{Conclusion}
\label{conc}

ANM can be used for modal analysis under a variety of different spatio-temporal sampling and compression schemes. In the noiseless case, theoretical analysis shows that perfect recovery of mode shapes and frequencies is possible under certain conditions. In particular, a minimum separation condition is required for all of the theorems in Section~\ref{sec:mans}: to distinguish any closely spaced pair of frequencies, it is necessary to observe the signals over a long time span, regardless of whether or how the samples are compressed.

While compression does not allow the observation time span to be shortened, it does allow the number of samples to be reduced. Using synchronous random sampling, asynchronous random sampling, and random temporal compression, for example, exact recovery of mode shapes and frequencies is possible when the average number of samples per sensor is roughly proportional to the number of active modes. Using random spatial compression, exact recovery is possible when the total number of compressed measurements scales with the number of degrees of freedom.

Currently, the theoretical results for random sampling and for random spatial compression require randomness assumptions on the mode shapes and, for random sampling, indicate that the performance worsens as the number of sensors increases. Removing and improving these aspects of the results would be worthy of further study. Another open question is to theoretically characterize the performance improvements of asynchronous random sampling over synchronous random sampling.

Damping and external forces may be present in practical scenarios. In this work, we have focused on developing a theoretical foundation for the idealized model of free vibration without damping. Extending the ANM framework and analysis to accommodate damping and forced inputs are interesting questions for future work. We do note, however, that although our theory has focused on systems with no damping, ANM can empirically work well even on systems with slight proportional damping. To show this, we repeat the uniform sampling experiment on the 6-degree-of-freedom system but with damping. For simplicity, we only consider proportional damping of the form $\C = \alpha\M+\beta\K$ with $\alpha = 0.01$ and $\beta = 0.006$. Here, we change the stiffness values to $k_1 = k_7 = 500, ~k_2 = 300, ~k_3 = 100,~ k_4 = 50, ~k_5 = 400, ~k_6 = 200$ N/m in order to slightly increase the minimum separation of true frequencies. The other parameters are the same as in Section~\ref{uniform}. Similarly, the true mode shapes and undamped natural frequencies of this system are obtained from the (normalized) generalized eigenvectors and square root of the generalized eigenvalues of $\K$ and $\M$. In particular, the true undamped natural frequencies are $F_1 = 0.6239,~ F_2 = 0.9820,~ F_3 = 1.7161, ~F_4 = 1.9365, ~F_5 = 2.3272, ~F_6 = 4.6879$ Hz. The damping ratios of this system are $\xi_1 = 0.0130,~\xi_2 = 0.0193, ~\xi_3 = 0.0328,~ \xi_4 = 0.0369,~ \xi_5 = 0.0442,~ \xi_6 = 0.0885$, and the true damped natural frequencies are given as $F_{d_k}=F_k \sqrt{1-\xi_k^2}$. In particular, $F_{d_1} = 0.6239,~ F_{d_2} = 0.9819,~ F_{d_3} = 1.7151, ~F_{d_4} = 1.9351, ~F_{d_5} = 2.3250, ~F_{d_6} = 4.6595$ Hz. The true  amplitudes are set as $A_1 = 1,~A_2 = 0.85,~A_3 = 0.7,~A_4 = 0.5,~A_5 = 0.25,~A_6 = 0.7$.
In order to make the experiment as close to realistic as possible, we first collect $300$ real-valued uniform samples from this system with sampling interval $T_s = 0.5/F_c$, where $F_c=\max_{1\leq k \leq 6} |F_k|$. We then compute the Hilbert transform of these real-valued samples to obtain the analytic samples. To eliminate border effects, we perform ANM only on the first $100$ analytical samples.

The damped real-valued uniform samples and the reconstructed mode shapes are presented in Figure~\ref{BoxcarDamping}. Although the input signals contain damping, we use the conventional undamped version of ANM to estimate the frequencies and the mode shapes. It can be seen that ANM still performs much better than SVD in general since the true mode shapes are not mutually orthogonal. In particular, the MAC for ANM is (0.9996, 0.9986, 0.9825, 0.9129, 0.8452, 0.9579), while the MAC for SVD is (0.9882, 0.9380, 0.8726, 0.7288, 0.9435, 0.7547). In addition, the estimated frequencies are $\hat{F}_1 = 0.6245,~ \hat{F}_2 = 0.9824,~ \hat{F}_3 = 1.7205, ~\hat{F}_4 = 1.9218, ~\hat{F}_5 = 2.3122, ~\hat{F}_6 = 4.5608$ Hz. Although the ANM algorithm returns reasonably accurate estimates, because it does not explicitly account for damping, it cannot perfectly recover the mode shapes and frequencies. We leave the development of such a damped ANM algorithm to future work.


\begin{figure}[t]
\begin{minipage}{0.49\linewidth}
\centering
\includegraphics[width=2.4in]{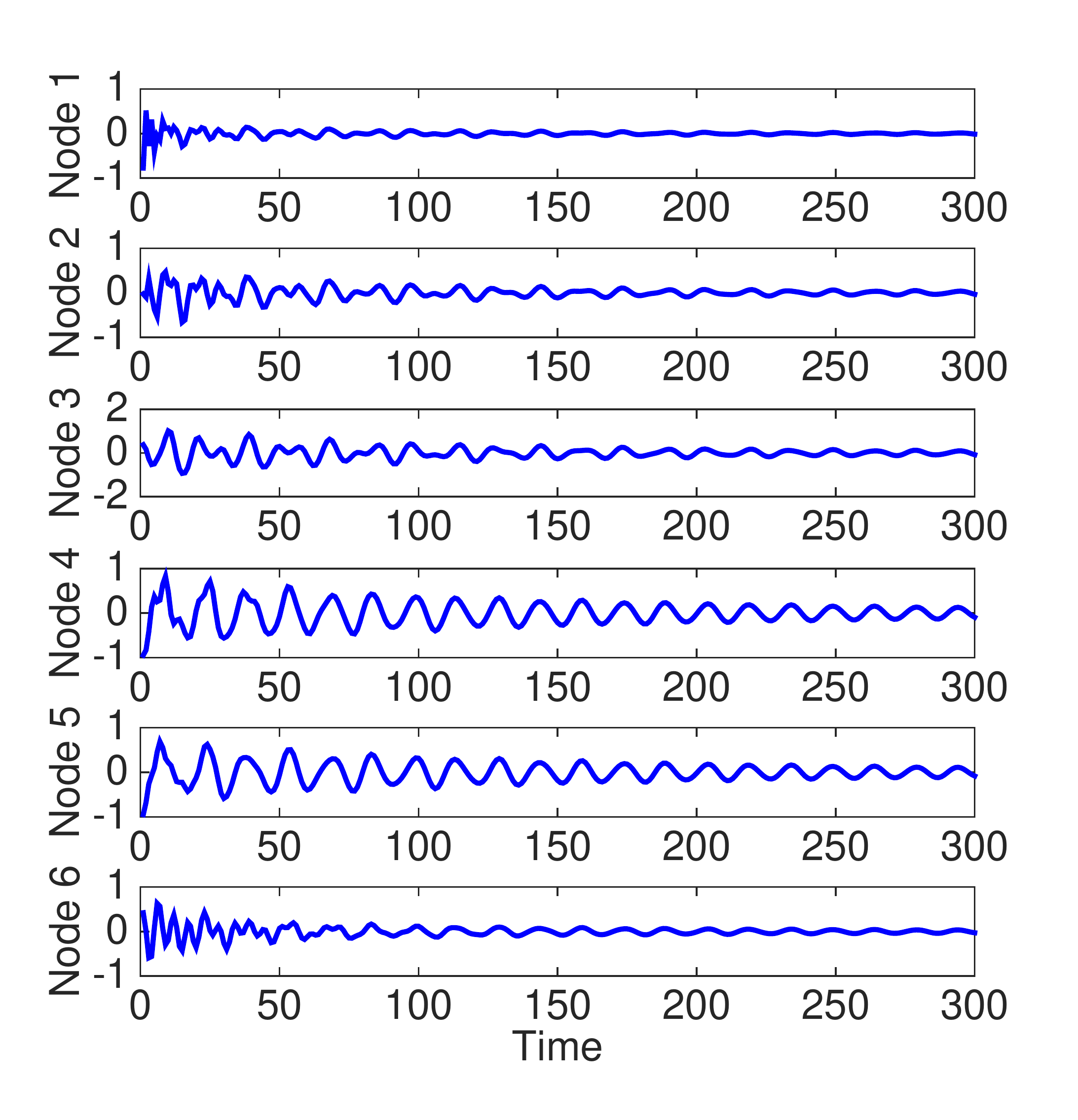}
\centerline{\footnotesize{(a)}}
\end{minipage}
\hfill
\begin{minipage}{0.49\linewidth}
\centering
\includegraphics[width=2.4in]{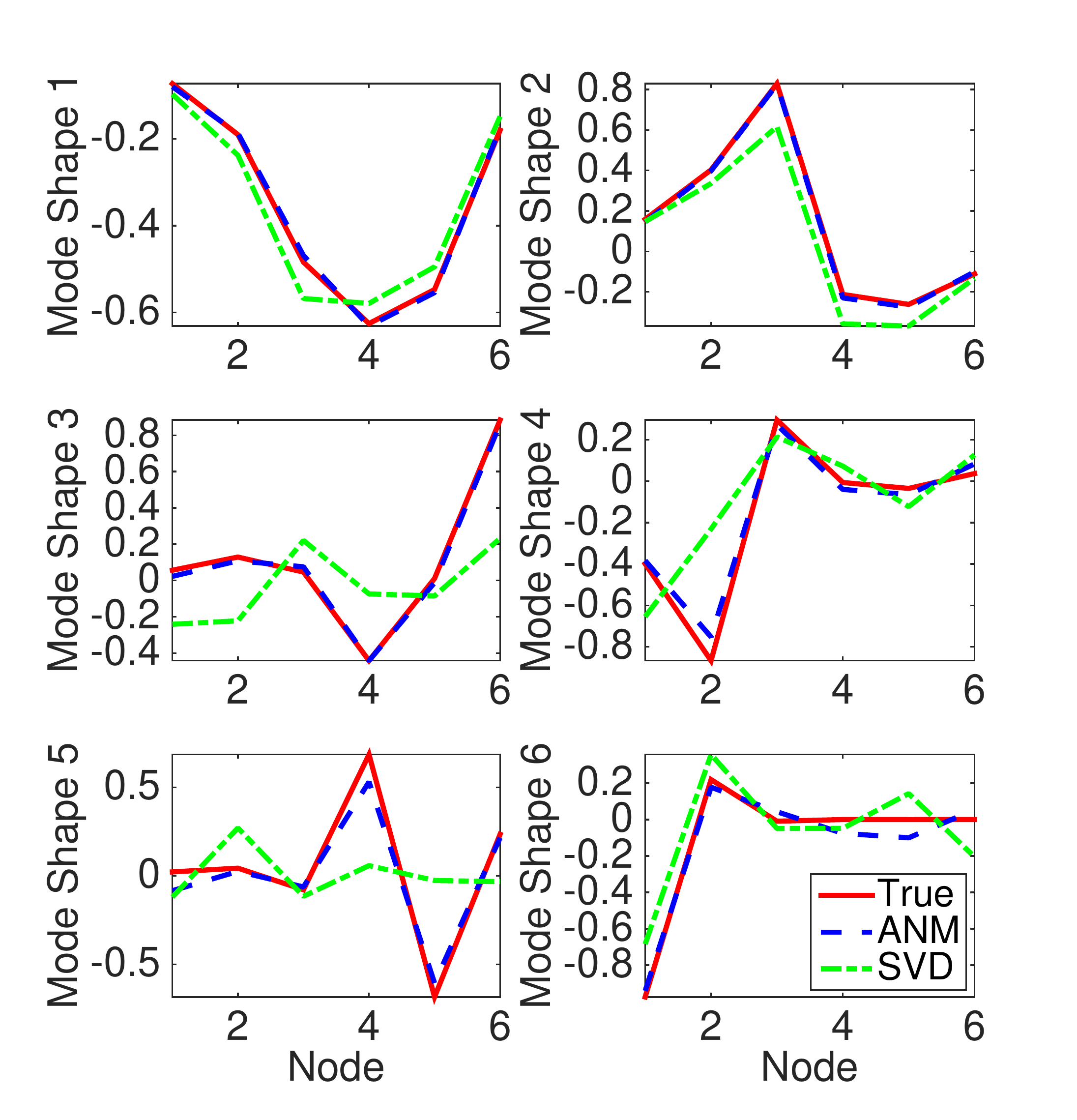}
\centerline{\footnotesize{(b)}}
\end{minipage}
\caption{Uniform sampling in the boxcar system with damping. (a) Uniform samples from each node (sensor). (b) Estimated mode shapes by ANM and SVD.}
\label{BoxcarDamping}
\end{figure}

Finally, in the noisy case, we have extended the SMV atomic norm denoising problem to an MMV atomic norm denoising problem and derived non-asymptotic theoretical bounds for the recovery error. While this analysis may be of its own independent interest, it is also used in our proof of Theorem~\ref{THMCS}.

\section{Acknowledgements}
\label{ackn}

The authors would like to thank Zhihui Zhu and Qiuwei Li at the Colorado School of Mines for many helpful discussions on the proof of Theorem \ref{THMCS}.
This work was supported by NSF grant CCF-1409258, NSF grant CCF-1464205, and NSF CAREER grant CCF-1149225. Preliminary versions of this work were presented at IEEE ICMEW 2016~\cite{li2016atomic} and IEEE ICASSP 2017~\cite{li2017atomic}.

\bibliographystyle{ieeetr}
\bibliography{ANMnotes2}

\begin{thebibliography}{10}

\bibitem{cunha2006experimental}
A.~Cunha and E.~Caetano, ``Experimental modal analysis of civil engineering
  structures,'' {\em Sound and Vibration}, vol.~40, no.~6, pp.~12--20, 2006.

\bibitem{kammer1991sensor}
D.~C. Kammer, ``Sensor placement for on-orbit modal identification and
  correlation of large space structures,'' {\em Journal of Guidance, Control,
  and Dynamics}, vol.~14, no.~2, pp.~251--259, 1991.

\bibitem{marshall1985modal}
K.~D. Marshall, ``Modal analysis of a violin,'' {\em The Journal of the
  Acoustical Society of America}, vol.~77, no.~2, pp.~695--709, 1985.

\bibitem{o2014compressed}
S.~O'Connor, J.~Lynch, and A.~Gilbert, ``Compressed sensing embedded in an
  operational wireless sensor network to achieve energy efficiency in long-term
  monitoring applications,'' {\em Smart Materials and Structures}, vol.~23,
  no.~8, p.~085014, 2014.

\bibitem{Chand12}
V.~Chandrasekaran, B.~Recht, P.~A. Parrilo, and A.~S. Willsky, ``The convex
  geometry of linear inverse problems,'' {\em Foundations of Computational
  Mathematics}, vol.~12, no.~6, pp.~805--849, 2012.

\bibitem{Tang13}
G.~Tang, B.~N. Bhaskar, P.~Shah, and B.~Recht, ``Compressed sensing off the
  grid,'' {\em IEEE Transactions on Information Theory}, vol.~59, no.~11,
  pp.~7465--7490, 2013.

\bibitem{Li15}
Y.~Li and Y.~Chi, ``Off-the-grid line spectrum denoising and estimation with
  multiple measurement vectors,'' {\em IEEE Transactions on Signal Processing},
  vol.~64, no.~5, pp.~1257--1269, 2016.

\bibitem{yang2014continuous}
Z.~Yang and L.~Xie, ``Continuous compressed sensing with a single or multiple
  measurement vectors,'' in {\em IEEE Workshop on Statistical Signal Processing
  (SSP)}, pp.~288--291, 2014.

\bibitem{heckel2016generalized}
R.~Heckel and M.~Soltanolkotabi, ``Generalized line spectral estimation via
  convex optimization,'' {\em arXiv preprint arXiv:1609.08198}, 2016.

\bibitem{Yang14}
Z.~Yang and L.~Xie, ``Exact joint sparse frequency recovery via optimization
  methods,'' {\em IEEE Transactions on Signal Processing}, vol.~64, no.~19,
  pp.~5145--5157, 2014.

\bibitem{Duarte11}
M.~F. Duarte and R.~G. Baraniuk, ``Spectral compressive sensing,'' {\em Applied
  and Computational Harmonic Analysis}, vol.~35, no.~1, pp.~111--129, 2013.

\bibitem{Chi11}
Y.~Chi, L.~L. Scharf, A.~Pezeshki, and A.~R. Calderbank, ``Sensitivity to basis
  mismatch in compressed sensing,'' {\em IEEE Transactions on Signal
  Processing}, vol.~59, no.~5, pp.~2182--2195, 2011.

\bibitem{Herman10}
M.~Herman and T.~Strohmer, ``General deviants: An analysis of perturbations in
  compressed sensing,'' {\em IEEE Journal of Selected Topics in Signal
  Processing}, vol.~4, no.~2, pp.~342--349, 2010.

\bibitem{yang2016sparse}
Z.~Yang, J.~Li, P.~Stoica, and L.~Xie, ``Sparse methods for
  direction-of-arrival estimation,'' {\em arXiv preprint arXiv:1609.09596},
  2016.

\bibitem{schmidt1982signal}
R.~Schmidt, ``A signal subspace approach to multiple emitter location and
  spectral estimation.,'' Ph.D. dissertation, Stanford University, 1981.

\bibitem{yang2016super}
D.~Yang, G.~Tang, and M.~B. Wakin, ``Super-resolution of complex exponentials
  from modulations with unknown waveforms,'' {\em IEEE Transactions on
  Information Theory}, vol.~62, no.~10, pp.~5809--5830, 2016.

\bibitem{bhaskar2013atomic}
B.~N. Bhaskar, G.~Tang, and B.~Recht, ``Atomic norm denoising with applications
  to line spectral estimation,'' {\em IEEE Transactions on Signal Processing},
  vol.~61, no.~23, pp.~5987--5999, 2013.

\bibitem{tang2015near}
G.~Tang, B.~N. Bhaskar, and B.~Recht, ``Near minimax line spectral
  estimation,'' {\em IEEE Transactions on Information Theory}, vol.~61, no.~1,
  pp.~499--512, 2015.

\bibitem{Ewins00}
D.~J. Ewins, {\em Modal testing: {T}heory and practice}, vol.~6.
\newblock Research studies press Letchworth, 1995.

\bibitem{park2014modal}
J.~Y. Park, M.~B. Wakin, and A.~C. Gilbert, ``Modal analysis with compressive
  measurements,'' {\em IEEE Transactions on Signal Processing}, vol.~62, no.~7,
  pp.~1655--1670, 2014.

\bibitem{mead1969forced}
D.~Mead and S.~Markus, ``The forced vibration of a three-layer, damped sandwich
  beam with arbitrary boundary conditions,'' {\em Journal of Sound and
  Vibration}, vol.~10, no.~2, pp.~163--175, 1969.

\bibitem{yamakoshi1990ultrasonic}
Y.~Yamakoshi, J.~Sato, and T.~Sato, ``Ultrasonic imaging of internal vibration
  of soft tissue under forced vibration,'' {\em IEEE Transactions on
  Ultrasonics, Ferroelectrics, and Frequency Control}, vol.~37, no.~2,
  pp.~45--53, 1990.

\bibitem{salawu1995bridge}
O.~S. Salawu and C.~Williams, ``Bridge assessment using forced-vibration
  testing,'' {\em Journal of Structural Engineering}, 1995.

\bibitem{cho2015forced}
D.~S. Cho, B.~H. Kim, J.-H. Kim, N.~Vladimir, and T.~M. Choi, ``Forced
  vibration analysis of arbitrarily constrained rectangular plates and
  stiffened panels using the assumed mode method,'' {\em Thin-Walled
  Structures}, vol.~90, pp.~182--190, 2015.

\bibitem{Rade}
D.~Rade and V.~Steffen~Jr, ``Structural dynamics and modal analysis,'' {\em
  EOLSS--Encyclopedia of Life Support}, 2008.

\bibitem{Yang13}
Y.~Yang and S.~Nagarajaiah, ``Output-only modal identification with limited
  sensors using sparse component analysis,'' {\em Journal of Sound and
  Vibration}, vol.~332, no.~19, pp.~4741--4765, 2013.

\bibitem{yu2014estimation}
K.~Yu, K.~Yang, and Y.~Bai, ``Estimation of modal parameters using the sparse
  component analysis based underdetermined blind source separation,'' {\em
  Mechanical Systems and Signal Processing}, vol.~45, no.~2, pp.~302--316,
  2014.

\bibitem{Sadhu13}
A.~Sadhu, B.~Hazra, and S.~Narasimhan, ``Decentralized modal identification of
  structures using parallel factor decomposition and sparse blind source
  separation,'' {\em Mechanical Systems and Signal Processing}, vol.~41, no.~1,
  pp.~396--419, 2013.

\bibitem{yang2012time}
Y.~Yang and S.~Nagarajaiah, ``Time-frequency blind source separation using
  independent component analysis for output-only modal identification of highly
  damped structures,'' {\em Journal of Structural Engineering}, vol.~139,
  no.~10, pp.~1780--1793, 2012.

\bibitem{donoho2006compressed}
D.~L. Donoho, ``Compressed sensing,'' {\em IEEE Transactions on Information
  Theory}, vol.~52, no.~4, pp.~1289--1306, 2006.

\bibitem{cande2008introduction}
E.~J. Cand{\`e}s and M.~B. Wakin, ``An introduction to compressive sampling,''
  {\em IEEE Signal Processing Magazine}, vol.~25, no.~2, pp.~21--30, 2008.

\bibitem{candes2006robust}
E.~J. Cand{\`e}s, J.~Romberg, and T.~Tao, ``Robust uncertainty principles:
  Exact signal reconstruction from highly incomplete frequency information,''
  {\em IEEE Transactions on Information Theory}, vol.~52, no.~2, pp.~489--509,
  2006.

\bibitem{feeny1998physical}
B.~Feeny and R.~Kappagantu, ``On the physical interpretation of proper
  orthogonal modes in vibrations,'' {\em Journal of Sound and Vibration},
  vol.~211, no.~4, pp.~607--616, 1998.

\bibitem{yang2017blind}
Y.~Yang, C.~Dorn, T.~Mancini, Z.~Talken, G.~Kenyon, C.~Farrar, and
  D.~Mascare{\~n}as, ``Blind identification of full-field vibration modes from
  video measurements with phase-based video motion magnification,'' {\em
  Mechanical Systems and Signal Processing}, vol.~85, pp.~567--590, 2017.

\bibitem{yang2017blind1}
Y.~Yang, C.~Dorn, T.~Mancini, Z.~Talken, S.~Nagarajaiah, G.~Kenyon, C.~Farrar,
  and D.~Mascare{\~n}as, ``Blind identification of full-field vibration modes
  of output-only structures from uniformly-sampled, possibly temporally-aliased
  (sub-nyquist), video measurements,'' {\em Journal of Sound and Vibration},
  vol.~390, pp.~232--256, 2017.

\bibitem{yang2015output}
Y.~Yang and S.~Nagarajaiah, ``Output-only modal identification by compressed
  sensing: Non-uniform low-rate random sampling,'' {\em Mechanical Systems and
  Signal Processing}, vol.~56, pp.~15--34, 2015.

\bibitem{lu2015distributed}
Z.~Lu, R.~Ying, S.~Jiang, P.~Liu, and W.~Yu, ``Distributed compressed sensing
  off the grid,'' {\em IEEE Signal Processing Letters}, vol.~22, no.~1,
  pp.~105--109, 2015.

\bibitem{boashash2015time}
B.~Boashash, {\em Time-frequency signal analysis and processing: a
  comprehensive reference}.
\newblock Academic Press, 2015.

\bibitem{smith2011spectral}
J.~O. Smith~III, {\em Spectral audio signal processing}.
\newblock W3K publishing, 2011.

\bibitem{lopez2005some}
M.~L{\'o}pez~Aenlle, R.~Brincker, and A.~C. Fern{\'a}ndez~Canteli, ``Some
  methods to determine scaled mode shapes in natural input modal analysis,'' in
  {\em Int. Modal Analysis Conference (IMAC)}, 2005.

\bibitem{grant2008cvx}
M.~Grant, S.~Boyd, and Y.~Ye, ``{CVX: M}atlab software for disciplined convex
  programming,'' 2008.

\bibitem{yang2016enhancing}
Z.~Yang and L.~Xie, ``Enhancing sparsity and resolution via reweighted atomic
  norm minimization,'' {\em IEEE Transactions on Signal Processing}, vol.~64,
  no.~4, pp.~995--1006, 2016.

\bibitem{baron2009distributed}
D.~Baron, M.~F. Duarte, M.~B. Wakin, S.~Sarvotham, and R.~G. Baraniuk,
  ``Distributed compressive sensing,'' {\em arXiv preprint arXiv:0901.3403},
  2009.

\bibitem{yang2017on}
Z.~Yang, J.~Tang, Y.~C. Eldar, and L.~Xie, ``On the sample complexity of
  multichannel frequency estimation via convex optimization,'' {\em arXiv
  preprint arXiv:1712.05674}, 2017.

\bibitem{bajwa2006compressive}
W.~Bajwa, J.~Haupt, A.~Sayeed, and R.~Nowak, ``Compressive wireless sensing,''
  in {\em Proceedings of the 5th International Conference on Information
  Processing in Sensor Networks (IPSN)}, pp.~134--142, ACM, 2006.

\bibitem{eftekhari2016stabilizing}
A.~Eftekhari, H.~L. Yap, M.~B. Wakin, and C.~J. Rozell, ``Stabilizing
  embedology: Geometry-preserving delay-coordinate maps,'' {\em arXiv preprint
  arXiv:1609.06347}, 2016.

\bibitem{boyd2004convex}
S.~Boyd and L.~Vandenberghe, {\em Convex optimization}.
\newblock Cambridge University Press, 2004.

\bibitem{gordon1988milman}
Y.~Gordon, ``On {M}ilman's inequality and random subspaces which escape through
  a mesh in {$R^n$},'' in {\em Geometric Aspects of Functional Analysis},
  pp.~84--106, Springer, 1988.

\bibitem{pisier1986probabilistic}
G.~Pisier, ``Probabilistic methods in the geometry of {B}anach spaces,'' in
  {\em Probability and Analysis}, pp.~167--241, Springer, 1986.

\bibitem{oymak2016sharp}
S.~Oymak and B.~Hassibi, ``Sharp mse bounds for proximal denoising,'' {\em
  Foundations of Computational Mathematics}, vol.~16, no.~4, pp.~965--1029,
  2016.

\bibitem{schaeffer1941inequalities}
A.~Schaeffer, ``Inequalities of {A. Markoff} and {S. Bernstein} for polynomials
  and related functions,'' {\em Bulletin of the American Mathematical Society},
  vol.~47, no.~8, pp.~565--579, 1941.

\bibitem{li2016atomic}
S.~Li, D.~Yang, and M.~B. Wakin, ``Atomic norm minimization for modal
  analysis,'' in {\em IEEE International Conference on Multimedia \& Expo
  Workshops (ICMEW)}, 2016.

\bibitem{li2017atomic}
S.~Li, D.~Yang, and M.~Wakin, ``Atomic norm minimization for modal analysis
  with random spatial compression,'' in {\em IEEE International Conference on
  Acoustics, Speech and Signal Processing (ICASSP)}, 2017.

\bibitem{heinonen2005lectures}
J.~Heinonen, {\em Lectures on Lipschitz analysis}.
\newblock Jyv\"{a}skyl\"{a} Summer School, 2005.

\bibitem{candes2014towards}
E.~J. Cand{\`e}s and C.~Fernandez-Granda, ``Towards a mathematical theory of
  super-resolution,'' {\em Communications on Pure and Applied Mathematics},
  vol.~67, no.~6, pp.~906--956, 2014.

\bibitem{candes2013super}
E.~J. Cand{\`e}s and C.~Fernandez-Granda, ``Super-resolution from noisy data,''
  {\em Journal of Fourier Analysis and Applications}, vol.~19, no.~6,
  pp.~1229--1254, 2013.

\end{thebibliography}

\newpage
\appendix

\section{Appendix}
\label{appe}

\subsection{Proof of Theorem \ref{EPhiZ}}
\label{ProofEPhiZ}

This proof is modified from \cite{gordon1988milman}. For $\Z\in \Omega \subset \SSS^{MN-1}$, $\Zt\in \SSS^{M'N-1}\triangleq\{\Zt\in \CCC^{M'\times N}: ~\|\Zt\|_F=1\}$, define two Gaussian processes
\begin{align*}
W_{\Z,\Zt}=\langle \G,\Zt \rangle+\langle \H,\Z \rangle,~\Wt_{\Z,\Zt}=\langle \bPhi\Zh, \Zt \rangle.
\end{align*}
Here, $\G$ and $\H$ are random matrices with i.i.d.\ entries from the distribution $\NN(0,1)$. For all $\Z,~\Z'\in\Omega$, $\Zt,~\Zt'\in\SSS^{M'N-1}$, it can be shown that
\begin{align*}
&\EEE |W_{\Z,\Zt}-W_{\Z',\Zt'}|^2-\EEE|\Wt_{\Z,\Zt}-\Wt_{\Z',\Zt'}|^2\\
=&4-\langle \Zt',\Zt \rangle-\langle \Z',\Z \rangle-\langle \Zt,\Zt' \rangle-\langle \Z,\Z' \rangle-\sum_{n=1}^N\langle \z_n,\z_n \rangle \langle \zt_n,\zt_n \rangle-\sum_{n=1}^N\langle \z'_n,\z'_n \rangle \langle \zt'_n,\zt'_n \rangle \\
&+\sum_{n=1}^N\langle \z_n,\z'_n \rangle \langle \zt'_n,\zt_n \rangle+\sum_{n=1}^N\langle \z'_n,\z_n \rangle \langle \zt_n,\zt'_n \rangle\\
=& 4-2 \change{Re}(\langle \Zt',\Zt \rangle)-2 \change{Re}(\langle \Z,\Z' \rangle) -\sum_{n=1}^N \left( \|\z_n\|^2 \|\zt_n\|^2+  \|\z'_n\|^2 \|\zt'_n\|^2 -2 \change{Re}(\langle \z_n,\z'_n \rangle \langle \zt'_n,\zt_n \rangle)  \right)      \\
\geq &  \change{Re}\left(4-2 \langle \Zt',\Zt \rangle-2 \langle \Z,\Z' \rangle -\sum_{n=1}^N \sum_{m=1}^N \left( \|\z_n\|^2 \|\zt_m\|^2+  \|\z'_n\|^2 \|\zt'_m\|^2 -2\langle \z_n,\z'_n \rangle \langle \zt'_m,\zt_m \rangle  \right)   \right)    \\
=& \change{Re}\left(2-2 \langle \Zt',\Zt \rangle-2 \langle \Z,\Z' \rangle +2\langle \Z,\Z' \rangle \langle \Zt',\Zt \rangle    \right)\\
=&2 \change{Re}\left( (1- \langle \Zt',\Zt \rangle)(1- \langle \Z,\Z' \rangle) \right)\\
\geq & 0.
\end{align*}
The first inequality holds since
\begin{align*}
&\|\z_n\|^2 \|\zt_m\|^2+  \|\z'_n\|^2 \|\zt'_m\|^2 -2\langle \z_n,\z'_n \rangle \langle \zt'_m,\zt_m \rangle \\
\geq &\|\z_n\|^2 \|\zt_m\|^2+  \|\z'_n\|^2 \|\zt'_m\|^2 -2\|\z_n\| \|\zt_m\| \|\z'_n\| \|\zt'_m\| \geq 0.
\end{align*}
Then, by \cite{gordon1988milman} we have
\begin{align*}
\EEE \left[ \min_{\Z\in\Omega} \max_{\Zt\in\SSS^{M'N-1}} W_{\Z,\Zt} \right] \leq \EEE \left[ \min_{\Z\in\Omega} \max_{\Zt\in\SSS^{M'N-1}} \Wt_{\Z,\Zt} \right],
\end{align*}
which is equivalent to
\begin{align}
\EEE \left[ \min_{\Z\in\Omega} \max_{\Zt\in\SSS^{M'N-1}} \langle \G,\Zt \rangle+\langle \H,\Z \rangle \right] \leq \EEE \left[ \min_{\Z\in\Omega} \max_{\Zt\in\SSS^{M'N-1}} \langle \bPhi\Zh, \Zt \rangle \right]. \label{EGH}
\end{align}
Since $\| \Zt \|_F=1$, using Cauchy-Schwarz inequality, we have
\begin{align}
\langle \bPhi\Zh, \Zt \rangle \leq \|\bPhi \Zh\|_F \label{CS}.
\end{align}
Plugging (\ref{CS}) into (\ref{EGH}) gives
\begin{align*}
\EEE \left[ \min_{\Z \in \Omega} \|\bPhi \Zh\|_F \right]& \geq \EEE\left[ \min_{\Z\in\Omega} \max_{\Zt\in\SSS^{M'N-1}} \langle \bPhi\Zh, \Zt \rangle \right]\\
&\geq \EEE \left[ \min_{\Z\in\Omega} \max_{\Zt\in\SSS^{M'N-1}} \langle \G,\Zt \rangle+\langle \H,\Z \rangle \right]\\
&=\EEE_{\G}\|\G\|_F+\EEE_{\H}\left[\min_{\Z\in \Omega}\langle \H,\Z \rangle \right]\\
&\geq \EEE_{\G}\|\G\|_F-\EEE_{\H}\left[\max_{\Z\in \Omega}\langle \H,\Z \rangle \right]\\
&=\lambda_{M'N}-\omega(\Omega).
\end{align*}
Moreover, we have
\begin{align*}
\lambda_{M'N}&=\EEE_{\G}\|\G\|_F=\EEE\sqrt{\sum_{n=1}^N\|\g_n\|_2^2}\\
&\geq \EEE\sqrt{N \min_{n\in[1,N]} \|\g_n\|_2^2}=\sqrt{N} \EEE\sqrt{ \min_{n\in[1,N]} \|\g_n\|_2^2}\\
&=\sqrt{N} \EEE\left[ \min_{n\in[1,N]} \|\g_n\|_2 \right]\\
&\geq \frac{\sqrt{N}M'}{\sqrt{M'+1}},
\end{align*}
where the last inequality uses the result that the expected length of an $M'$-dimensional Gaussian random vector is lower bounded by $\frac{M'}{\sqrt{M'+1}}$ \cite{Chand12}, \cite{gordon1988milman}. \qed

\subsection{Proof of Corollary \ref{MpGauss}}
\label{ProofMpGauss}

First, we will show that the following function
\begin{align*}
f:\bPhi \rightarrow \min_{\Z\in \Omega} \|\bPhi \Zh\|_F
\end{align*}
is Lipschitz with respect to the Frobenius norm with constant 1.

Define a function $f_{\Z}(\bPhi)=\|\bPhi\Zh\|_F$ with $\Z=[\z_1,\z_2,\cdots,\z_N]$ and $\Zh=\change{diag}(\Z)$. The gradient of $f_{\Z}(\bPhi)$ is
\begin{align*}
\nabla_{\bPhi}f_{\Z}(\bPhi)=\frac{\bPhi \Zh\Zh^*}{\|\bPhi\Zh\|_F}.
\end{align*}
Thus, using the mean value theorem, we can get
\begin{align*}
|f_{\Z}(\bPhi_1)-f_{\Z}(\bPhi_2)|&=|\langle \nabla_{\bPhi}f_{\Z}(\bPhi),\bPhi_1-\bPhi_2 \rangle| \leq \|\nabla_{\bPhi}f_{\Z}(\bPhi)\|_F \|\bPhi_1-\bPhi_2\|_F\\
&= \frac{1}{\|\bPhi\Zh\|_F}\|\bPhi \Zh\Zh^*\|_F \|\bPhi_1-\bPhi_2\|_F\\
&\leq \|\bPhi_1-\bPhi_2\|_F.
\end{align*}
So, the function $f_{\Z}(\bPhi)=\|\bPhi\Zh\|_F$ is Lipschitz with respect to the Frobenius norm with constant 1.
According to Lemma 2.1 in \cite{heinonen2005lectures}, we can conclude that $f:\bPhi \rightarrow \min_{\Z\in \Omega} \|\bPhi \Zh\|_F$ is also Lipschitz with respect to the Frobenius norm with constant 1.

The Gaussian concentration inequality for Lipschitz functions provided in \cite{pisier1986probabilistic} implies that
\begin{align*}
\PPP\left\{ \min_{\Z\in \Omega} \|\bPhi \Zh\|_F\geq \EEE\left[\min_{\Z\in \Omega} \|\bPhi \Zh\|_F \right]-t \right\}\geq 1-e^{-\frac{t^2}{2}}
\end{align*}
holds for any $t>0$. Using Theorem \ref{EPhiZ}, we can get
\begin{align*}
\PPP\left\{ \min_{\Z\in \Omega} \|\bPhi \Zh\|_F\geq \lambda_{M'N}-\omega(\Omega)-t \right\}\geq 1-e^{-\frac{t^2}{2}}.
\end{align*}

Thus, we can choose $M'$ such that
\begin{align*}
\lambda_{M'N}-\omega(\Omega)-t>0,
\end{align*}
which is equivalent to
\begin{align*}
\lambda_{M'N}^2>(\omega(\Omega)+t)^2\triangleq \epsilon^2.
\end{align*}
So, we can let
\begin{align*}
\frac{NM'^2}{M'+1}>\epsilon^2,
\end{align*}
which means that we need
\begin{align}
M'>\frac{\epsilon^2+\sqrt{\epsilon^4+4N\epsilon^2}}{2N}. \label{Mpxi}
\end{align}
Thus, we can choose
\begin{align*}
M'\geq\frac{\epsilon^2+\sqrt{\epsilon^4+4N\epsilon^2+4N^2}}{2N}=\frac{1}{N}\epsilon^2+1=\frac{1}{N}(\omega(\Omega)+t)^2+1
\end{align*}
to satisfy the above condition (\ref{Mpxi}).
Moreover, if $M' \geq \frac{4}{N}\omega^2(\Omega)$, we have
\begin{align*}
M'&=\frac{M'}{2}+\frac{M'-2}{2}+1\geq \frac{2}{N}\omega^2(\Omega)+2\left( \frac{\sqrt{M'-2}}{2} \right)^2+1\\
&\geq \left( \frac{\omega(\Omega)}{\sqrt{N}}+\frac{\sqrt{N(M'-2)}}{2\sqrt{N}} \right)^2+1\\
&=\frac{1}{N} \left[ \omega(\Omega)+\frac{\sqrt{N(M'-2)}}{2} \right]^2+1.
\end{align*}
Taking $t=\frac{\sqrt{N(M'-2)}}{2}$, the proof for Corollary \ref{MpGauss} is finished. \qed

\subsection{Proof of key lemmas in Section \ref{XhmXs}}
\label{lemmas}

To prove these lemmas, we need the following theorem, which is an extension of Theorem 4 in \cite{tang2015near}.
\begin{Theorem}
\label{dualstability}
Let $\BB=\{\b\in \CCC^N: \| \b \|_2=1\}$ be a set of vectors with unit $\ell_2$ norm.
For any $f_1, f_2, \ldots, f_K$ satisfying the minimum separation condition (\ref{mise}) , there exists a vector-valued trigonometric polynomial $\QQ(f)=\Q^*\a(f)$ satisfying the following properties for some $\Q\in \CCC^{M\times N}$.
\begin{enumerate}
\item For each $k=1,\ldots,K$, $\QQ(f_k)=\b_k$ with $\b_k \in \BB$.
\item In each near region $N_k= \{ f:d(f,f_k)<0.16/M  \}$, there exist constants $C_a$ and $C'_a$ such that
\begin{align}
\| \QQ(f) \|_2&\leq 1-\frac{C_a}{2}M^2(f-f_k)^2 \label{SN1}\\
\left\| \b_k-\QQ(f) \right\|_2 &\leq \frac{C'_a}{2}M^2(f-f_k)^2. \label{SN2}
\end{align}
\item In the far region, i.e., $f \in F=[0,1)/\cap_{k=1}^K N_j$, there exists a constant $C_b>0$ such that
\begin{align}
\|\QQ(f)\|_2\leq 1-C_b. \label{SF}
\end{align}
\end{enumerate}
\end{Theorem}

\begin{proof} Proposition 5.1 and inequality (47) in \cite{Yang14} guarantee the first and third statements. The second statement can be proved by following the proof strategies in Section 2.4 of~\cite{candes2014towards} and Appendix A of~\cite{candes2013super}.
\end{proof}

In the remainder of this section, we will prove the three key lemmas presented in Section \ref{XhmXs}.

\subsubsection{Proof of Lemma \ref{UBE}}
\label{proofUBE}

We define $\FF\triangleq\{f_1,f_2,\ldots,f_K\}$ as a set containing the true frequencies. Recall that the near region $N_k$ corresponding to $f_k$ and the far region $F$ are defined in (\ref{NFregion}). Using the Cauchy-Schwarz inequality, we have
\begin{align}
\left| \int_F \bxi^*(f)  \bnu(f) df\right|\leq \int_F |\bxi^*(f)  \bnu(f)| df \leq   \int_F \|\bxi(f)\|_2  \|\bnu(f)\|_2 df \leq \|\bxi(f)\|_{2,\infty}   \int_F  \|\bnu(f)\|_2 df, \label{Eb1}
\end{align}
where $\|\bxi(f)\|_{2,\infty}$ is defined in (\ref{xi}).

Let $\u$ be any vector with unit $\ell_2$ norm. Define
\begin{align*}
\gamma (f)\triangleq \u^*\bxi(f),
\end{align*}
which is a trigonometric polynomial with degree $M$. According to Bernstein's inequality for polynomials \cite{schaeffer1941inequalities}, we have
\begin{align}
\sup_{f} |\gamma'(f)| &\leq M\sup_{f} |\gamma(f)|,\label{Bern1}\\
\sup_{f} |\gamma''(f)| &\leq M^2\sup_{f} |\gamma(f)|. \label{Bern2}
\end{align}
Note that
\begin{align*}
\sup_{f} \| \bxi'(f) \|_2&=\sup_{f,\u} |\u^*\bxi'(f)|=\sup_{f,\u} |\gamma'(f)|\\
&\leq M\sup_{f,\u} |\gamma(f)|=M\sup_{f,\u} |\u^*\bxi(f)|=M\sup_{f} \| \bxi(f) \|_2,
\end{align*}
which implies
\begin{align*}
\|\bxi'(f)\|_{2,\infty} \leq  M\|\bxi(f)\|_{2,\infty}.
\end{align*}
Using a similar argument, we can get
\begin{align*}
\|\bxi''(f)\|_{\infty}&\leq M^2\|\bxi(f)\|_{2,\infty}.
\end{align*}

Note that
the Taylor expansion of $\gamma(f)$ at $f_k$ is
\begin{align*}
\gamma(f)=\gamma(f_k)+(f-f_k)\gamma'(f_k)+\frac{1}{2}(f-f_k)^2\gamma''(\tilde{f}_k)
\end{align*}
for some $\tilde{f}_k\in N_k$. It follows that
\begin{align*}
&\sup_{\u}|\gamma(f)-\gamma(f_k)-(f-f_k)\gamma'(f_k)|\\
=&\frac{1}{2}(f-f_k)^2 \sup_{\u}  |\gamma''(\tilde{f}_k)|\\
=&\frac{1}{2}(f-f_k)^2 \sup_{\u}  |\u^*\bxi''(\tilde{f}_k)|\\
\leq & \frac{1}{2}(f-f_k)^2 \|\bxi''(f)  \|_{2,\infty}.
\end{align*}
Plugging in Bernstein's inequality (\ref{Bern2}), we have
\begin{align*}
\sup_{\u}|\gamma(f)-\gamma(f_k)-(f-f_k)\gamma'(f_k)|\leq  \frac{M^2}{2}(f-f_k)^2 \|\bxi(f)  \|_{2,\infty}.
\end{align*}
Note that the left hand side of the above inequality is equal to
\begin{align*}
&\sup_{\u} |\langle\bxi(f)-\bxi(f_k)-(f-f_k)\bxi'(f_k),\u  \rangle|\\
=&\|\bxi(f)-\bxi(f_k)-(f-f_k)\bxi'(f_k)\|_2.
\end{align*}

Denoting $\r(f)=\bxi(f)-\bxi(f_k)-(f-f_k)\bxi'(f_k)$, we have
\begin{align*}
\| \r(f) \|_2 \leq \frac{M^2}{2}(f-f_k)^2 \|\bxi(f)  \|_{2,\infty}.
\end{align*}

Finally, we can obtain
\begin{equation}
\begin{aligned}
&\left| \int_{N_k} \bxi^*(f)  \bnu(f) df\right|\\
\leq & \|\bxi(f_k)\|_2 \left\| \int_{N_k} \bnu (f) df \right\|_2 +  \|\bxi'(f_k)\|_2 \left\| \int_{N_k} (f-f_k)\bnu (f) df \right\|_2 +  \int_{N_k} \|\r(f)\|_2 \|\bnu (f) \|_2 df \\
\leq & \|\bxi(f)\|_{2,\infty} \left( \left\| \int_{N_k} \bnu (f) df \right\|_2+M\left\| \int_{N_k} (f-f_k)\bnu (f) df \right\|_2 +\frac{M^2}{2} \int_{N_k}(f-f_k)^2 \|\bnu_n(f) \|df \right). \label{Eb2}
\end{aligned}
\end{equation}
We have plugged in $\bxi(f)=\bxi(f_k)+(f-f_k)\bxi'(f_k)+\r(f)$ and used the triangle inequality to get the first inequality.
Substituting (\ref{Eb1}) and (\ref{Eb2}) into (\ref{Eb0}), we can obtain (\ref{EI}). \qed

\subsubsection{Proof of Lemma \ref{I0I1}}
\label{proofI0I1}

Let $\QQ(f)=\Q^*\a(f)$ be the dual polynomial as in \cite{Yang14}. Then, we have
\begin{align}
|\langle \bxi(f),\QQ(f) \rangle|\leq \int_0^1 |\QQ^*(f)\bxi(f)|df\leq \int_0^1 \|\QQ(f)\|_2 \|\bxi(f)\|_2df\leq \|\bxi(f)\|_{2,\infty} \|\QQ(f)\|_1, \label{xiQ}
\end{align}
where $\|\QQ(f)\|_1\triangleq \int_0^1 \|\QQ(f)\|_2 df$. The first and second inequalities follow from the triangle inequality and the Cauchy-Schwarz inequality, respectively. In this work, we
use the dual polynomial $\QQ(f)$ constructed in \cite{Yang14} of the form
\begin{align*}
\QQ(f)=\sum_{f_k\in\FF} \balpha_k \KK(f-f_k)+\sum_{f_k\in\FF} \bbeta_k \KK'(f-f_k),
\end{align*}
where $\KK(f)$ is the squared Fej\'{e}r kernel. All the $\balpha_k$ and $\bbeta_k$ are $N\times 1$ vectors. Define $\|\X\|_{2,\infty}=\max_j\|\X_j\|_2$, with $\X_j$ being the $j$th column of $\X$. Denote $\balpha=[\balpha_1,\balpha_2,\cdots,\balpha_K]$ and $\bbeta=[\bbeta_1,\bbeta_2,\cdots,\bbeta_K]$. It is shown in \cite{Yang14} that
\begin{align*}
\|\balpha\|_{2,\infty}\leq C_{\alpha},~~\|\bbeta\|_{2,\infty}\leq \frac{C_{\beta}}{M}
\end{align*}
for some numerical constants $C_{\alpha}$ and $C_{\beta}$.
Then, we have
\begin{equation}
\begin{aligned}
\|\QQ(f)\|_1&\triangleq \int_0^1 \|\QQ(f)\|_2 df\\
&\leq \int_0^1 \sum_{f_k\in \FF} \|\balpha_k\|_2|\KK(f-f_k)|df+\int_0^1 \sum_{f_k\in \FF} \|\bbeta_k\|_2|\KK'(f-f_k)|df\\
&\leq KC_{\alpha}\int_0^1|\KK(f)|df+\frac{C_{\beta}K}{M}\int_0^1|\KK'(f)|df\\
&\leq \frac{CK}{M}.\label{Qf1}
\end{aligned}
\end{equation}
For the last inequality, we have used the results $\int_0^1|\KK(f)|df\leq\frac{C}{M}$ and $\int_0^1 |\KK'(f)|df\leq C$, which can be found in Appendix C of~\cite{tang2015near}.

Now, we can bound $I_0$ as
\begin{align*}
I_0&=\sum_{k=1}^K \left\| \int_{N_k} \bnu(f) df \right\|_2=\sum_{k=1}^K \int_{N_k} \bnu^*(f) df \frac{\int_{N_k} \bnu(\fh) d\fh}{\left\| \int_{N_k} \bnu(\fh) d\fh \right\|_2}  \\
&=\sum_{k=1}^K \int_{N_k} \bnu^*(f) \frac{\int_{N_k} \bnu(\fh) d\fh}{\left\| \int_{N_k} \bnu(\fh) d\fh \right\|_2} df   \\
&=\sum_{k=1}^K \int_{N_k} \bnu^*(f) \QQ(f) df+\sum_{k=1}^K \int_{N_k} \bnu^*(f)\left[ \frac{\int_{N_k} \bnu(\fh) d\fh}{\left\| \int_{N_k} \bnu(\fh) d\fh \right\|_2}-\QQ(f)\right] df   \\
&\leq  \left| \int_0^1 \bnu^*(f)\QQ(f) df  \right|+\left| \int_F \bnu^*(f)\QQ(f) df  \right|+\sum_{k=1}^K \int_{N_k} \left| \bnu^*(f)\left[ \frac{\int_{N_k} \bnu(\fh) d\fh}{\left\| \int_{N_k} \bnu(\fh) d\fh \right\|_2}-\QQ(f)\right] \right|df  \\
&\leq  \left| \int_0^1 \bnu^*(f)\QQ(f) df  \right|+ \int_F \|\bnu(f)\|_2 df  +C'_aI_2.
\end{align*}
We have used the triangle inequality and the Cauchy-Schwarz inequality in the last two inequalities. We also use the fact that $\|\QQ(f)\|_2\leq 1$ and set $\b_k=\frac{\int_{N_k} \bnu(\fh) d\fh}{\left\| \int_{N_k} \bnu(\fh) d\fh \right\|_2}$, then use the boundary condition presented in (\ref{SN2}) to get the last inequality.
Note that the first term can be upper bounded with
\begin{align*}
&\left| \int_0^1 \bnu^*(f)\QQ(f) df  \right| = \left| \int_0^1 \bnu^*(f)\Q^*\a(f) df  \right|= \left| \int_0^1 \langle \a(f)\bnu^*(f),\Q \rangle df  \right|\\
=& \left| \left\langle \int_0^1 \a(f)\bnu^*(f)df,\Q   \right\rangle \right|=|\langle \E,\Q \rangle|=|\langle \bxi(f),\QQ(f) \rangle|\\
\leq & \frac{CK\lambda}{M},
\end{align*}
where the last equality follows from Parseval's theorem and the last inequality is a direct result from inequalities (\ref{xiQ}) and (\ref{Qf1}).
As a consequence, we have that
\begin{align*}
I_0\leq C_0\left( \frac{K\lambda}{M} +I_2+\int_F \|\bnu(f)\|_2 df \right)
\end{align*}
holds for some numerical constant $C_0$. Similarly, we can bound $I_1$ in such a way. \qed

\subsubsection{Proof of Lemma \ref{FI2}}
\label{proofFI2}

Denote $\PP_{\FF}(\bnu)$ as the projection of the difference measure $\bnu(f)$ on the support set $\FF=\{f_1,f_2,\ldots,f_K\}$. Set $\QQ(f)$ as the dual polynomial in Theorem \ref{dualstability}. To avoid confusion with the traditional definition of total variation (TV) norm, we use $\|\cdot\|_{2,TV}$ in this work to denote an extension of traditional TV norm. Then, we have
\begin{align*}
\|\PP_{\FF}(\bnu) \|_{2,TV}&=\int_0^1 \PP_{\FF}(\bnu^*) \QQ(f) df=\int_{\FF} \bnu^*(f) \QQ(f) df\\
&\leq \left| \int_0^1 \bnu^*(f) \QQ(f) df  \right| + \left|\int_{\FF^c} \bnu^*(f) \QQ(f) df \right|\\
&\leq \frac{CK\lambda}{M}+\sum_{f_k\in \FF} \left| \int_{N_k/\{f_k\}} \bnu^*(f) \QQ(f) df \right| + \left| \int_F \bnu^*(f) \QQ(f) df \right|,
\end{align*}
where $\FF^c$ is the complement set of $\FF$ on $[0,1)$. Using (\ref{SF}) and (\ref{SN1}), the integration over the far region $F$ can be bounded by
\begin{align*}
\left| \int_F \bnu^*(f) \QQ(f) df \right| \leq (1-C_b) \int_F \| \bnu(f) \|_2df,
\end{align*}
and the integration over $N_k/\{f_k\}$ can be bounded by
\begin{align*}
\left| \int_{N_k/\{f_k\}} \bnu^*(f) \QQ(f) df \right| &\leq \int_{N_k/\{f_k\}} \|\bnu(f)\|_2 \|\QQ(f)\|_2 df \\
&\leq  \int_{N_k/\{f_k\}} \left(1-\frac{1}{2}M^2C_a(f-f_k)^2  \right) \|\bnu(f)\|_2  df \\
&\leq  \int_{N_k/\{f_k\}}  \|\bnu(f)\|_2  df -C_aI_2^k.
\end{align*}
Therefore, we can get
\begin{align*}
\|\PP_{\FF}(\bnu) \|_{2,TV}&\leq \frac{CK\lambda}{M}+\sum_{f_k\in \FF}\int_{N_k/\{f_k\}}  \|\bnu(f)\|_2  df -C_aI_2+(1-C_b) \int_F \| \bnu(f) \|_2df\\
&=\frac{CK\lambda}{M}+\|\PP_{\FF^c}(\bnu) \|_{2,TV}-C_aI_2-C_b \int_F \| \bnu(f) \|_2df,
\end{align*}
which implies
\begin{align}
\|\PP_{\FF^c}(\bnu) \|_{2,TV}-\|\PP_{\FF}(\bnu) \|_{2,TV}\geq C_aI_2+C_b \int_F \| \bnu(f) \|_2df-\frac{CK\lambda}{M}.
\label{FFCL}
\end{align}

Since $\Xh$ is the solution of the MMV atomic norm denoising problem (\ref{AND}), we have
\begin{align*}
\frac{1}{2}\| \Y-\Xh \|_F^2+\lambda\| \Xh \|_{\AA}\leq \frac{1}{2}\| \Y-\X^\star \|_F^2+\lambda\| \X^\star \|_{\AA}.
\end{align*}
As a consequence, we obtain
\begin{align*}
\lambda\|\Xh\|_{\AA}\leq \lambda\|\X^{\star}\|_{\AA}+\langle \W,\Xh-\X^{\star}\rangle
\end{align*}
by elementary calculations.
It follows that
\begin{align*}
\|\hat{\bmu}\|_{2,TV}\leq \|\bmu\|_{2,TV}+\frac{|\langle \W,\E \rangle|}{\lambda}.
\end{align*}
Then, similar to Lemma \ref{UBE}, we have
\begin{align*}
|\langle \W,\E \rangle|&=\left|\left\langle \W,\int_0^1 \a(f)\bnu^*(f)df \right\rangle\right|=\left| \int_0^1 \left\langle \W,\a(f)\bnu^*(f) \right\rangle df\right|\\
&=\left| \int_0^1 \a^*(f)\W\bnu(f) df\right|\\
&\leq \|\W^* \a(f)\|_{2,\infty} \left( \frac{CK\lambda}{M} +I_0+I_1+I_2   \right)\\
&\leq C\frac{\lambda}{\eta} \left( \frac{K\lambda}{M} +I_2 +\int_F \| \bnu(f) \|_2df  \right).
\end{align*}
Here, the last inequality follows from $ \|\W^* \a(f)\|_{2,\infty} = \|\W\|_{\AA}^* \leq \frac{\lambda}{\eta}$ and Lemma \ref{I0I1}.  As a consequence, we have
\begin{align*}
\|\bmu\|_{2,TV}&+C\eta^{-1}\left( \frac{K\lambda}{M} +I_2 +\int_F \| \bnu(f) \|_2df  \right)\geq \|\hat{\bmu}\|_{2,TV}\\
&\geq \|\bmu+\bnu\|_{2,TV}\geq \|\bmu\|_{2,TV}-\|\PP_{\FF}(\bnu) \|_{2,TV}+\|\PP_{\FF^c}(\bnu) \|_{2,TV},
\end{align*}
which implies
\begin{align}
\|\PP_{\FF^c}(\bnu) \|_{2,TV}-\|\PP_{\FF}(\bnu) \|_{2,TV}\leq C\eta^{-1}\left( \frac{K\lambda}{M} +I_2 +\int_F \| \bnu(f) \|_2df  \right).
\label{FFCU}
\end{align}
Combining (\ref{FFCL}) and (\ref{FFCU}), we get
\begin{align}
(C_a-C\eta^{-1})I_2+(C_b-C\eta^{-1})\int_F \| \bnu(f) \|_2df \leq (1+\eta^{-1}) \frac{CK\lambda}{M}.
\label{ConstCombine}
\end{align}
Thus, Lemma \ref{FI2} is proved under the assumption that $\eta$ is large enough with respect to the constants appearing in~\eqref{ConstCombine}.
\qed

\end{document}